\def \WithSectionEight {}
\documentclass[iicol]{sn-jnl}%

\usepackage[utf8]{inputenc} %
\usepackage[T1]{fontenc}
\usepackage{amssymb}
\usepackage{amsmath}
\usepackage{stmaryrd}
\usepackage{mathtools}
\usepackage{enumerate}
\usepackage{color}
\usepackage{lmodern}
\usepackage{multicol}
\usepackage{subcaption}
\usepackage{wrapfig}

\usepackage{thmtools}
\usepackage{thm-restate}

\usepackage{amsthm}

\theoremstyle{plain}
\newtheorem{lemma}{Lemma}

\newtheorem{theorem}{Theorem}
\newtheorem{corollary}{Corollary}

\theoremstyle{definition}
\newtheorem{definition}{Definition}
\newtheorem{example}{Example}
\newtheorem{assumption}{Assumption}

\theoremstyle{remark}
\newtheorem{remark}{Remark}

\renewenvironment{definition}
{\begin{olddefinition}}
{\hfill\ensuremath{\square}\end{olddefinition}}

\renewenvironment{example}
{\begin{oldexample}}
{\hfill\ensuremath{\square}\end{oldexample}}

\renewenvironment{remark}
{\begin{oldremark}}
{\hfill\ensuremath{\square}\end{oldremark}}
\usepackage{paralist} %

\newenvironment{ienumerate}
	{\begin{enumerate}}
	{\end{enumerate}}

\newenvironment{oneenumerate}
	{\begin{enumerate}}
	{\end{enumerate}}

\ifdefined \WithSectionEight
	\newcommand{\sectionEight}[1]{#1}
\else
	\newcommand{\sectionEight}[1]{}
\fi

\allowdisplaybreaks

\usepackage[svgnames,table]{xcolor}
\definecolor{USPNsable}{HTML}{ad947e}

\usepackage[capitalise,english,nameinlink]{cleveref} %
\crefname{line}{\text{line}}{\text{lines}} %

\usepackage{tikz}
\usetikzlibrary{arrows,automata,positioning,shapes} %

\tikzstyle{pta}=[auto, >=stealth'] %
\tikzstyle{every node}=[initial text=]
\tikzstyle{location}=[rounded rectangle, minimum size=12pt, draw=black, fill=blue!10, inner sep=2pt]
\tikzstyle{invariant}=[draw=black, dotted, inner sep=1pt, node distance=0] %
\tikzstyle{final}=[double, fill=blue!50]
\tikzstyle{locationdup}=[location, fill=green!10, inner sep=2pt]
\tikzstyle{edgedup}=[very thick] %

\tikzstyle{beliefautomaton}=[scale= .7, every node/.style={scale=0.8}, >=stealth'] %
\tikzstyle{rct}=[thick, draw=purple!75!black, fill=purple!20!white]

\tikzstyle{private}=[fill=yellow,thick]

\definecolor{coloract}{rgb}{0.50, 0.70, 0.30}
\definecolor{colorclock}{rgb}{0.4, 0.4, 1}
\definecolor{colordisc}{rgb}{1, 0, 1}
\definecolor{colorloc}{rgb}{0.4, 0.4, 0.65}
\definecolor{colorparam}{rgb}{1, 0.6, 0.0}

\newcommand{\styleact}[1]{\ensuremath{\textcolor{coloract}{{\mathit{#1}}}}}
\newcommand{\styleclock}[1]{\ensuremath{\textcolor{colorclock}{{#1}}}}

\newcommand{\textstyleact}[1]{\ensuremath{\mathit{#1}}}
\newcommand{\textstyleclock}[1]{\ensuremath{\mathit{#1}}}

\newcommand{\textstyleloc}[1]{#1}
\newcommand{\textstyleparam}[1]{\ensuremath{{#1}}}
\newcommand{\textstyleedge}[1]{\ensuremath{\mathsf{e}_{#1}}}

\usepackage{xspace}
\newcommand{\eg}{e.g.,\xspace}

\newcommand{\ie}{i.e.,\xspace}
\newcommand{\st}{s.t.\xspace}

\newcommand{\defProblem}[3]
{%

	\smallskip

	\noindent\fcolorbox{black}{USPNsable!20}{
	\begin{minipage}{.95\columnwidth}
		\textbf{#1:}\\
		\textsc{Input}: #2\\
		\textsc{Problem}: #3
	\end{minipage}
}

	\smallskip

}

\newcommand{\game}{\mathcal{G}}
\newcommand{\gameStateSet}{Q}

\newcommand{\gameInitState}{q_0}
\newcommand{\gameTransitions}{\delta^{\game}}

\newcommand{\gameEvent}{\Sigma}
\newcommand{\gameGood}{G}

\newcommand{\Finitelyvarying}{Finitely-varying}
\newcommand{\finitelyvarying}{finitely-varying}
\newcommand{\NonFinitelyvarying}{Non-finitely-varying}
\newcommand{\nonFinitelyvarying}{non-finitely-varying}

\newcommand{\badBelief}{leaking}
\newcommand{\BadBelief}{Leaking}

\newcommand{\activated}{\ensuremath{\mathfrak{E}}}
\newcommand{\translabelany}{\dagger}
\newcommand{\styleAutomaton}[1]{\ensuremath{\mathcal{#1}}}
\newcommand{\styleSetRegion}[1]{\ensuremath{\mathsf{#1}}}
\newcommand{\styleBigSet}[1]{\ensuremath{\mathbb{#1}}}
\newcommand{\styleSetBelief}[1]{\ensuremath{\mathfrak{#1}}}

\newcommand{\regset}[1]{\styleSetRegion{R}_{#1}}
\newcommand{\regaut}[1]{\styleAutomaton{R}_{#1}}

\newcommand{\clockeq}{\approx}

\newcommand{\constantmax}[1]{c_{#1}}
\newcommand{\class}[1]{[#1]}

\newcommand{\beliefcontrolset}[2]{\styleBigSet{I}_{#1}^{#2}}
\newcommand{\beliefencountset}[2]{\styleBigSet{E}_{#1}^{#2}}
\newcommand{\finalclass}[1]{\styleSetRegion{R}^F_{#1}}
\newcommand{\secret}[1]{\styleSetRegion{Private}_{#1}}
\newcommand{\notsecret}[1]{\styleSetRegion{Public}_{#1}}
\newcommand{\globaltime}{t}
\newcommand{\compatible}{\ensuremath{\strategy}\mbox{-compatible}}
\newcommand{\compatiblearg}[1]{\ensuremath{#1}\mbox{-compatible}}

\newcommand{\beliefAutomaton}{belief automaton} %
\newcommand{\BeliefAutomaton}{Belief automaton} %
\newcommand{\beliefNoArg}{\ensuremath{\styleSetBelief{b}}}
\newcommand{\belief}[1]{\ensuremath{\beliefNoArg_{#1}}}
\newcommand{\beliefInitState}{\ensuremath{\bot}}

\newcommand{\BeliefAut}[1]{\ensuremath{\styleAutomaton{B}_{#1}}}

\newcommand{\strategy}{\ensuremath{\sigma}}
\newcommand{\substrat}{\ensuremath{\nu}}

\newcommand{\bStrategy}{\beliefNoArg{}-strategy}

\newcommand{\beliefcontrol}[2]{\beliefNoArg_{#1}^{#2}}
\newcommand{\BeliefAutState}[1]{\ensuremath{\styleSetBelief{S}^{\BeliefAut{#1}}}}
\newcommand{\BeliefAutActions}[1]{\styleSetBelief{A}^{\BeliefAut{#1}}}
\newcommand{\BeliefAutTransitions}[1]{\styleSetBelief{d}^{\BeliefAut{#1}}}
\newcommand{\BeliefControlAut}[2]{\BeliefAut{#1}^{#2}}
\newcommand{\BeliefControlAutState}[2]{\ensuremath{\styleSetBelief{S}^{\BeliefControlAut{#1}{#2}}}}
\newcommand{\BeliefControlAutActions}[2]{\ensuremath{\styleSetBelief{A}^{\BeliefControlAut{#1}{#2}}}}
\newcommand{\BeliefControlAutTransitions}[2]{\ensuremath{\styleSetBelief{d}^{\BeliefControlAut{#1}{#2}}}}

\newcommand{\intervalBelief}{interval belief}
\newcommand{\intervalBeliefs}{interval beliefs}
\newcommand{\shrink}{\ensuremath{\mathit{shrink}}}

\newcommand{\fract}[1]{\ensuremath{\textnormal{fr}\left(#1\right)}}
\newcommand{\intpart}[1]{\ensuremath{\lfloor#1\rfloor}}
\newcommand{\uppart}[1]{\ensuremath{\lceil#1\rceil}}

\newcommand{\assign}{\leftarrow}

\newcommand{\checkUseMacro}[1]{#1}

\usepackage[english]{nomencl}

\newcommand{\set}[1]{\ensuremath{\left\{#1\right\}}}

\newcommand{\setN}{\ensuremath{\mathbb{N}}}
\newcommand{\setQ}{\ensuremath{\mathbb{Q}}}
\newcommand{\setQgeqzero}{\ensuremath{\setQ_{\geq 0}}}
\newcommand{\setR}{\ensuremath{\mathbb{R}}}
\newcommand{\setRgeqzero}{\ensuremath{\setR_{\geq 0}}}
\newcommand{\setRpositif}{\ensuremath{\setR_{>0}}}
\newcommand{\setRnegatif}{\ensuremath{\setR_{<0}}}
\newcommand{\setZ}{\ensuremath{\mathbb{Z}}}

\newcommand{\interval}{\ensuremath{\mathcal{I}}}
\newcommand{\subinterval}{\ensuremath{\iota}}

\newcommand{\init}{\ensuremath{0}}
\newcommand{\priv}{\ensuremath{{\mathit{priv}}}}

\newcommand{\final}{\ensuremath{f}}

\newcommand{\admits}{\ensuremath{\vdash}}
\newcommand{\satisfies}{\ensuremath{\models}}

\newcommand{\clock}{\ensuremath{\textstyleclock{x}}}
\nomenclature[C]{\clock}{A clock} %
\newcommand{\clocky}{\ensuremath{\textstyleclock{y}}}

\newcommand{\clocki}[1]{\ensuremath{\textstyleclock{\clock_{#1}}}}
\newcommand{\ClockCard}{H} %
\newcommand{\clockval}{\ensuremath{\mu}}\nomenclature[C]{\clockval}{Clock valuation} %
\newcommand{\ClockSet}{\ensuremath{\mathbb{X}}} %
\newcommand{\ClocksZero}{\ensuremath{\vec{0}}}\nomenclature[C]{\ClocksZero}{Clock valuation assigning $0$ to all clocks}

\newcommand{\resets}{\ensuremath{R}}
\newcommand{\reset}[2]{\ensuremath{[#1]_{#2}}}

\newcommand{\clockextra}{\ensuremath{\textstyleclock{z}}}

\newcommand{\TA}{\ensuremath{\checkUseMacro{\styleAutomaton{A}}}}
\newcommand{\TAprivextend}{\ensuremath{(\ActionSet, \LocSet, \locinit, \locpriv, \LocFinalSet, \ClockSet, \invariant, \EdgeSet)}}
\newcommand{\TAdup}{\ensuremath{\mathcal{A}^{\mathit{dup}}}}
\newcommand{\TAdupi}[1]{\mathcal{A}^{dup}_{#1}}
\newcommand{\TAi}[1]{\TA_{#1}}
\newcommand{\exTAopaque}{\ensuremath{\TA_{\mathit{opaque}}}}
\newcommand{\exTAopaquebis}{\ensuremath{\TA'_{\mathit{opaque}}}}
\newcommand{\exTANFV}{\ensuremath{\TA_{\mathit{nfv}}}}
\newcommand{\exTAcounterex}{\ensuremath{\TA_{\mathit{counterex}}}}

\newcommand{\action}{\ensuremath{\textstyleact{a}}}%
\newcommand{\ActionSet}{\ensuremath{\Sigma}}%
\newcommand{\ActionSetC}{\ensuremath{\ActionSet_c}}
\newcommand{\ActionSetU}{\ensuremath{\ActionSet_u}}

\newcommand{\constraint}{\ensuremath{C}}
\newcommand{\true}{\ensuremath{\mathit{true}}}

\newcommand{\edge}{\ensuremath{\checkUseMacro{\mathsf{e}}}}\nomenclature[P]{\edge}{An edge}
\newcommand{\edgei}[1]{\ensuremath{\checkUseMacro{\edge_{#1}}}}
\newcommand{\edgeicolor}[1]{\ensuremath{\textcolor{purple}{\edgei{#1}}}}
\newcommand{\EdgeSet}{\ensuremath{E}}%
\newcommand{\edgedec}{(\loc, \guard, \action, \resets, \loc')}

\newcommand{\guard}{\ensuremath{g}}

\newcommand{\SeqTransitions}{\ensuremath{v}}

\newcommand{\invariant}{\ensuremath{I}}
\newcommand{\invariantof}[1]{\invariant(#1)}
\newcommand{\invariantofdup}[1]{\invariant'(#1)}

\newcommand{\loc}{\ensuremath{\textstyleloc{\ell}}}\nomenclature[P]{\loc}{A location} %
\newcommand{\loci}[1]{\ensuremath{\textstyleloc{\loc_{#1}}}}
\newcommand{\locdup}{\loc^\textstyleloc{p}}
\newcommand{\locdupi}[1]{\loc_{#1}^\textstyleloc{p}}
\newcommand{\locinit}{\ensuremath{\textstyleloc{\loc_\init}}}
\newcommand{\locfinal}{\ensuremath{\textstyleloc{\loc_\final}}}
\newcommand{\locfinalpriv}{{\textstyleloc{\loc}}^\textstyleloc{p}_{\textstyleloc{\final}}}
\newcommand{\locpriv}{\ensuremath{\textstyleloc{\loc_\priv}}}
\newcommand{\locduppriv}{\ensuremath{\textstyleloc{\loc_\priv^\textstyleloc{p}}}}
\newcommand{\LocSet}{\ensuremath{L}}\nomenclature[P]{\LocSet}{Set of locations} %
\newcommand{\LocFinalSet}{F}

\newcommand{\run}{\checkUseMacro{\rho}} %

\newcommand{\duration}{\ensuremath{\mathit{dur}}} %
\newcommand{\runduration}[1]{\ensuremath{\checkUseMacro{\duration}(#1)}}

\newcommand{\region}{\ensuremath{r}}

\newcommand{\regionclass}[1]{\ensuremath{\left[#1\right]}}

\newcommand{\silentaction}{\ensuremath{\varepsilon}}
\newcommand{\translabelinstant}{0}
\newcommand{\translabelsame}{0^+}
\newcommand{\translabelchge}{1}

\newcommand{\RegAutTransitions}{\delta^{\styleSetRegion{R}}}
\newcommand{\RegAutActions}{\ActionSet^{\styleSetRegion{R}}}
\newcommand{\RegAutTransitionWith}[1]{\xrightarrow{#1}_{\styleSetRegion{R}}}

\newcommand{\metaStrategy}{meta-\mbox{strategy}}
\newcommand{\MetaStrategy}{Meta-strategy}
\newcommand{\metaStrategies}{meta-strategies}
\newcommand{\MetaStrategies}{Meta-strategies}
\newcommand{\metastrat}{\ensuremath{\phi}}
\renewcommand{\next}[1]{\ensuremath{\mathit{next}_{#1}}}

\newcommand{\TTS}{\ensuremath{\mathit{TTS}}}
\newcommand{\TTSStates}{\ensuremath{S}}
\newcommand{\TTSstate}{\ensuremath{s}}
\newcommand{\TTSinit}{\ensuremath{\TTSstate_{\init}}}

\newcommand{\TTStransition}{\ensuremath{\delta}}
\newcommand{\TTStransitionControl}{\ensuremath{\delta^\strategy}} %

\newcommand{\semantics}[1]{\ensuremath{\TTS_{#1}}}
\newcommand{\semanticsextend}{\ensuremath{\left(\TTSStates, \TTSinit, \ActionSet \cup \setRgeqzero, \TTStransition \right)}}
\newcommand{\semanticscontrolextend}{\ensuremath{\left(\TTSStates, \TTSinit, \ActionSet \cup \setRgeqzero, \TTStransitionControl \right)}}

\newcommand{\transitionWith}[1]{\stackrel{#1}{\mapsto}}
\newcommand{\transitionWithControl}[1]{\stackrel{#1}{\mapsto}_{\sigma}}

\newcommand{\transitionLabel}[1]{\xrightarrow{#1}}

\newcommand{\timeEdge}[1]{ \paramd_{#1}, \edge_{#1}}

\newcommand{\transitionControlLabel}[2]{\transitionLabel{\timeEdge{#1}}_{#2}}

\newcommand{\laststate}{\ensuremath{\mathit{last}}}

\newcommand{\paramd}{\textstyleparam{\ensuremath{d}}}
\newcommand{\expiringBound}{\checkUseMacro{\ensuremath{\Delta}}}\nomenclature[P]{\expiringBound}{Expiring bound considered in \vref{sec:opacity:ICECCS}}
\newcommand{\setBound}{\ensuremath{\mathcal{D}}}\nomenclature[P]{\setBound}{Set of expiring bounds} %

\newcommand{\PrivVisit}[1]{\ensuremath{\mathit{Visit}^{\mathit{priv}}(#1)}}
\newcommand{\PubVisit}[1]{\ensuremath{\mathit{Visit}^{\mathit{pub}}(#1)}}
\newcommand{\PrivVisitStrat}[2]{\ensuremath{\mathit{Visit}^{\mathit{priv}}_{#2}(#1)}}
\newcommand{\PubVisitStrat}[2]{\ensuremath{\mathit{Visit}^{\mathit{pub}}_{#2}(#1)}}

\newcommand{\PrivDurVisit}[1]{\ensuremath{\mathit{DVisit}^\mathit{priv}(#1)}}
\newcommand{\PubDurVisit}[1]{\ensuremath{D\mathit{Visit}^{\mathit{pub}}(#1)}}
\newcommand{\PrivDurVisitStrat}[2]{\ensuremath{\mathit{DVisit}^\mathit{priv}_{#2}(#1)}}
\newcommand{\PubDurVisitStrat}[2]{\ensuremath{D\mathit{Visit}^{\mathit{pub}}_{#2}(#1)}}

\newcommand{\TAtext}{TA}
\newcommand{\TAstext}{TAs}

\newcommand{\opaqueText}{ET-opaque}
\newcommand{\opacityText}{ET-opacity}

\newcommand{\existentialOpaqueText}{\checkUseMacro{\ensuremath{\exists}-ET-opaque}}
\newcommand{\existentialOpacityText}{\checkUseMacro{\ensuremath{\exists}-ET-opacity}}
\newcommand{\ExistentialOpacityText}{\existentialOpacityText{}}

\newcommand{\weakOpaqueText}{\checkUseMacro{weakly ET-opaque}}
\newcommand{\weakOpacityText}{\checkUseMacro{weak ET-opacity}}
\newcommand{\fullOpaqueText}{\checkUseMacro{fully ET-opaque}}
\newcommand{\fullOpacityText}{\checkUseMacro{full ET-opacity}}

\newcommand{\WeakOpacityText}{\checkUseMacro{Weak ET-opacity}}
\newcommand{\FullOpacityText}{\checkUseMacro{Full ET-opacity}}

\newcommand{\almostOpacityText}{almost \opacityText{}}

\newcommand{\closedOpacityText}{closed \opacityText{}}

\newcommand{\almostOpaqueText}{almost \opaqueText{}}

\newcommand{\AlmostFullOpacityText}{Almost \fullOpacityText{}}

\newcommand{\almostFullOpacityText}{almost \fullOpacityText{}}

\newcommand{\almostFullOpaqueText}{almost \fullOpaqueText{}}
\newcommand{\ClosedFullOpacityText}{Closed \fullOpacityText{}}
\newcommand{\closedFullOpacityText}{closed \fullOpacityText{}}

\newcommand{\closedFullOpaqueText}{closed \fullOpaqueText{}}

\newcommand{\open}[1]{\llparenthesis #1 \rrparenthesis}
\newcommand{\closed}[1]{\llbracket #1 \rrbracket}

\newcommand{\ComplexityFont}[1]{{\sffamily\upshape #1}}

\newcommand{\NEXPTIME}{\ComplexityFont{NEXPTIME}\xspace}
\newcommand{\twoEXPTIME}{\ComplexityFont{2EXPTIME}\xspace}
\newcommand{\NLOGSPACE}{\ComplexityFont{NLOGSPACE}\xspace}

\newcommand{\EXPSPACE}{\ComplexityFont{EXPSPACE}\xspace}
\newcommand{\NEXPSPACE}{\ComplexityFont{NEXPSPACE}\xspace}

\newcommand{\counterOne}{\ensuremath{C_1}}
\newcommand{\counterTwo}{\ensuremath{C_2}}
\newcommand{\counteri}{\ensuremath{C}}
\newcommand{\twoCM}{\ensuremath{\mathcal{M}}}
\newcommand{\twoCMcommand}{\ensuremath{c}}

\newcommand{\orcidID}[1]{\href{https://orcid.org/#1}{\includegraphics[height=1em]{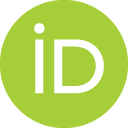}}}

\usepackage{fontawesome}
\newcommand{\homepage}[1]{\href{#1}{\color{gray}\faHome}}

\begin{document}

\title[Execution-time opacity control for timed automata]{Execution-time opacity control for timed automata}

\author[1,2]{\fnm{\'Etienne} \sur{Andr\'e}}%
\equalcont{These authors contributed equally to this work.}

\author[3]{\fnm{Marie} \sur{Duflot}}%
\equalcont{These authors contributed equally to this work.}

\author[4]{\fnm{Laetitia} \sur{Laversa}}%
\equalcont{These authors contributed equally to this work.}

\author[3]{\fnm{Engel} \sur{Lefaucheux}}%
\equalcont{These authors contributed equally to this work.}

\affil[1]{\orgdiv{CNRS, Laboratoire d’Informatique de Paris Nord, LIPN}, \orgname{Université~Sorbonne~Paris~Nord}, \orgaddress{\street{Av.\ Jean-Baptiste Clément}, \city{Villetaneuse}, \postcode{F-93430}, %
	\country{France}}}

\affil[2]{\orgname{Institut Universitaire de France (IUF)}, \orgaddress{\country{France}}}

\affil[3]{\orgdiv{\orgname{Université de Lorraine,
CNRS, Inria, LORIA}, \country{France}}}

\affil[4]{Université Paris Cité, CNRS, IRIF, F-75013, Paris, France}

\abstract{%
	Timing leaks in timed automata (TA) can occur whenever an attacker is able to deduce a secret by observing some timed behaviour.
	In execution-time opacity, the attacker aims at deducing whether a private location was visited, by observing only the execution time.
	In earlier work, it was shown that it can be decided whether a TA is opaque in this setting.
	In this work, we address control, and investigate whether a TA can be controlled by a strategy at runtime to ensure opacity, by enabling or disabling some controllable actions over time.
	We first show that, in general, it is undecidable to determine whether such a strategy exists.
	Second, we show that deciding whether a meta-strategy ensuring opacity exists can be done in \EXPSPACE{}. Such a meta-strategy is a set of strategies allowing an arbitrarily large---yet finite---number of strategy changes per time unit, and with only weak ordering relations between such changes.
	Our method is constructive, in the sense that we can exhibit such a meta-strategy.
	We also extend our method to the case of weak opacity, when it is harmless that the attacker deduces that the private location was \emph{not} visited.
	\sectionEight{%
	Finally, we consider a variant where the attacker cannot have an infinite precision in its observations.
	}

	}


\maketitle{}
\section{Introduction}

In order to infer sensitive information, side-channels attacks \cite{Stan10} exploit various observable aspects of a system rather than directly exploiting its computational processes; such observable aspects can include power consumption, electromagnetic emissions, or time.
In particular, by observing subtle differences in timing, attackers can infer valuable information about the internal state of the system.
For example, in~\cite{CHSJX22}, a timing attack vulnerability is identified in the Chinese public key cryptography standard; the authors show how the most significant zero-bit leakage obtained from the execution time allows to extract the secret key.

Timing attacks such as timing leaks often depend on the precise duration of operations, which finite-state automata cannot model.
Timed automata~\cite{AD94} (\TAstext{}), on the other hand, incorporate explicit clocks and timing constraints, making them essential for analysing and detecting vulnerabilities related to timing information.
\TAstext{} are a powerful formalism to reason about real-time systems mixing timing constraints and concurrency.
Timing leaks can occur whenever an attacker is able to deduce a secret by observing some (timed) behaviour of a~TA.

\subsection{Related works}

\paragraph{Opacity in timed automata}
Franck~Cassez proposed in~\cite{Cassez09} a first definition of \emph{timed} opacity for \TAstext{}: the system is opaque when an attacker can never deduce whether some secret sequence of actions (possibly with timestamps) was performed, by only observing a given set of observable actions together with their timestamp.
It is then proved in~\cite{Cassez09} that it is undecidable whether a \TAtext{} is opaque.
	This notably relates to the undecidability of timed language inclusion for \TAstext{}~\cite{AD94}.
	The undecidability of opacity is strong: it holds even for the restricted class of \emph{event-recording automata}~\cite{AFH99}---a subclass of~\TAstext{} for which language inclusion is actually decidable.

The aforementioned negative result leaves hope only if the definition or the setting is changed, which was done in four main lines of work.

First, in~\cite{WZ18,WZA18}, the input model is simplified to \emph{real-time automata}~\cite{Dima01}, a restricted formalism compared to \TAstext{}: real-time automata can be seen as \TAstext{} with a single clock, reset at each transition.
\cite{LLHL22}~works on constant-time labelled automata, a subclass of real-time automata where events occur at constant values.
In this setting, initial-state opacity (``according to the observations, what was the initial state?'')\ and current-state opacity (``according to the observations, what is the current location?'')\ become decidable.
In~\cite{Zhang24}, Zhang studies labelled real-timed automata (a subclass of labelled \TAstext{}); in this setting, state-based (at the initial time, the current time, etc.)\ opacity is proved to be decidable by extending the observer (that is, the classical powerset construction) from finite automata to labelled real-timed automata.

Second, in~\cite{AGWZH24,ADL24}, the opacity was studied in the setting of Cassez' definition, but with restrictions in the model: one-clock automata, one-action automata, or over discrete time.
Similarly, in~\cite{KKG24}, discrete-time automata with several clocks are considered and transformed into tick automata in order to verify the current-state opacity.
The discrete time setting yields decidability, while restricting the number of actions to~1 preserves undecidability; for a single clock, decidability can only be envisioned without silent actions~\cite{ADL24} (allowing silent actions or allowing two clocks immediately leads to undecidability).

Third, in~\cite{AEYM21}, the authors consider a \emph{time-bounded} notion of the opacity of~\cite{Cassez09}, where the attacker has to disclose the secret before a deadline, using a partial observability.
	This can be seen as a secrecy with an \emph{expiration date}.
	The rationale is that retrieving a secret ``too late'' is useless; this is understandable, \eg{} when the secret is the value in a cache; if the cache has been overwritten since, then knowing the secret is probably useless in most situations.
	In addition, the analysis is carried over a time-bounded horizon; this means there are two time bounds in~\cite{AEYM21}: one for the secret expiration date, and one for the bounded-time execution of the system.
Deciding opacity in this setting is shown to be decidable for~TAs.

Fourth, in~\cite{ALLMS23}, an alternative definition to Cassez' opacity is proposed, by studying \opacityText{} (execution-time opacity): the attacker has only access to the \emph{execution time} of the system, as opposed to Cassez' partial observations where some events (with their timestamps) are observable.
The goal for the attacker is to deduce whether a special secret location was visited, by observing only the execution time.
In that case, most problems for TAs become decidable, including some problems when introducing an expiration date~\cite{ALM23} (see \cite{ALLMS23} for a survey).
Our current work fits in this \opacityText{} context, with the additional goal to control the system.

\paragraph{Opacity in other formalisms}
Different variants of opacity are also studied for other types of systems, such as stochastic systems.
In this case, we can quantify the probability that a system is opaque~\cite{BMS15}.
In particular, opacity can be related to the bandwidth of a language~\cite{JDA22,ADDI23} used to encode information by the delay necessary to produce it.

\paragraph{Non-interference in timed automata}
Several works address non-interference for \TAstext{}.
	In this context, actions are either high-level or low-level, and only low-level actions are observable.
	A \TAtext{} satisfies non-interference whenever its behaviour in absence of high-level actions is equivalent to the observation of its behaviour when high-level actions occur.
Different notions of equivalence (\eg{} bisimulation) can be considered for this property.
Several papers \cite{BDST02,BT03,AK20} present some decidability results,
while control is considered in~\cite{BCLR15}.

General security problems for \TAstext{} are surveyed in~\cite{AA23survey}.

\paragraph{Control}
A preliminary version of control for \opacityText{} in TAs was considered in~\cite{ABLM22}, but only untimed, \ie{} the actions could only be enabled or disabled once and for all, thus severely restricting the possibilities to render the system \opaqueText{}.
In addition, \cite{ALLMS23} considers \emph{parametric} versions of the opacity problems, in which timing parameters~\cite{AHV93} can be used in order to make the system \opaqueText{}.
This parametric analysis was then used for the analysis of C~code~\cite{ABCFLMRS25}.
Our notion of control is orthogonal to parameter synthesis, as another way to ensure the system becomes \opaqueText{}.

Controller synthesis can be described and solved thanks to game theory; finding a strategy for a controller can be equivalent to computing a winning strategy in a corresponding game. Several game models have been considered, as timed games that can be used to solve synthesis problem on \TAstext{}.
In this context, \cite{AMPS98} aims to restrict the transition relation in order to satisfy certain properties, while \cite{JT07} completes this result, minimizing the execution time, and \cite{BFM15} studies the reachability with robust strategies only.

\subsection{Contributions}

In this work, we aim at tuning a system to make it \opaqueText{}, by \emph{controlling} it at runtime.

Our attacker model is as follows: the attacker has a knowledge of the system model, but can only observe the execution time.
This can correspond to an attacker observing the energy consumption of a device, clearly denoting the execution time of a program or process;
or to an attacker observing communications over a shared network, with an observable message acknowledging the end of execution.
The attacker aims at deducing---only by observing the execution time---whether a special secret location was visited.

As usual, we consider that the system actions are partitioned between controllable and uncontrollable.
Our controller relies on the following notion of strategy: at each timestamp, the strategy enables only a subset of the controllable actions.

We mainly consider in this paper \fullOpacityText{}, \ie{} whenever the durations corresponding to executions visiting the secret location match the durations corresponding to executions \emph{not} visiting the secret location.
Our first contribution is to show that the \fullOpacityText{} strategy emptiness problem, \ie{} the emptiness of the set of strategies such that a TA is \fullOpaqueText{} with such a strategy, is undecidable.
Second, we move to a weaker version, that of \emph{\metaStrategies{}}, \ie{} sets of strategies that specify a finite number of strategy changes per time unit, but not the precise time at which those changes take place.
In that case, we show that not only the \fullOpacityText{} \metaStrategy{} emptiness problem is decidable in \EXPSPACE{}, but our approach is also constructive, in the sense that we can build such a controller (in \twoEXPTIME{}).
Our technique relies on a novel \emph{ad-hoc} construction inspired by the region automaton for~\TAstext{}.

Depending of the system, various degrees of opacity can be interesting,
and we therefore consider several variants introduced in~\cite{ALLMS23} as our third contribution:
in addition to \fullOpacityText{} (the durations corresponding to executions visiting the secret location match the durations corresponding to executions \emph{not} visiting the secret location),
we also consider \weakOpacityText{} (the durations corresponding to executions visiting the secret location are \emph{included} in the durations corresponding to executions not visiting the secret location),
and we briefly discuss \existentialOpacityText{} (in which we are simply interested in the \emph{existence} of one execution time for which opacity is ensured).
We show that, for both variants, the \metaStrategy{} emptiness problem is decidable.

Finally, and as a fourth contribution, we address the case when the attacker cannot have an infinite precision in its observations, \ie{} enhancing our method with a notion of robustness.

\paragraph{About this manuscript}
This manuscript is an extended version of the paper published in the proceedings of {SEFM}~2024~\cite{ADLL24}.
We list the enhancements (and differences) with respect to the conference version:
\begin{enumerate}
	\item We add the missing proofs of all results, and additional examples;
	\item We correct the previous construction from~\cite{ADLL24}, adding the notion of \metaStrategy{} and, on the one hand, we prove the undecidability of the general problem of the existence of a strategy making the system \fullOpaqueText{} while, on the other hand, we prove the decidability of the existence of a \metaStrategy{};
	\item We investigate control with respect to \weakOpacityText{} and \existentialOpacityText{};
	\sectionEight{%
		\item We investigate several robust definitions of \opacityText{} and add details and proofs for results already stated in~\cite{ADLL24};
	}
	\item We rename ``bad beliefs'' into ``\badBelief{} beliefs''.
\end{enumerate}

\subsection{Outline}

\cref{section:preliminaries} recalls the necessary material.
\cref{section:problem} defines the control problem for \opacityText{}.
\cref{section:undecidability} proves our main undecidability result.
We then introduce the core of our decidable approach:
\cref{section:beliefs} introduces the notion of belief automaton, while
\cref{section:solution} solves the \opacityText{} problems thanks to this notion.
\cref{section:weak-existential-opacity} then extends our approach to \emph{existential} and \emph{weak} opacity.
\sectionEight{%
The last contribution in \cref{section:robust-extensions} considers an attacker which cannot have an infinite precision in observing the execution time.
}
\cref{section:conclusion} highlights future works.

\section{Preliminaries}\label{section:preliminaries}

Let $\setN$, $\setZ$, $\setQgeqzero$, $\setRgeqzero$, $\setRpositif$, $\setRnegatif$
denote the sets of
non-negative integer numbers,
integer numbers,
non-negative rational numbers,
non-negative real numbers,
positive real numbers
and
negative real numbers respectively.

\paragraph{Clock constraints}
\emph{Clocks} are real-valued variables that all evolve over time at the same rate.
Throughout this paper, we assume a set~$\ClockSet = \{ \clocki{1}, \dots, \clocki{\ClockCard} \} $ of \emph{clocks}.
A \emph{clock valuation} is a function
$\clockval : \ClockSet \rightarrow \setRgeqzero^\ClockCard$, assigning a non-negative real value to each clock.
Given $\resets \subseteq \ClockSet$, we define the reset of a valuation~$\clockval$ with respect to $\resets$, denoted by $\reset{\clockval}{\resets}$, as follows: $\reset{\clockval}{\resets}(\clock) = 0 $ if $\clock \in \resets$, and  $\reset{\clockval}{\resets}(\clock) = \clockval(\clock) $ otherwise.
We write $\ClocksZero$ for the clock valuation assigning $0$ to all clocks.
Given a constant $\paramd \in \setRgeqzero$, $\clockval + \paramd$ denotes the valuation \st\ $(\clockval + \paramd)(\clock) = \clockval(\clock) + \paramd$, for all $\clock \in \ClockSet$.

We assume ${\bowtie} \in \{<, \leq, =, \geq, >\}$.
A \emph{constraint}~$\constraint$ is a conjunction of inequalities over $\ClockSet$ of the form
$\clock \bowtie \paramd$, with
$\paramd \in \setZ$.

	A table of the notations used throughout this paper is available in \cref{appendix:notations}.
\subsection{Timed automata}
\paragraph{Syntax of TAs}

We define timed automata as in~\cite{AD94}, with an extra private location as in~\cite{ALLMS23}, which encodes the secret that shall not be leaked.

\begin{definition}[Timed automaton]\label{def:TAs}
 A \emph{timed automaton} (TA) $\TA$ is a tuple $\TA = \TAprivextend$ where:
\begin{oneenumerate}%
\item $\ActionSet$ is a finite set of actions,
\item $\LocSet$ is a finite set of locations,
\item $\locinit \in \LocSet$ is the initial location,
\item $\locpriv \in \LocSet$ is the private location,
\item $\LocFinalSet \subseteq \LocSet \setminus \{ \locpriv \}$ is the set of final locations,
\item $\ClockSet = \{ \clocki{1}, \dots, \clocki{\ClockCard} \}$ is a finite set of clocks,
\item $\invariant$ is the invariant, assigning to every $\loc \in \LocSet$ a constraint $\invariantof{\loc}$,
\item $\EdgeSet$ is a finite set of edges $\edge = \edgedec$
where
$\loc, \loc' \in \LocSet$ are the source and target locations, $a \in \ActionSet \cup \{\silentaction\}$, where $\silentaction$ denotes the silent action, $\resets \subseteq \ClockSet$ is a set of clocks to be reset, and $\guard$ is a constraint over $\ClockSet$ (called \emph{guard}).
\end{oneenumerate}%
\end{definition}
\begin{example}\label{example:TA}
	\cref{fig:aut:non-strategy-opaque} depicts a TA $\TAi{1}$ with a single clock~$\clock$, where $\ActionSet = \{a, b, u\}$.
	The edge $\mathsf{e_1}$ between the initial location~$\locinit$ and the private location~$\locpriv$ is available only when the valuation of~$\clock$ equals~0.
	The edge $\mathsf{e_6}$ between $\locinit$ and $\loci{2}$ resets~$\clock$.

\end{example}

Since we are only interested in the (first) arrival time in a final location, the following assumption does not restrict our framework, but simplifies the subsequent definitions and results.

\begin{assumption}\label{assumption-urgent-final}
	We consider every final location as \emph{urgent} (where time cannot elapse): formally, there exists \mbox{$\clock \in \ClockSet$} such that, for all $\edgedec \in \EdgeSet, \loc' \in \LocFinalSet$, we have $\clock \in \resets$ and ``$x = 0$''$\ \in \invariantof{\loc'}$.
	Moreover, final locations cannot have outgoing transitions: formally, there is no \mbox{$\edgedec \in \EdgeSet$} s.t.\ $\loc \in \LocFinalSet$.
\end{assumption}
\paragraph{Semantics of TAs}
\begin{definition}[Semantics of a TA]\label{definition-TA-semantics}
	Let \mbox{$\TA = \TAprivextend$} be a \TAtext{}.
	The semantics of~$\TA$ is given by the timed transition system %
		\mbox{$\semantics{\TA} = \semanticsextend$}, with
	\begin{oneenumerate}%
		\item $\TTSStates = \big\{ (\loc, \clockval) \in \LocSet \times \setRgeqzero^\ClockCard \mid \clockval \models \invariantof{\loc} \big\}$,
		\item $\TTSinit = (\locinit, \ClocksZero) $,
		\item  $\TTStransition$ consists of the discrete and (continuous) delay transition relations:
		\begin{ienumerate}%
			\item discrete transitions: $(\loc,\clockval) \transitionWith{\edge} (\loc',\clockval')$,
			if $(\loc, \clockval) , (\loc',\clockval') \in \TTSStates$, and there exists $\edge = (\loc,\guard,\action,\resets,\loc') \in \EdgeSet$, such that $\clockval'= \reset{\clockval}{\resets}$, and $\clockval\models\guard$.
			\item delay transitions: $(\loc,\clockval) \transitionWith{\paramd} (\loc, \clockval+\paramd)$, with $\paramd \in \setRgeqzero$, if $\forall \paramd' \in [0, \paramd], (\loc, \clockval+\paramd') \in \TTSStates$.
		\end{ienumerate}%
	\end{oneenumerate}%
\end{definition}
We write $(\loc, \clockval) \transitionLabel{\paramd, \edge} (\loc', \clockval')$  for a combination of a delay and a discrete transitions when $\exists \clockval'' : (\loc, \clockval) \transitionWith{\paramd} (\loc, \clockval'') \transitionWith{\edge} (\loc', \clockval')$.

Given a TA $\TA$ with semantics $\semanticsextend$, a \emph{run} of~$\TA$ is a finite alternating sequence of states of $\semantics{\TA}$ and pairs of delays and edges starting from the initial state~$\TTSinit$
of the form $\TTSinit, (\timeEdge{0}), \TTSstate_1, \ldots, \TTSstate_n$ where for all $i < n, \edge_i \in \EdgeSet, \paramd_i \in \setRgeqzero$ and $\TTSstate_i \transitionLabel{\paramd_i, \edge_i} \TTSstate_{i+1}$.
The duration of a run $\run = \TTSinit, (\timeEdge{0}), \TTSstate_1, \ldots, (\timeEdge{n-1}), \TTSstate_n$ is $\runduration{\run} =
 \sum_{0 \leq i \leq n-1} \paramd_i$.\label{def:duration} We define $\laststate(\run) = \TTSstate_n$. %

\newcommand{\decalage}{3.5}
\newcommand{\yfactor}{.6}

\begin{figure*}[tb]
		\begin{subfigure}[b]{0.49\textwidth}
			\begin{tikzpicture}[pta, scale=1, every node/.style={scale=1}]
				\node[location, initial, initial text=] (q0) {$\locinit$};
				\node[location, private] at (3, 2*.8) (qpriv) {$\locpriv$};
				\node[location] at (3, 0) (q1) {$\loci{1}$};
				\node[location] at (2, -2*.8) (q2) {$\loci{2}$};
				\node[location] at (4, -2*.8) (q3) {$\loci{3}$} ;
				\node[location, final] at (6, 0) (qfin) {$\locfinal$};
				\path[->]
				(q0) edge[loop above] node[above left,align=center] {\shortstack{ ${\edgeicolor{3}}$ \\ $\styleclock{x} =1$ \\ $\styleact{u}$ \\ $\styleclock{x} \assign 0 $}} (q0)
					edge node[above, sloped, align=center] {\shortstack{ ${\edgeicolor{1}}$ \\ $\styleclock{x} = 0$}} node[below, sloped] {$\styleact{u}$} (qpriv)
					edge node[below, sloped,, pos=.7] {$\styleact{a}$} node[above, sloped, pos=.7] {${\edgeicolor{4}}$} (q1)
					edge node[below, sloped, align=center] {\shortstack {$\styleact{b}$ \\ $\styleclock{x} \assign 0$}} node[above, sloped] { ${\edgeicolor{6}}$ }(q2)
				(qpriv) edge node[above, sloped, align=center] {\shortstack{${\edgeicolor{2}}$ \\ $\styleclock{x} = 0$}} node[below, sloped] {$\styleact{u}$} (qfin)
				(q1) edge node[below, sloped, pos=.3] {$\styleact{b}$} node[above, sloped, pos=.3] {${\edgeicolor{5}}$} (qfin)
				(q2) edge node[above, sloped, align=center] {\shortstack{${\edgeicolor{7}}$ \\ $\styleclock{x} = 0$}} node[below, sloped] {$\styleact{a}$} (q3)
				(q3) edge node[below, sloped] {$\styleact{u}$} node[above,sloped] {${\edgeicolor{8}}$} (qfin)
				;
			\end{tikzpicture}
			\caption{TA $\TAi{1}$}
			\label{fig:aut:non-strategy-opaque}
		\end{subfigure}
		\hfill{}
		\begin{subfigure}[b]{0.49\textwidth}
			\begin{tikzpicture}[pta,  scale= 1, every node/.style={scale=1}]
				\node[location, initial, initial text=] at (0,0) (q0) {$\locinit$};
				\node[location, private] at (3, 2*\yfactor) (qpriv) {$\locpriv$};
				\node[location] at (3, 0) (q1) {$\loci{1}$};
				\node[location] at (2, -2*\yfactor) (q2) {$\loci{2}$};
				\node[location] at (4, -2*\yfactor) (q3) {$\loci{3}$} ;
				\node[location, final] at (6, 0) (qfin) {$\locfinal$};

				\node[locationdup] at (0,\decalage) (q0p) {$\locdupi{0}$};
				\node[locationdup] at (3, \decalage + 2*\yfactor) (qprivp) {$\locduppriv$};
				\node[locationdup] at (3, \decalage) (q1p) {$\locdupi{1}$};
				\node[locationdup] at (2, \decalage -2*\yfactor) (q2p) {$\locdupi{2}$};
				\node[locationdup] at (4, \decalage -2*\yfactor) (q3p) {$\locdupi{3}$} ;
				\node[locationdup, final] at (6, \decalage ) (qfinp) {$\locfinalpriv$};

				\path[->]
				(q0) edge[loop above] node[above left,align=center] {\shortstack{$\styleclock{x} =1$ \\ $\styleact{u}$ \\ $\styleclock{x} \assign 0 $}} (q0)
					edge node[above, sloped] {$\styleclock{x} = 0$} node[below, sloped] {$\styleact{u}$} (qpriv)
					edge node[above, sloped, pos=.7] {$\styleact{a}$} (q1)
					edge node[below, sloped] {$\styleclock{x} \assign 0$} node[above, sloped] {$\styleact{b}$}(q2)
				(q1) edge node[above, sloped, pos=.3] {$\styleact{b}$} (qfin)
				(q2) edge node[above, sloped] {$\styleclock{x} = 0$} node[below, sloped] {$\styleact{a}$} (q3)
				(q3) edge node[above, sloped] {$\styleact{u}$} (qfin)

				(qpriv) edge[edgedup,bend right] node[above, sloped] {$\styleclock{x} = 0$} node[below, sloped] {$\styleact{u}$} (qfinp)

				(q0p) edge[loop above] node[above left,align=center] {\shortstack{$\styleclock{x} =1$ \\ $\styleact{u}$ \\ $\styleclock{x} \assign 0 $}} (q0p)
					edge node[above, sloped] {$\styleclock{x} = 0$} node[below, sloped] {$\styleact{u}$} (qprivp)
					edge node[above, sloped, pos=.7] {$\styleact{a}$} (q1p)
					edge node[below, sloped] {$\styleclock{x} \assign 0$} node[above, sloped] {$\styleact{b}$}(q2p)
				(q1p) edge node[above, sloped, pos=.3] {$\styleact{b}$} (qfinp)
				(q2p) edge node[above, sloped] {$\styleclock{x} = 0$} node[below, sloped] {$\styleact{a}$} (q3p)
				(q3p) edge node[above, sloped] {$\styleact{u}$} (qfinp)
				(qprivp) edge node[above, sloped] {$\styleclock{x} = 0$} node[below, sloped] {$\styleact{u}$} (qfinp)
				;
			\end{tikzpicture}
		\caption{Duplicated version of~$\TAi{1}$}
		\label{fig:aut:non-strategy-opaque:duplicated}
	\end{subfigure}
	\caption{A TA and its duplicated version (introduced in \cref{section:beliefs})}
	\label{fig:aut:non-strategy-opaque:full-figure}
\end{figure*}

\paragraph{Extra clock}

We will need an extra clock~$\clockextra$ that will help us later to keep track of the elapsed absolute time.
This clock is reset exactly every 1~time unit, and therefore each reset corresponds to a ``tick'' of the absolute time.
(Note that its actual value remains in $[0,1]$ and therefore always matches the fractional part of the absolute time.)
In all subsequent region constructions, we assume the existence of $\clockextra \in \ClockSet$.
For each location $\loc$, we add the constraint ``$\clockextra \leq 1$'' to~$\invariantof{\loc}$, and we add a self-loop edge  $(\loc, \clockextra = 1, \silentaction, \{\clockextra\}, \loc)$.

\subsection{Regions}\label{subsection:regions}

Given a TA~$\TA$, for a clock~$\clock_i$, we denote by~$\constantmax{i}$ the largest constant to which~$\clock_i$ is compared within the guards and invariants
 of~$\TA$:
 formally, $\constantmax{i} = \max_j\{ \paramd_j \mid  x_i \bowtie \paramd_j \text{ appears in a guard or invariant of }\TA \}$.
Given a clock valuation $\clockval$ and a clock~$\clock_i$,
$\intpart{\clockval(\clock_i)}$ (resp.\ $\fract{\clockval(\clock_i)}$) denotes the integral (resp.\ fractional) part of $\clockval(\clock_i)$.

We now recall the equivalence relation between clock valuations.
\begin{definition}[Equivalence relation~\cite{AD94}]
	Two clocks valuations $\clockval,\clockval'$ are \emph{equivalent}, denoted by $\clockval \clockeq \clockval'$, when the following three conditions hold for any clocks $\clock_i, \clock_j \in \ClockSet$:
	\begin{enumerate}
		\item $\intpart{\clockval(\clocki{i})} = \intpart{\clockval'(\clocki{i})}$ or $\clockval(\clocki{i}) > \constantmax{i}$ and $\clockval'(\clocki{i}) > \constantmax{i}$; \label{item:equiv-int}
		\item if  $\clockval(\clocki{i}) \leq \constantmax{i}$ and $\clockval(\clocki{j}) \leq \constantmax{j}$: $\fract{\clockval(\clock_i)} \leq \fract{\clockval(\clock_j)}$ iff $\fract{\clockval'(\clock_i)} \leq \fract{\clockval'(\clock_j)}$; and\label{item:equiv-order}
		\item if  $\clockval(\clocki{i}) \leq \constantmax{i}$: $\fract{\clockval(\clock_i)} = 0$ iff $\fract{\clockval'(\clock_i)} = 0$.\label{item:equiv-zero}
	\end{enumerate}
\end{definition}
In other words, two valuations are equivalent when, for a given clock,
	the integral part is the same in both valuations or greater than the largest constant in both valuations (\cref{item:equiv-int}),
	for any two clocks, the fractional parts are in the same order in the two valuations (\cref{item:equiv-order}),
	and
	the fractional part is zero in both valuations or neither (\cref{item:equiv-zero}).

The equivalence relation $\clockeq$ is extended to the states of $\semantics{\TA}$: given two states $\TTSstate = (\loc, \clockval), \TTSstate' = (\loc', \clockval')$ of $\semantics{\TA}$, we write $\TTSstate \clockeq \TTSstate'$ iff $\loc = \loc'$ and $\clockval \clockeq \clockval'$.
We denote by $\class{\TTSstate}$ and call \emph{region} the equivalence class of a state~$\TTSstate$ for~$\clockeq$.
Then, $\TTSstate' \in \class{\TTSstate}$ when $\TTSstate \clockeq \TTSstate'$.
The set of all regions of~$\TA$ is denoted~$\regset{\TA}$.
A region $\region = \class{(\loc, \clockval)}$ is \emph{final} whenever $\loc \in \LocFinalSet$.
The set of final regions is denoted by~$\finalclass{\TA}$.
A region~$\region$ is \emph{reachable} when there exists a run~$\run$ such that $\laststate(\run) \in \region$.

\paragraph{Region automaton}

We now define a region automaton %
inspired by
\cite[Proposition~5.3]{BDR08} with two-component labels on the transitions
: the first component indicates how the fractional part of
$\clockextra$ evolved,
with symbol~``$\translabelinstant$'' corresponding to an absence of change,
symbol~``$\translabelsame$'' corresponding to a change remaining in $(0,1)$ and
symbol~``$\translabelchge$'' corresponding to a change where the fractional part either
starts or ends at $0$.
The second component is either $\silentaction$ if the transition represents time elapsing
in the TA, or provides the action that labels the corresponding discrete transition in the TA.

Given a state $\TTSstate = (\loc, \clockval)$, and $\paramd \in \setRgeqzero$, we write $\TTSstate + \paramd $ to denote $(\loc, \clockval + \paramd)$.
Given two regions $r$ and $r'$, we write $\region \cup \region'$ for $ \set{\TTSstate \mid \TTSstate \in \region \mbox{ or } \TTSstate \in \region'}$.
\begin{definition}[Labelled Region Automaton]
	For a given TA~$\TA$,
	the labelled region automaton~$\regaut{\TA}$ is given by the tuple
	$(\regset{\TA}, \RegAutActions, \RegAutTransitions)$ where:
	\begin{enumerate}%
		\item $\regset{\TA}$ is the set of states,
		\item $\RegAutActions = \{\translabelinstant, \translabelsame, \translabelchge\} \times (\ActionSet \cup \set{\silentaction})$,
		\item  given two regions $\region, \region' \in \regset{\TA}$ and $\zeta \in \RegAutActions$, we have $(\region,\zeta,\region') \in \RegAutTransitions$ if there exist $s =(\loc, \clockval) \in \region$ and $ s'=(\loc', \clockval') \in \region'$ such that one of the following holds:
		\begin{enumerate}%
			\item $\zeta = (\translabelinstant, \action)$ and
			$(\loc, \clockval) \transitionWith{\edge} (\loc', \clockval') \in \TTStransition $ in $\TTS_{\TA}$
			with $\edge = (\loc, \guard, \action, \resets, \loc')$
			for some $\guard$ and~$\resets$;
		\item $\zeta = (\translabelsame,\silentaction)$ and  $\exists \paramd \in \setRpositif$ such that
		\begin{oneenumerate}
			\item $\TTSstate \transitionWith{\paramd} \TTSstate'$,
			\item $\forall  0< \paramd' < \paramd$,
			$\TTSstate+d' \in r \cup r'$ and
			\item $\fract{\clockval(\clockextra)} \neq 0$ and $\fract{\clockval'(\clockextra)} \neq 0$;\footnote{%
				Condition~(ii) ensures that we only move from one region to the ``next'' one (no intermediate region), and condition~(iii) adds that we stay in the same region for~$\clockextra$ (changing the region for~$\clockextra$ is handled in item~(c)).}
		\end{oneenumerate}
			\item $\zeta = (\translabelchge, \silentaction)$ and
			 $\exists \paramd \in \setRpositif$ such that
 \begin{oneenumerate}
	\item $\TTSstate \transitionWith{\paramd} \TTSstate'$,
	\item $\forall 0< \paramd' < \paramd$, $\TTSstate+\paramd' \in \region \cup \region'$, and
	\item $\fract{\clockval'(\clockextra)}=0$  iff $\fract{\clockval(\clockextra)}\neq 0$.
 \end{oneenumerate}
		\end{enumerate}%
	\end{enumerate}%
\end{definition}

We write $\region \RegAutTransitionWith{\zeta}  \region'$ for $(\region,\zeta,\region') \in \RegAutTransitions$.

In the remainder, we will refer to labelled region automata as region automata to alleviate notation.

\subsection{Execution-time opacity of a TA}

Let us now recall from~\cite{ALLMS23} the notions of private and public runs.
\paragraph{Durations}
Given a \TAtext{}~$\TA$ and a run~$\run$, we say that $\locpriv$ is \emph{visited on the way to a final location in~$\run$} when $\run$ is of the form
\((\loci{0}, \clockval_0), (\paramd_0, \edgei{0}), (\loci{1}, \clockval_1), \ldots, (\loci{m}, \clockval_m),$ $(\paramd_m, \edgei{m}), \ldots, (\loci{n}, \clockval_n)\)
\noindent{}for some~$m,n \in \setN$ such that $\loci{m} = \locpriv$ and $\loci{n} \in \LocFinalSet$.
We denote by $\PrivVisit{\TA}$ the set of those runs, and refer to them as \emph{private} runs.
We denote by $\PrivDurVisit{\TA}$ the set of all the durations of these runs.

Conversely, we say that
$\locpriv$ is \emph{avoided on the way to a final location in~$\run$}
when $\run$ is of the form
\((\loci{0}, \clockval_0), (\paramd_0, \edgei{0}), (\loci{1}, \clockval_1), \ldots, (\loci{n}, \clockval_n )\)
\noindent{}with $\loci{n} \in \LocFinalSet$ and $\forall 0 \leq i < n, \loci{i} \neq \locpriv$.
We denote the set of those runs by~$\PubVisit{\TA}$, referring to them as \emph{public} runs,
and by $\PubDurVisit{\TA}$ the set of all the durations of these public runs.

These concepts can be seen as the set of execution times from the initial location~$\locinit$ to a final location while
visiting (resp.\ not visiting) the private location~$\locpriv$.

\begin{example}\label{example:private-public}
	Consider the following two runs of the TA~$\TAi{1}$ in \cref{fig:aut:non-strategy-opaque}.
	Note that we use $(\locinit, \cdot)$ as a shortcut for $(\locinit, \clockval)$ such that $\clockval(\clock) = \cdot$.
\begin{align*}
	\run_1 = &~(\locinit, 0), (1, \edge_3), (\locinit, 0), (0, \edge_1), (\locpriv, 0), (0, \edge_2), \\ &(\locfinal, 0) \\
	\run_2 = &~(\locinit, 0), (0.1, \edge_6), (\loci{2}, 0), (0, \edge_7), (\loci{3}, 0),\\ & (0.8, \edge_8), (\locfinal, 0.8)
\end{align*}
Run $\run_1 \in \PrivVisit{\TAi{1}}$ is a private run, and $\runduration{\run_1}= 1 \in \PrivDurVisit{\TAi{1}}$.
Run $\run_2  \in \PubVisit{\TAi{1}}$ is a public run with $\runduration{\run_2} = 0.9 \in \PubDurVisit{\TAi{1}}$.%
\end{example}
\begin{definition}[\FullOpacityText{}]\label{def:full-weak-opacity}
	A \TAtext{}~$\TA$ is \emph{\fullOpaqueText{}} when $\PrivDurVisit{\TA} = \PubDurVisit{\TA}$.
\end{definition}

That is, if for any run of duration~$\paramd$ reaching a final location after visiting~$\locpriv$, there exists another run of the same duration reaching a final location but not visiting~$\locpriv$, and vice versa, then the TA is \fullOpaqueText{}.
\begin{example}\label{example:full-weak-opacity}
	Consider again $\TAi{1}$ in \cref{fig:aut:non-strategy-opaque}.
	Each time $\clock$ equals~$1$, we can reset it via~$\mathsf{e_3}$, and take the edges $\mathsf{e_1}$ and~$\mathsf{e_2}$ instantaneously.
	It results that $\PrivDurVisit{\TAi{1}} = \setN$.
	We have seen in \cref{example:private-public} that $0.9 \in \PubDurVisit{\TAi{1}}$.
	So $\PrivDurVisit{\TAi{1}} \neq \PubDurVisit{\TAi{1}}$ and $\TAi{1}$ is not \fullOpaqueText{}.
\end{example}
\section{Problem: Controlling TA to achieve \opacityText{}}\label{section:problem}
	Let us formally define the main problem addressed in this work.
We assume $\ActionSet = \ActionSetC \uplus \ActionSetU$ where $\ActionSetC$ (resp.\ $\ActionSetU$) denotes controllable (resp.\ uncontrollable) actions.
The uncontrollable actions are always available, whereas the controllable actions can be enabled and disabled at runtime. 

The controller has a \emph{strategy}, \ie{} a function $\strategy : \setRgeqzero \rightarrow 2^{\ActionSetC}$
which associates to each time a subset of~$\ActionSetC$, denoting that these actions are enabled, while other controllable actions are disabled.

We define the semantics of a controlled TA as follows.
Compared to \cref{definition-TA-semantics}, we also add to the states the absolute time.
\begin{definition}[Semantics of a controlled TA]
	Given a TA \mbox{$\TA = \TAprivextend$} and a strategy $\strategy : \setRgeqzero \rightarrow 2^{\ActionSetC}$, the \emph{semantics of the TA~$\TA$ controlled by strategy~$\strategy$} is given by
	\mbox{$\semanticscontrolextend$ with}
	\begin{enumerate}
		\item $\TTSStates = \big\{ (\loc, \clockval, \globaltime) \in \LocSet \times \setRgeqzero^{\ClockCard} \times \setRgeqzero \mid \clockval \models \invariantof{\loc} \big\}$,
		\item $\TTSinit = (\locinit, \ClocksZero, 0) $,
		\item  $\TTStransitionControl$ consists of the discrete and (continuous) delay transition relation:
		\begin{enumerate}
			\item discrete transitions: $
			(\loc,\clockval, \globaltime) \transitionWithControl{\edge} (\loc',\clockval', \globaltime)$,
 			if $(\loc, \clockval, \globaltime) , (\loc',\clockval', \globaltime) \in \TTSStates$ and there exists $\edge = (\loc,\guard,\action,\resets,\loc') \in \EdgeSet$ such that $\clockval'= \reset{\clockval}{\resets}$,  $\clockval\models \guard$, and $\action \in \strategy(\tau) \cup \ActionSetU$
 			(that is, $\action$ is either enabled by the strategy at time~$\tau$, or uncontrollable);
			\item delay transitions: $(\loc,\clockval, \globaltime) \transitionWithControl{\paramd} (\loc, \clockval+\paramd, \globaltime+\paramd)$, with $\paramd \in \setRgeqzero$, if
			$\forall  \paramd' \in
			\setRpositif$ such that $\paramd' < \paramd$,
			$(\loc, \clockval+\paramd', \globaltime + \paramd') \in \TTSStates$.
		\end{enumerate}
	\end{enumerate}
\end{definition}
We write $(\loc, \clockval, t) \xrightarrow{\paramd,e}_{\strategy} (\loc', \clockval', t')$ for a combination of a delay and a discrete transitions when
$\exists \clockval''$ such that $ (\loc, \clockval, t) \transitionWithControl{\paramd} (\loc, \clockval'', t') \transitionWithControl{\edge} (\loc', \clockval', t')$.

A run $\run =  (\loc_0, \clockval_0), (\timeEdge{0}), \ldots, (\loc_n, \clockval_n)$ is \emph{\compatible} when, $\forall 0 \leq i < n$, it holds that $(\loc_i, \clockval_i, \sum_{j < i} \paramd_j) \transitionControlLabel{i}{\strategy} (\loc_{i+1}, \clockval_{i+1}, \sum_{j \leq i} \paramd_j)$.

We let
$\PrivVisitStrat{\TA}{\strategy}$ the set of private and \compatible{} runs,
$\PubVisitStrat{\TA}{\strategy}$ the set of public and \compatible{} runs,
$\PrivDurVisitStrat{\TA}{\strategy}$ the set of durations of private and  \compatible{} runs, and
$\PubDurVisitStrat{\TA}{\strategy}$ the set of durations of public and \compatible{}  runs.

\begin{definition}[\FullOpacityText{} with strategy]\label{def:opacity-strategy}
	Given a strategy~$\strategy$, a TA~$\TA$ is \emph{\fullOpaqueText{} with~$\strategy$} whenever $\PrivDurVisitStrat{\TA}{\strategy} = \PubDurVisitStrat{\TA}{\strategy}$.
\end{definition}
\begin{figure}
	{\centering
		\begin{tikzpicture}[ pta,  scale=.9, every node/.style={scale=1}]
			\node[location, initial, initial text =] (q0) {$\locinit$};
			\node[location, private] at (2.5,2) (qpriv) {$\locpriv$};f
			\node[location, final] at (2.5,-2) (qfin) {$\locfinal$};
			\node[invariant] at (q0.north) [yshift=.5em] (invariant) {$\styleclock{x} \leq 1$};

			\path[->] (q0) edge[loop below] node[below,align=center] {\shortstack{ ${\edgeicolor{1}}$ \\ $\styleclock{x} = 1$ \\$\styleact{u}$ \\$\styleclock{x} \assign 0 $}} (q0)
			edge node[above, sloped, align=center] { \shortstack{${\edgeicolor{2}}$ \\ $\styleclock{x} = 0$}} node[below, sloped] {$\styleact{u}$} (qpriv)
			edge node[below, sloped] {$\styleact{a}$} node[above, sloped] {${\edgeicolor{4}}$} (qfin)
			(qpriv) edge node[right, align=left] {\shortstack{${\edgeicolor{3}}$\\ $\styleclock{x} = 0$}} node[left] {$\styleact{u}$} (qfin);

		\end{tikzpicture}

	}
	\caption{TA $\exTAopaque$}
	\label{fig:aut:strategy-opaque}
\end{figure}
\begin{example}[\opaqueText{} TA]\label{example:opaque-TA}
	Consider the TA~$\exTAopaque$ in \cref{fig:aut:strategy-opaque}.
	Assume $\ActionSetU = \{u\}$ and $\ActionSetC = \{a\}$.
	First consider the strategy~$\strategy_1$ such that
	$\forall \tau \in \setRgeqzero , \strategy_1(\tau) = \{ a \}$, \ie{} $a$ is allowed anytime.
	We have $\PrivDurVisitStrat{\TA}{\strategy_1} = \setN$ while $\PubDurVisitStrat{\TA}{\strategy_1} = \setRgeqzero$.
	Therefore $\PrivDurVisitStrat{\TA}{\strategy_1} \neq \PubDurVisitStrat{\TA}{\strategy_1}$, and hence $\exTAopaque$ is not \fullOpaqueText{} with~$\strategy_1$.

	Now consider the strategy~$\strategy_2$ such that
	\[\strategy_2(\tau) = \left\{
    \begin{array}{ll}
        \{ a \} & \mbox{if } \tau \in \setN \\
        \emptyset & \mbox{otherwise.}
    \end{array}
\right.
	\]

	We now have $\PrivDurVisitStrat{\TA}{\strategy_2} = \PubDurVisitStrat{\TA}{\strategy_2} = \setN$, hence $\exTAopaque$ is \fullOpaqueText{} with~$\strategy_2$.
\end{example}
\begin{example}[Non-\opaqueText{} TA]\label{example:non-strategy-opaque}%
	There is no strategy such that TA~$\TAi{1}$ in \cref{fig:aut:non-strategy-opaque} is \fullOpaqueText{}.
	Recall that $\ActionSetU = \{u\}$ and $\ActionSetC = \{a, b\}$.
	The transitions $\mathsf{e_1}$ and~$\mathsf{e_2}$ are uncontrollable, so we can reach~$\locfinal$ at any integer time along a run visiting~$\locpriv$.
	However the set of public durations will either not contain~$0$, or contain~$\setRgeqzero$.
	Indeed, in order to reach $\locfinal$ with a public run, the system must take transitions associated to both actions $a$ and~$b$.
	As a consequence, to contain the duration~$0$, the strategy must allow $\{a, b\}$ at time~$0$.
	Hence, $\loc_3$ can be reached at time~$0$, and as the next transition~$\mathsf{e_8}$ is uncontrollable, it can be taken at any time.
	Thus, the set of public durations is~$\setRgeqzero$.
\end{example}
\paragraph*{\Finitelyvarying{} strategies}
In the following, to match the fact that a \metaStrategy{} has a finite number of strategy changes between two integer time instants, we only consider strategies that behave in a ``reasonable'' way.
We thus only consider \emph{\finitelyvarying{} strategies}, in which the number of changes are finite for any closed time interval.

Indeed, we can assume that a controller cannot change infinitely frequently its strategy in a finite time: it is unrealistic to consider, in a bounded interval, neither a system that can perform an infinite number of actions, nor a controller that can make an infinite number of choices.
\Finitelyvarying{} strategies are reminiscent of non-Zeno behaviours, in the sense that finitely-varying strategies have a finite number of strategy changes in every bounded interval.

Formally:

\begin{definition}[\Finitelyvarying{} strategy]\label{definition:finitelyvarying}
	A strategy~$\strategy$ is \emph{\finitelyvarying{}} whenever, for any closed time interval~$\interval$,
	there is a finite partition $\subinterval_1, \ldots, \subinterval_n$ of~$\interval$ such that
	each $\subinterval_i$ is an interval within which $\strategy$ makes the same choice.
	That is, for all $\tau_1, \tau_2 \in \subinterval_i$, $1 \leq i \leq n$, $\strategy(\tau_1) = \strategy(\tau_2)$.
\end{definition}
\begin{example}[\NonFinitelyvarying{} strategy]\label{ex:non-reasonable}
	Let $\exTANFV$ be the automaton in \cref{fig:aut:rationel}, with a global invariant $x,y \leq 1$.
	Let $\ActionSetU = \{u\}$ and $\ActionSetC = \{a, b\}$.
	Since transition $\mathsf{e}_4$ is uncontrollable, for any strategy~\strategy, $(0,1) \subseteq \PubDurVisitStrat{\TA}{\strategy}$. To ensure a \fullOpaqueText{} system, all those durations need to be also enabled through the private state. This implies that, for every $\alpha < 1$, both $a$ and $b$ should be enabled at a time instant before $\alpha$. Now if we consider the bottom path, if $a$ and $b$ are enabled simultaneously, then state $\loc_3$ is reachable and it is possible to reach the final state when $x =1$. Then we need to prevent $b$ to be enabled at the same time as $a$. 
	We can build a strategy~$\strategy$ to make~$\exTANFV$ \fullOpaqueText{} with~$\strategy$.
	This strategy (illustrated in red boxes in \cref{fig:aut:rationel}) is defined by
	\[\strategy(\tau) = \left\{
    \begin{array}{ll}
        \{ a \} & \mbox{if } \tau \in \setQgeqzero \\
        \{ b \} & \mbox{if } \tau \notin \setQgeqzero \\
    \end{array}
\right.
	\]

	With strategy $\strategy$, it means that the order in which actions $a$ and $b$ can be performed is not fixed, while preventing them from being performed at the same time. This is only possible with a \nonFinitelyvarying{} strategy, where the number of changes is infinite in a given interval. 

\begin{figure}[tb]
	\begin{center}
		\begin{tikzpicture}
		\begin{scope}[pta,  scale= 1, every node/.style={scale=1}]
			\node[style={rounded corners,draw=blue!80,fill=red!20,inner sep = 4pt},xshift=0ex,yshift=0pt,text width=1.3cm,draw] (st1) at (0,2) {$\styleclock{x} \in \setQgeqzero$};
			\node[style={rounded corners,draw=blue!80,fill=red!20,inner sep = 4pt},xshift=0ex,yshift=0pt,text width=1.3cm,draw] (st2) at (3,3) {$\styleclock{x} \notin \setQgeqzero$};
			\draw[->, blue!80] (.5,1) -- (st1.south);
			\draw[->, blue!80] (3,2) -- (st2.south);
			\node[location, initial, initial text =] (q0) {$\locinit$};
			\node[location, private] at (2,1.5) (qp) {$\locpriv$};
			\node[location] at (4,1.5) (q1) {$\loci{1}$};
			\node[location, final] at (6,0) (qf) {$\locfinal$};
			\node[location] at (2,-1.5) (q2) {$\loci{2}$};
			\node[location] at (4,-1.5) (q3) {$\loci{3}$};
			\node[invariant] at (5.5,3) {$x,y \leq 1$}; 
			\path[->]  (q0) edge node[below, sloped] {$\styleact{a}$} node[above, sloped] {\shortstack{${\edgeicolor{1}}$ \\ $\styleclock{x} > 0$}} (qp)
			(qp) edge node[above] {${\edgeicolor{2}}$} node[below] {$\styleact{b}$}  (q1)
			(q1) edge node[above, sloped]  {\shortstack{${\edgeicolor{3}}$ \\ $\styleclock{x} < 1$}}node[below, sloped] {$\styleact{u}$} (qf)
			(q0) edge node[above, sloped]  {${\edgeicolor{5}}$} node[below, sloped]{\shortstack{$\styleact{a}$ \\ $\styleclock{y} \assign 0$}} (q2)
			(q2) edge node[above, sloped]  {\shortstack{${\edgeicolor{6}}$ \\ $\styleclock{y} = 0$}} node[below, sloped]{$\styleact{b}$}(q3)
			(q3) edge node[above, sloped] {\shortstack{${\edgeicolor{7}}$ \\ $\styleclock{x} = 1$}} node[below, sloped] {$\styleact{u}$} (qf)
			(q0) edge node[above, sloped] {\shortstack{${\edgeicolor{4}}$ \\ $0<\styleclock{x}<1$}} node[below, sloped] {$\styleact{u}$} (qf);
		\end{scope}
	\end{tikzpicture}
	\caption{Automaton $\exTANFV$}
	\label{fig:aut:rationel}

\end{center}

\end{figure}
\end{example}
\paragraph*{\MetaStrategies}
As defined previously, strategies have infinite precision in determining when to activate or deactivate controllable actions.
While it might seem reasonable to assume precise knowledge of when a clock tick occurs, in practice, achieving such infinite precision between integer time points is not feasible.
Therefore, the idea is to group together all strategies that enable the same actions at exact integer times and permit subsets of these actions---maintaining the same order---between those times.

\begin{definition}[\MetaStrategy{}]\label{definition:metastrategy}
	A \metaStrategy{} $\metastrat{}$ is a partial function on integer bounded intervals, such that:
	\begin{itemize}
		\item for all $n \in \setN$,  $\metastrat([n,n]) \in 2^{\ActionSetC}$,
		\item for all $n \in \setN$, $\metastrat((n,n+1))  \in (2^{\ActionSetC})^*$
	\end{itemize}

\end{definition}

In other words, a \metaStrategy{} is a function which associates to each integer time a set~$\substrat \in 2^{\ActionSetC}$ of enabled actions, and to each open interval between two consecutive integers a finite sequence $(\substrat_1, \ldots, \substrat_m) \in (2^{\ActionSetC})^*$ giving the order in which sets of actions are allowed in the system.

\begin{definition}[Ordered partition]
	Given an interval~$\interval$, we call \emph{ordered partition of~$\interval$} a finite sequence of disjoint intervals $\subinterval_1, \dots, \subinterval_n$ such that:
	\begin{enumerate}
		\item $\bigcup_{j=1}^{n}\subinterval_j = \interval$,
		\item $\subinterval_1$ is left open, $\subinterval_n$ is right open,
		\item each $\subinterval_j$ is non empty,
			and
		\item for all $j \in \{1, \dots ,n-1\}$, the right boundary of $\subinterval_j$ is the left boundary of $\subinterval_{j+1}$.
	\end{enumerate}
\end{definition}

\begin{example}
The sequence $(0,0.2], (0.2,0.42),$ $[0.42, 0.42], (0.42,0.83], (0.83,1)$ is an ordered partition of the interval $\interval = (0,1)$.
\end{example}

\begin{definition}[\MetaStrategy{} satisfaction]
	A strategy $\strategy{}$ is said to \emph{satisfy} a \metaStrategy{} $\metastrat{}$, denoted $\strategy \satisfies \metastrat$, when:
	\begin{itemize}
		\item for all $n \in \setN$, $\strategy(n)= \metastrat([n,n])$,
		\item for all  $n \in \setN$, denoting $\interval = (n,n+1)$ and $\metastrat(\interval) = \substrat_1, \ldots, \substrat_m$ %
		 there is an ordered partition  $\subinterval_1,  \ldots, \subinterval_m$ of~$\interval$ such that for all $\tau \in \subinterval_i$, $1 \leq i \leq m$,  $\strategy(\tau) = \substrat_i$.
	\end{itemize}

\end{definition}

The set of private (or public) durations for a \metaStrategy{} is defined as the union of private (or public) durations of all the strategies it represents:
$\PrivDurVisitStrat{\TA}{\metastrat} = \{\tau \mid \exists \strategy \text{ s.t.\ } \strategy \satisfies \metastrat \text{ and } \tau \in \PrivDurVisitStrat{\TA}{\strategy}\}$ and $\PubDurVisitStrat{\TA}{\metastrat} = \{\tau \mid \exists \strategy \text{ s.t.\ } \strategy \satisfies \metastrat \text{ and } \tau \in \PubDurVisitStrat{\TA}{\strategy}\}$.

\begin{definition}[\FullOpacityText{} with a \metaStrategy{}]\label{def:opacity-metastrategy}
	Given a \metaStrategy{} $\metastrat$, a TA~$\TA$ is \emph{\fullOpaqueText{} with~$\metastrat$} whenever
	$\PrivDurVisitStrat{\TA}{\metastrat} = \PubDurVisitStrat{\TA}{\metastrat}$.
\end{definition}

\cref{definition:metastrategy,definition:finitelyvarying} immediately give a correspondence between \finitelyvarying{} strategies and \metaStrategies{}:

\begin{lemma}
	For any \finitelyvarying{} strategy $\strategy{}$, there exists a unique \metaStrategy{} $\metastrat{}$ such that $\strategy \satisfies \metastrat$.

	For any \metaStrategy{} $\metastrat{}$, there exists a (\finitelyvarying{}) strategy~$\strategy{}$ such that $\strategy \satisfies \metastrat$.
\end{lemma}
\paragraph*{Problems}\label{paragraph:problems}

In this paper, we are interested in several \opacityText{} control problems, \ie{} related to a \mbox{(meta-)}strategy making the corresponding controlled TA~$\TA$ \opaqueText{}.

The first problem we are interested in will be the \emph{existence} of a strategy enforcing the \opacityText{}.

\defProblem
	{\FullOpacityText{} strategy emptiness problem}
	{A TA $\TA$}
	{Decide the emptiness of the set of strategies~$\strategy$ such that $\TA$ is \fullOpaqueText{} with~$\strategy$.}

Because we will show that this problem is undecidable (\cref{theorem:undecidability}), even when restricted to \finitelyvarying{} strategies, we define a variant of the setting, and consider the existence of \metaStrategies{} instead.
We will consider both the \emph{existence} of a \metaStrategy{} enforcing the \opacityText{}, or the \emph{synthesis} of such a \metaStrategy{}.

\defProblem
	{\FullOpacityText{} \metaStrategy{} emptiness problem}
	{A TA $\TA$}
	{Decide the emptiness of the set of \metaStrategies{} $\metastrat$ such that $\TA$ is \fullOpaqueText{} with~$\metastrat$.}
\defProblem
	{\FullOpacityText{} \metaStrategy{} synthesis problem}
	{A TA $\TA$}
	{Synthesize a \metaStrategy{} $\metastrat$ such that $\TA$ is \fullOpaqueText{} with~$\metastrat$.}
\section{Undecidability of the \fullOpacityText{} strategy emptiness problem}\label{section:undecidability}

In this section, we show that the \fullOpacityText{} strategy  emptiness problem is undecidable, vindicating the reliance on meta-strategies.
To do so, we will use a reduction from the termination of the two-counter Minsky machine problem---which is known to be undecidable~\cite{Minsky67}.

Let us start with a quick recall of Minsky machines.
A two-counter Minsky machine~$\twoCM$ is described by two counters $\counterOne$ and~$\counterTwo$, as well as a sequence of commands $\twoCMcommand_0,\dots,\twoCMcommand_m$ where $\twoCMcommand_0$ is the
starting command, $\twoCMcommand_m$ ends the run of the system, and every command $\twoCMcommand_0$ to~$\twoCMcommand_{m-1}$ is of one of the following three types:
\begin{itemize}
\item increment counter $\counteri \in \{\counterOne,\counterTwo\}$, move to the next command
\item decrement counter $\counteri \in \{\counterOne,\counterTwo\}$, move to the next command (note that the machine must be designed so that this command cannot occur when counter~$\counteri$ is equal to~0)
\item if counter $\counteri \in \{\counterOne,\counterTwo\}$ is equal to~0, move to command~$\twoCMcommand_k$, otherwise move to command~$\twoCMcommand_j$.
\end{itemize}

In summary, the machine goes through a list of commands starting with~$\twoCMcommand_0$, incrementing, decrementing counters, or testing whether a counter is equal to~0 in order to select the new command to jump to---and it terminates whenever it reaches~$\twoCMcommand_m$.
The \emph{termination problem} for two-counter Minsky machines consists in deciding, given a machine~$\twoCM$, whether $\twoCM$ terminates.

\begin{theorem}\label{theorem:undecidability}
The \fullOpacityText{} strategy emptiness problem is undecidable.
\end{theorem}
\begin{proof}
Let $\twoCM$ be a Minsky machine over two counters $\counterOne$ and~$\counterTwo$, described by the commands $\twoCMcommand_0,\dots, \twoCMcommand_m$.
We will build a TA~$\TA$ such that there exists a strategy $\strategy$ enforcing
\fullOpacityText{} of~$\TA$ iff $\twoCM$ does \emph{not} terminate.

\paragraph{Overall intuition of the encoding}
Intuitively, in order to be opaque, the TA~$\TA$ coupled with a strategy $\strategy$ will
have to correctly emulate the behaviour of the Minsky machine, with the strategy's choices depending 
on the values of the counters.
More precisely, each command of~$\twoCM$ will take exactly three time units to be handled by~$\TA$, and therefore the $i$-th step of the machine corresponds to the interval $(3i,3i+3]$.
Considering this interval modulo~3, $(0,1]$ is dedicated to actions associated to counter~$\counterOne$,
$(1,2]$ is dedicated to actions associated to counter~$\counterTwo$,
and 
$(2,3]$ is used to detect whether the strategy makes the system \fullOpaqueText{}; in particular, we have that no private run can reach the target within this period of time.
Due to the periodic nature of the system, the intervals $(0,1]$, $(1,2]$ and $(2,3]$ should be understood modulo~3.

Let us give an idea of how counter~$\counterOne$ is represented within the first interval. More detailed explanation as well as the corresponding \TA{} used are given later in this section. In our model, whenever~$\counterOne$ has the value~$k$, then $k$ public runs will reach the final location at different times during this interval. 
In order to ensure opacity, the strategy needs to allow the 
action~$\textstyleact{a_{\counterOne}}$ at those $k$ instants, as this action
produces a private run that immediately reaches the final state. However it also
produces a public run that will reach the final state three time units later.
Hence, the number of public runs reaching the final location during the next 
$(0,1]$ interval remains the same (ignoring the potential impact of other actions), 
thus preserving the value of the counter as well.

Incrementing the counter can then be done by forcing an action which will add a new public duration three times unit later, while decrementing the counter is done by producing a private run immediately reaching the target, and thus removing the need for one
$\textstyleact{a_{\counterOne}}$.
The test command can be ensured by first requiring, at time~$0$, the controller to claim (allowing either the action \textstyleact{=_0} or \textstyleact{\neq_0}) whether counter~$\counterOne$ is zero or not.
Then, by observing whether action $\textstyleact{a_{\counterOne}}$ occurs or not within the following interval. 

The actions on counter~$\counterTwo$ are the same, only occurring within the intervals~$(1,2]$ and thus will not be detailed fully in the proof.

Finally, if a violation
is made (for instance by claiming the counter is zero when it is not, or by refusing to play an action linked to an increment or decrement of a counter), a location can be reached from which one can go to the final location at any time. This means that any duration beyond this point corresponds to a public run.
This violates opacity as, as mentioned earlier, no private
run will be able to reach the final location within the intervals~$(2,3]$.
Reaching the final command~$\twoCMcommand_m$ will similarly trigger a violation of opacity, hence
the only way for a strategy to ensure opacity is to properly emulate the Minsky machine 
but fail to reach $\twoCMcommand_m$, and thus for the Minsky machine not to terminate.

\newcommand{\gadgetOneAct}{\ensuremath{\mathit{1act}}}
\newcommand{\gadgetCone}{\ensuremath{G_{\counterOne}}}
\newcommand{\gadgetCtwo}{\ensuremath{G_{\counterTwo}}}
\newcommand{\gadgetIfCone}{\ensuremath{\mathit{if}_{\counterOne}}}

\newcommand{\gadgetIncCone}{\ensuremath{\textit{inc}_{\counterOne}}}
\newcommand{\gadgetIncCtwo}{\ensuremath{\textit{inc}_{\counterTwo}}}
\newcommand{\gadgetDecCone}{\ensuremath{\textit{dec}_{\counterOne}}}
\newcommand{\gadgetDecCtwo}{\ensuremath{\textit{dec}_{\counterTwo}}}

\newcommand{\actIncCone}{\ensuremath{\textstyleact{\gadgetIncCone{}}}}
\newcommand{\actIncCtwo}{\ensuremath{\textstyleact{\gadgetIncCtwo{}}}}
\newcommand{\actDecCone}{\ensuremath{\textstyleact{\gadgetDecCone{}}}}
\newcommand{\actDecCtwo}{\ensuremath{\textstyleact{\gadgetDecCtwo{}}}}

\paragraph{Gadgets and actions}
Our TA will be built by combining several smaller timed automata fragments called ``gadgets''.
We first describe the construction gadget per gadget, and then explain how they should be combined. As the strategy's choice is based only on time elapsed, it does not know in which gadget the run is, and thus must assume it might be in any of them, ensuring opacity in all cases.

For all the gadgets, we fix a set of actions 
$\ActionSet=\{\textstyleact{u},\textstyleact{a_{\counterOne}},\textstyleact{a_{\counterTwo}},\actIncCone, \actIncCtwo, \actDecCone,\actDecCtwo,\textstyleact{=_0},\textstyleact{\neq_0}\}$ with~$\textstyleact{u}$ being the only uncontrollable action.
The TA~$\TA$, and a fortiori the different gadgets, will rely on a single clock~\textstyleclock{x}.
In particular, we do not rely on the extra ``tick'' clock~\textstyleclock{z} that is assumed throughout the rest of this document.

\paragraph{Gadget \gadgetOneAct{}}
We first describe a gadget preventing the controller from allowing two actions simultaneously.
Formally, the gadget \gadgetOneAct{} is the TA (see \cref{fig:aut:1Act})
$\TA_{\gadgetOneAct} = (\ActionSet, \LocSet_{\gadgetOneAct}, \locinit^{\gadgetOneAct}, \locpriv, \locfinal, \{\textstyleclock{x}\}, \invariant_{\gadgetOneAct}, \EdgeSet_{\gadgetOneAct}) $ where:
\begin{itemize}
\item $\LocSet_{\gadgetOneAct} =\{\locinit^{\gadgetOneAct}, \loc_e,\locfinal\}\cup\{\loc_{\textstyleact{v}}\mid \textstyleact{v}\in \ActionSet\}$,
\item $\invariant_{\gadgetOneAct}(\loc)=\true$ for all $\loc \in \LocSet_{\gadgetOneAct}$,
\item $\EdgeSet_{\gadgetOneAct} = \big\{(\locinit^{\gadgetOneAct}, \true, \textstyleact{v}, \{\textstyleclock{x}\}, \loc_v)\mid \textstyleact{v}\in \ActionSet \big\} \cup
\big\{(\loc_{\textstyleact{v}}, \textstyleclock{x}=0, \textstyleact{v'}, \emptyset, \loc_e)\mid \textstyleact{v},\textstyleact{v'}\in \ActionSet, \textstyleact{v}\neq \textstyleact{v'}
\big\}
\cup \big\{(\loc_e,\true,u,\emptyset,\locfinal \big\}
$.
\end{itemize}

In the \gadgetOneAct{} gadget, whenever an action~$\textstyleact{v}$ is allowed, the system can go to a location~$\loc_{\textstyleact{v}}$ via \textstyleedge{1}, resetting~$\textstyleclock{x}$; then, if any action other than~$\textstyleact{v}$ is also allowed at the same time (which is tested by requesting $\textstyleclock{x}=0$ in edge~\textstyleedge{2}), then the following location~$\loc_e$ can be reached.
From~$\loc_e$ the system can reach the final location at any time via~\textstyleedge{3} via the uncontrollable action~$u$, whatever the controller does.
As a consequence, every duration beyond this point corresponds to a public run.
As previously mentioned, some durations (the intervals $(2,3]$) cannot be achieved by private runs, and therefore opacity is violated.

\begin{figure}
	{\centering
		\begin{tikzpicture}[ pta,  scale=.9, every node/.style={scale=1}]
			\node[location, initial, initial text =] (q0) {$\locinit^{\gadgetOneAct}$};

			\node[location] at (2.4,0) (qa) {$\loc_v$};
			\node[location] at (4.8,0) (qe) {$\loc_e$};
			\node[location, final] at (6.6,0) (qfin) {$\locfinal$};

			\path[->] (q0) 
				edge node[above] {\shortstack{$\edgeicolor{1}$\\ $\styleclock{x} \assign 0$}} node[below] {$\styleact{v}$} (qa)
				(qa) edge node[above] {\shortstack{$\edgeicolor{2}$\\$\styleclock{x} = 0$}} node[below] {$\styleact{v'}\neq \styleact{v}$} (qe)
				(qe) edge node[above] {$\edgeicolor{3}$} node[below] {$\styleact{u}$} (qfin);

		\end{tikzpicture}
	}
	\caption{\gadgetOneAct{} gadget on a generic action $\textstyleact{v}$}
	\label{fig:aut:1Act}
\end{figure}

\paragraph{Gadgets $\gadgetCone{}$ and $\gadgetCtwo{}$}
We now introduce the gadget $\gadgetCone{}$ (resp.\ $\gadgetCtwo{}$) which forces, barring intervention by other gadgets, the strategy to repeat the same behaviour within the intervals $(0,1)$ (resp.\ $(1,2)$).
Formally, the gadget $\gadgetCone{}$ is the TA (see \cref{fig:aut:Gx})
$\TA_{\gadgetCone{}} = (\ActionSet, \LocSet_{\gadgetCone{}}, \locinit^{\gadgetCone{}}, \locpriv, \locfinal, \{\textstyleclock{x}\}, \invariant_{\gadgetCone{}}, \EdgeSet_{\gadgetCone{}}) $ where:
\begin{itemize}
\item $\LocSet_{\gadgetCone{}} =\{\locinit^{\gadgetCone{}},\loc_{\counterOne},\locpriv, \locfinal\}$,
\item $\invariant_{\gadgetCone{}}(\loc)=\true$ for all $\loc \in \LocSet_{\gadgetCone{}}$,
\item $\EdgeSet_{\gadgetCone{}} =
\big\{(\locinit^{\gadgetCone{}}, \textstyleclock{x}=3, \textstyleact{u}, \{\textstyleclock{x}\}, \locinit^{\gadgetCone{}}),$ \\$
(\locinit^{\gadgetCone{}}, 0<\textstyleclock{x}<1, \textstyleact{a_{\counterOne}}, \{\textstyleclock{x}\}, \loc_{\counterOne}),$ \\$
(\loc_{\counterOne}, \textstyleclock{x}=0, \textstyleact{u}, \emptyset, \locpriv),$ $
(\locpriv, \textstyleclock{x}=0, \textstyleact{u}, \emptyset, \locfinal),$ $
(\loc_{\counterOne}, \textstyleclock{x}=3, \textstyleact{u}, \emptyset, \locfinal)
\big\}$.
\end{itemize}

\begin{figure*}
	\begin{subfigure}{0.45\textwidth}
	{\centering
		\begin{tikzpicture}[ pta,  scale=.9, every node/.style={scale=1}]
			\node[location, initial, initial text =] (q0) {$\locinit^{\gadgetCone{}}$};
			\node[location] at (0,-2) (qcounter) {$\loc_{\counterOne}$};
			\node[location,private] at (2,-3) (qpriv) {$\locpriv$};
			\node[location, final] at (4,-2) (qfin) {$\locfinal$};

			\path[->] (q0) 
			edge[loop above] node[above] {\shortstack{$ \edgeicolor{1}, \styleact{u}$ \\$\styleclock{x} = 3$\\$\styleclock{x} \assign 0$}} (q0)
			edge node[left] {\shortstack{$\edgeicolor{2}, \styleact{a_{\counterOne}}$ \\ $0<\styleclock{x}<1$ \\ $\styleclock{x} \assign 0$}} (qcounter)
			(qcounter) edge[bend right=20] node[above, sloped] {$\edgeicolor{3}, \styleact{u}$} node[below, sloped] { $\styleclock{x} = 0$ } (qpriv)
			edge node[above,sloped] {\shortstack{${\edgeicolor{4}}, \styleact{u}$ \\$\styleclock{x} = 3$}} (qfin)
			(qpriv) edge[bend right=20] node[above, sloped] {$\edgeicolor{5}, \styleact{u}$} node[below, sloped]  { $\styleclock{x} = 0$}  (qfin);
		\end{tikzpicture}

	}
	\caption{Gadget $\gadgetCone{}$}\label{fig:aut:Gx}
	\end{subfigure}
	\hfill
	\begin{subfigure}{0.45\textwidth}
	{\centering
		\begin{tikzpicture}[ pta,  scale=.9, every node/.style={scale=1}]
			\node[location, initial, initial text =] (q0) {$\locinit^{\gadgetCtwo{}}$};
			\node[location] at (0,-2) (qcounter) {$\loc_{\counterTwo}$};
			\node[location,private] at (2,-3) (qpriv) {$\locpriv$};
			\node[location, final] at (4,-2) (qfin) {$\locfinal$};

			\path[->] (q0) 
			edge[loop above] node[above] {\shortstack{$\edgeicolor{1}, \styleact{u}$ \\$\styleclock{x} = 3$\\$\styleclock{x} \assign 0$ }} (q0)
			edge node[left] {\shortstack{$\edgeicolor{2}, \styleact{a_{\counterTwo}}$ \\ $1<\styleclock{x}<2$ \\$\styleclock{x} \assign 0$ }} (qcounter)
			(qcounter) edge[bend right=20] node[above, sloped] {$\edgeicolor{3}, \styleact{u}$} node[below, sloped] { $\styleclock{x} = 0$} (qpriv)
			edge node[above,sloped] {\shortstack{${\edgeicolor{4}}, \styleact{u}$ \\$\styleclock{x} = 3$}}  (qfin)
			(qpriv) edge[bend right=20] node[above, sloped] {$\edgeicolor{5}, \styleact{u}$} node[below, sloped]  { $\styleclock{x} = 0$ }  (qfin);
		\end{tikzpicture}

	}
	\caption{Gadget $\gadgetCtwo{}$}\label{fig:aut:Gy}
		\end{subfigure}

	\caption{Gadgets $\gadgetCone{}$ and~$\gadgetCtwo{}$}
\end{figure*}

Gadget~$\gadgetCtwo{}$ only differs by one transition (see \cref{fig:aut:Gy}), moving the impact of allowing the action~$\textstyleact{a_{\counterTwo}}$ to the interval~$(1,2)$.
Note that, since the intervals in which $\textstyleact{a_{\counterOne}}$ and~$\textstyleact{a_{\counterTwo}}$ have an effect are disjoint, we could use a single action for both.
We keep two for ease of understanding.
In the following, we will not present the variations associated to counter~$\counterTwo$.

Assume that public runs are expected to reach the target at times $\tau_1,\dots,\tau_k$ within the interval $(0,1)$ of this step.
Ignoring future gadgets, thanks to~$\gadgetCone{}$,
the strategy can preserve opacity by allowing $\textstyleact{a_{\counterOne}}$ exactly at times $\tau_1,\dots,\tau_k$.
This immediately produce a private run (going via~\textstyleedge{2} then~\textstyleedge{3}).
However, it also produces a public run that will reach the final location 3~time units later (via~\textstyleedge{4}), hence forcing the strategy to  redo the same choices during the next step.

As the number of times $\textstyleact{a_{\counterOne}}$ is allowed during one step of the process represents the value of counter~$\counterOne$, barring external intervention, $\gadgetCone{}$ allows to maintain the value of the counter within the strategy.

\paragraph{Increment gadget}
We now show how external interventions modify the number of times the
strategy must repeat~$\textstyleact{a_{\counterOne}}$ within the interval~$(0,1)$.
We start with the gadget corresponding to a command incrementing counter~$\counterOne$.
More precisely, if command~$\twoCMcommand_i$ is incrementing~$\counterOne$, we build the TA (see \cref{fig:aut:incx}) $\TA_i
= (\ActionSet, \LocSet_i, \locinit^{i}, \locpriv, \locfinal, \{\textstyleclock{x}\}, \invariant_i, \EdgeSet_i) $ where:
\begin{itemize}
\item $\LocSet_{i} =\{\locinit^{i},\loc_1^i,\loc_2^i,\loc_3^i, \locinit^{i+1},\locpriv, \locfinal\}$,
\item $\invariant_{i}(\loc)=true$ for all $\loc \in \LocSet_{i}$,
\item $\EdgeSet_{i} = 
\big\{(\locinit^{i}, \textstyleclock{x}=3, \textstyleact{u}, \{\textstyleclock{x}\}, \locinit^{i+1}),
(\locinit^{i}, 0<\textstyleclock{x}<1,$ $\actIncCone, \emptyset, \loc^i_1),
(\locinit^{i}, 0<\textstyleclock{x}<1, \actIncCone, \{\textstyleclock{x}\}, \loc^i_2),$\\
$(\locinit^{i}, \textstyleclock{x}=1, \textstyleact{u}, \emptyset, \locfinal),
(\loc_1^{i}, \textstyleclock{x}=1, \textstyleact{u},  \{\textstyleclock{x}\}, \locpriv),$\\
$(\locpriv, \textstyleclock{x}=0, \textstyleact{u}, \emptyset, \locfinal),
(\loc_2^{i}, \textstyleclock{x}=3, \textstyleact{u},  \emptyset, \locfinal),
$ \mbox{$(\loc_2^{i}, 0<\textstyleclock{x}<1, \actIncCone,\emptyset, \loc_3^i),
(\loc_3^{i}, \true, \textstyleact{u},  \emptyset, \locfinal)
\big\}$.}
\end{itemize}

\begin{figure*}
	\begin{subfigure}{0.45\textwidth}
	{\centering
		\begin{tikzpicture}[ pta,  scale=.8, every node/.style={scale=.9}]
			\node[location, initial, initial text =] (q0) {$\locinit^{i}$};
			\node[location] at (0,-2) (qi1) {$\loc^i_1$};
			\node[location] at (4,0) (qi2) {$\loc^i_2$};
			\node[location] at (7,0) (qi3) {$\loc^i_3$};
			\node[location] at (0,2) (qnext) {$\locinit^{i+1}$};
			\node[location,private] at (2,-3.5) (qpriv) {$\locpriv$};
			\node[location, final] at (4,-2) (qfin) {$\locfinal$};

			\path[->] (q0) 
			edge node[left] {\shortstack{${\edgeicolor{6}}, \styleact{u}$\\ $\styleclock{x}=3$ \\ $\styleclock{x} \assign 0$}} (qnext)
			edge node[above,sloped] {\shortstack{${\edgeicolor{1}}, \styleact{\gadgetIncCone{}}$\\$0<\styleclock{x}<1$}} node[below,sloped] {$\styleclock{x} \assign 0 $} (qi2)
			edge node[left] {\shortstack{${\edgeicolor{2}}, \styleact{\gadgetIncCone{}}$\\$0<\styleclock{x}<1$}}  (qi1)
			edge[bend right=30] node[above, sloped] {${\edgeicolor{7}},\styleact{u}$} node[below, sloped] {\shortstack{$\styleclock{x}=1$}} (qfin)

			(qi1) edge[bend right=20] node[above, sloped] {${\edgeicolor{3}}, \styleact{u}$} node[below,sloped] {\shortstack{$\styleclock{x} = 1$\\$ \styleclock{x} \assign 0$}} (qpriv)
			(qpriv) edge[bend right=20] node[above, sloped] {${\edgeicolor{4}}, \styleact{u}$} node[below, sloped] {\shortstack{$\styleclock{x} = 0$}} (qfin)
			(qi2) edge node[left] {\shortstack{${\edgeicolor{5}}, \styleact{u}$\\$\styleclock{x} = 3$}}  (qfin)
			edge node[above] {\shortstack{${\edgeicolor{8}}, \styleact{\gadgetIncCone{}}$ \\$0<\styleclock{x}<1$}} (qi3)
			(qi3) edge[bend left=30] node[above, sloped] {${\edgeicolor{9}},\styleact{u}$} (qfin)
			;
		\end{tikzpicture}

	}
	\caption{Gadget $\gadgetIncCone{}$}
	\label{fig:aut:incx}
	\end{subfigure}
	\hfill
	\begin{subfigure}{0.45\textwidth}
	{\centering
		\begin{tikzpicture}[ pta,  scale=.8, every node/.style={scale=.9}]
			\node[location, initial, initial text =] (q0) {$\locinit^{i}$};
			\node[location] at (0,-2) (qi1) {$\loc^i_1$};
			\node[location] at (4,0) (qi2) {$\loc^i_2$};
			\node[location] at (7,0) (qi3) {$\loc^i_3$};
			\node[location] at (0,2) (qnext) {$\locinit^{i+1}$};
			\node[location,private] at (2,-3.5) (qpriv) {$\locpriv$};
			\node[location, final] at (4,-2) (qfin) {$\locfinal$};

			\path[->] (q0)
			edge node[left] {\shortstack{${\edgeicolor{6}}, \styleact{u}$\\ $\styleclock{x}=3$ \\ $\styleclock{x} \assign 0$}} (qnext)
			edge node[above,sloped] {\shortstack{${\edgeicolor{1}}, \styleact{\gadgetDecCone{}}$\\$0<\styleclock{x}<1$}} node[below,sloped] {$\styleclock{x} \assign 0 $} (qi2)
			edge node[left] {\shortstack{${\edgeicolor{2}}, \styleact{\gadgetDecCone{}}$\\$0<\styleclock{x}<1$}}  (qi1)
			edge[bend right=30] node[above, sloped] {${\edgeicolor{7}},\styleact{u}$} node[below, sloped] {\shortstack{$\styleclock{x}=1$}} (qfin)

			(qi1) edge[bend right=20] node[above, sloped] {${\edgeicolor{3}}, \styleact{u}$} node[below,sloped] {\shortstack{$\styleclock{x} = 1$\\$ \styleclock{x} \assign 0$}} (qpriv)
			(qpriv) edge[bend right=20] node[above, sloped] {${\edgeicolor{4}}, \styleact{u}$} node[below, sloped] {\shortstack{$\styleclock{x} = 0$}} (qfin)
			(qi2) edge[out=-30,in=-30] node[below right] {\shortstack{${\edgeicolor{5}}, \styleact{u}$\\$\styleclock{x} = 0$}}  (qpriv)
			edge node[above] {\shortstack{${\edgeicolor{8}}, \styleact{\gadgetDecCone{}}$ \\$0<\styleclock{x}<1$}} (qi3)
			(qi3) edge[bend left=30] node[above, sloped] {${\edgeicolor{9}},\styleact{u}$} (qfin)
			;
		\end{tikzpicture}

	}
	\caption{Gadget $\gadgetDecCone{}$}
	\label{fig:aut:decx}
	\end{subfigure}

	\caption{Increment and decrement gadgets}
\end{figure*}

In this gadget, because of edge~\textstyleedge{7}, a public run will reach the final location at time~1.
The only way to make this time opaque is via edges \textstyleedge{2}, \textstyleedge{3} and~\textstyleedge{4} (which creates a private run of the same duration).
This requires allowing action $\actIncCone$ somewhen in~$(0,1)$.
Because of edge~\textstyleedge{1} however, this additionally creates a public run that will reach the final location 3 time units later and will have to be made opaque thanks to~$\gadgetCone{}$.
At time~3, a run can then go to the initial state of command~$\twoCMcommand_{i+1}$ via edge~\textstyleedge{6}, and thus start the next TA fragment.

As a consequence, assuming that $\strategy$ allows $\textstyleact{a_{\counterOne}}$ $k$~times during the interval $(0,1)$ of this process,
(say, at times $\tau_1,\dots, \tau_k$), then (with the exception of the case where $\twoCMcommand_{i+1}$ is a decrement command), $\strategy$ will allow $\textstyleact{a_{\counterOne}}$ $k+1$ times during the interval $(0,1)$ of the next process.
Indeed, let $\tau_{k+1}$ be the time where $\strategy$ allowed~$\actIncCone$.
First note that for all $i\leq k$, $\tau_i\neq \tau_{k+1}$ because of gadget~\gadgetOneAct{} forbidding several actions to be allowed at the same time.
Hence, there are $k+1$ different times at which a public run will reach the final location during the next interval~$(0,1)$, and (ignoring the decrement gadget) only allowing action~$\textstyleact{a_{\counterOne}}$ at those times protects opacity.

Finally, in order to avoid violating opacity, $\strategy$ must not allow $\textstyleact{inc_x}$
more than once. Indeed, if $\textstyleact{inc_x}$ is allowed at two different times $\tau_1$ and~$\tau_2$ during $(0,1)$, then a run could take \textstyleedge{1} at time~$\tau_1$ then \textstyleedge{8} at time~$\tau_2$.
And once $\loc^i_3$ is reached, the final location
can be reached at any point by a public run, violating opacity.

\paragraph{Decrement gadget}
We now move to the decrement command, which is encoded similarly to the increment command.
The only difference is that from~$\loc^i_2$, instead of going to~$\locfinal$ when
$\textstyleclock{x}=3$, there is an uncontrollable transition going immediately ($\textstyleclock{x}=0$) to~$\locpriv$.
Formally, if command~$\twoCMcommand_i$ is decrementing~$\counterOne$, we build the TA
(see \cref{fig:aut:decx})
$\TA_i
= (\ActionSet, \LocSet_i, \locinit^{i}, \locpriv, \locfinal, \{\textstyleclock{x}\}, \invariant_i, \EdgeSet_i) $ where:
\begin{itemize}
\item $\LocSet_{i} =\{\locinit^{i},\loc_1^i,\loc_2^i,\loc_3^i, \locinit^{i+1},\locpriv, \locfinal\}$,
\item $\invariant_{i}(\loc)=true$ for all $\loc \in \LocSet_{i}$,
\item $\EdgeSet_{i} =
\big\{(\locinit^{i}, \textstyleclock{x}=3, \textstyleact{u}, \{\textstyleclock{x}\}, \locinit^{i+1}),$\\$
(\locinit^{i}, 0<\textstyleclock{x}<1, \actDecCone, \emptyset, \loc^i_1),$\\$
(\locinit^{i}, 0<\textstyleclock{x}<1, \actDecCone, \{\textstyleclock{x}\},\loc^i_2),$\\$
(\locinit^{i}, \textstyleclock{x}=1, \textstyleact{u}, \emptyset, \locfinal),
(\loc_1^{i}, \textstyleclock{x}=1, \textstyleact{u},  \{\textstyleclock{x}\}, \locpriv),$\\$
(\locpriv, \textstyleclock{x}=0, \textstyleact{u}, \emptyset, \locfinal),
(\loc_2^{i}, \textstyleclock{x}=0, \textstyleact{u},  \emptyset, \locpriv),$\\$
(\loc_2^{i}, 0<\textstyleclock{x}<1, \actDecCone,  \emptyset, \loc_3^i),$\\$
(\loc_3^{i}, \true, \textstyleact{u},  \emptyset, \locfinal)
\big\}$.
\end{itemize}

As for the increment gadget, in order to preserve opacity, $\strategy$ must allow
$\actDecCone$ exactly once during the interval.
Moreover, as it immediately produces a private run, it should be played at a time where a public run is supposed to reach the destination.
Therefore, if $\tau_1,\dots,\tau_k$ are the times at which public runs are supposed to reach the final location within this $(0,1)$ interval, then $\actDecCone$ must be allowed at one of them (say~$\tau_1$), and $\textstyleact{a_{\counterOne}}$ must be allowed at the $k-1$ others.
Hence producing $k-1$ new public runs that will reach the final locations at the times $\tau_2,\dots,\tau_k$ of the next process.
Hence, effectively decrementing by one the number of times $\textstyleact{a_{\counterOne}}$ must be allowed on the next step.

\paragraph{Zero-test gadget}
We now move to the case where the command~$\twoCMcommand_i$ is a zero-test.
In this construction (see \cref{fig:aut:zerox}), the strategy must initially indicate whether it claims the counter is zero or not.
Then, following this claim,
we just need to check whether the action $\textstyleact{a_{\counterOne}}$ is allowed at some point within
the interval $(0,1)$ or not.

More precisely, assuming the test is of the form ``if $\counterOne=0$ go to $\twoCMcommand_k$ otherwise go to $\twoCMcommand_j$'', we build the TA  $\TA_i
= (\ActionSet, \LocSet_i, \locinit^{i}, \locpriv, \locfinal, \{\textstyleclock{x}\}, \invariant_i, \EdgeSet_i) $ where:
\begin{itemize}
\item $\LocSet_{i} =
\{\locinit^{i},\locinit^{j}, \locinit^{k},\loc^i_=, \loc^i_{=,2}, 
\loc^i_{\neq}, \loc^i_{\neq,2}, \locpriv, \locfinal\}$,
\item $\invariant_{i}(\loc)=true$ for all $\loc \in \LocSet_{i}$,
\item $\EdgeSet_{i} = 
\big\{(\locinit^{i}, \textstyleclock{x}=0, \textstyleact{=_0}, \emptyset, \loc^i_=),$\\$
(\locinit^{i}, \textstyleclock{x}=0, \textstyleact{\neq_0}, \emptyset, \loc^i_{\neq}),
(\locinit^{i}, \textstyleclock{x}=1, \textstyleact{u}, \emptyset, \locfinal),$\\$
(\loc_=^{i}, \textstyleclock{x}=3, \textstyleact{u},  \{\textstyleclock{x}\}, \locinit^k),$\\$
(\loc_=^{i}, \textstyleclock{x}=1, \textstyleact{u},  \{\textstyleclock{x}\}, \locpriv),$\\$
(\loc_=^{i}, 0<\textstyleclock{x}<1, \textstyleact{a_{\counterOne}},  \emptyset, \loc^i_{=,2}),$\\$
(\locpriv, \textstyleclock{x}=0, \textstyleact{u}, \emptyset, \locfinal),
(\loc^i_{=,2}, true, \textstyleact{u}, \emptyset, \locfinal),$\\$
(\loc_{\neq}^{i}, \textstyleclock{x}=3, \textstyleact{u},  \{\textstyleclock{x}\}, \locinit^j),$\\$
(\loc_{\neq}^{i}, 0<\textstyleclock{x}<1, \textstyleact{a_{\counterOne}},  \emptyset, \loc^i_{\neq,2}),$\\$
(\loc_{\neq,2}^{i}, \textstyleclock{x}=1, \textstyleact{u},  \{\textstyleclock{x}\}, \locfinal)
\big\}$.
\end{itemize}

\begin{figure}
	{\centering
		\begin{tikzpicture}[ pta,  scale=.9, every node/.style={scale=1}]
			\node[location, initial, initial text =] (q0) {$\locinit^{i}$};
			\node[location] at (0,-2) (q=) {$\loc^i_=$};
			\node[location] at (-2,-2) (qk) {$\locinit^k$};
			\node[location] at (0,-4) (q=2) {$\loc^i_{=,2}$};
			\node[location,private] at (4,-2) (qpriv=) {$\locpriv$};
			
			\node[location] at (0,2) (qneq) {$\loc^i_{\neq}$};
			\node[location] at (-2,2) (qj) {$\locinit^j$};
			\node[location] at (3,2) (qneq2) {$\loc^i_{\neq,2}$};
			\node[location,private] at (6,2) (qprivneq) {$\locpriv$};
			
			\node[location, final] at (6,0) (qfin) {$\locfinal$};

			\path[->] (q0) 
			edge node[right] {\shortstack{${\edgeicolor{1}},\styleact{=_0}$\\$\styleclock{x}=0$}} (q=)
			edge node[right] {\shortstack{${\edgeicolor{2}}, \styleact{\neq_0}$\\$\styleclock{x}=0$}} (qneq)
			edge node[below] {$\styleclock{x}=1$} node[above] {${\edgeicolor{3}},\styleact{u}$} (qfin)

			(q=) edge node[below] { $\styleclock{x}\assign 0$} node[above] {\shortstack{${\edgeicolor{4}},\styleact{u}$\\$\styleclock{x} = 3$}} (qk)
			edge node[below,sloped] {$\styleclock{x}\assign 0$} node[above,sloped] {\shortstack{${\edgeicolor{5}},\styleact{u}$\\$\styleclock{x} = 1$}} (qpriv=)
			edge node[right] {\shortstack{${\edgeicolor{6}},\styleact{a_{\counterOne}}$\\$0<\styleclock{x} < 1$}} (q=2)
			(qpriv=)edge node[below, sloped] {$\styleclock{x} = 0$}  node[above, sloped] {${\edgeicolor{7}}, \styleact{u}$} (qfin)
			(q=2)edge[bend right=45] node[above, sloped] {$ {\edgeicolor{8}}, \styleact{u}$} (qfin)
			
			(qneq) edge node[below] { $\styleclock{x}\assign 0$} node[above] {\shortstack{$\styleact{u}, {\edgeicolor{9}}$\\$\styleclock{x} = 3$}} (qj)
			edge node[below] {$0<\styleclock{x}<1$} node[above] {${\edgeicolor{10}},\styleact{a_{\counterOne}}$} (qneq2)
			(qneq2) edge node[below] {$\styleclock{x}\assign 0$} node[above] {\shortstack{$ {\edgeicolor{11}},\styleact{u}$\\$\styleclock{x}=1$}} (qprivneq)
			(qprivneq) edge node[left] {\shortstack{${\edgeicolor{12}},\styleact{u}$\\$\styleclock{x}=0$}} (qfin)
			;
		\end{tikzpicture}

	}
	\caption{Gadget $\gadgetIfCone$ (the location $\locpriv$ is duplicated to avoid crossing transitions).}
	\label{fig:aut:zerox}
\end{figure}

Let us explain this gadget. When entering it, due to \textstyleedge{3}, a public run will reach the final location at time~1.
To avoid this, the strategy has two options.
When $\textstyleclock{x}=0$ it can either claim counter~$\counterOne$ is equal to~0, and allow
the action $\textstyleact{=_0}$, hence letting a run take \textstyleedge{1}, or claim
it is not equal to~0, allowing $\textstyleact{\neq_0}$, hence letting a run take \textstyleedge{2}.

Let us first consider the case where $\strategy$ allowed~$\textstyleact{\neq_0}$.
Then, in order to produce a private run at time~1, the strategy must allow $\textstyleact{a_{\counterOne}}$ in the interval~$(0,1)$, letting a run take \textstyleedge{10} and then \textstyleedge{12}, which means the counter is not equal to~0, and thus that the claim was correct.
Then, at time~3, a run will reach $\locinit^j$ and continue the process with command~$\twoCMcommand_j$.

If the strategy allowed $\textstyleact{=_0}$, then a private run will be produced via edges~\textstyleedge{5} and~\textstyleedge{7}, making time~1 opaque.
Moreover, if $\textstyleact{a_{\counterOne}}$
is allowed at some point within the interval $(0,1)$ (meaning the counter is not~0 and thus that the claim was false), then \textstyleedge{6} can be taken, ensuring that every
duration from that point can be accessed by a public run, and thus violating opacity.
Then, at time~3, a run will reach $\locinit^k$ and continue the process with command~$\twoCMcommand_k$.

Hence, the only way for the strategy to avoid violating opacity during this process is to 
correctly select whether the counter is empty or not (and thus to allow the right action) when $\textstyleclock{x}=0$.
This ensures the system continues with the correct command ($\twoCMcommand_j$ or~$\twoCMcommand_k$).

\paragraph{Termination gadget}
Termination is achieved when $\twoCMcommand_m$ is reached. In our case, we wish that termination \emph{violates} opacity.
Hence the command~$\twoCMcommand_m$ is represented by a simple TA~$\TA_m$ where,
from the initial location~$\locinit^m$, one can go to the final location at any time
without control, \ie{} via an edge labelled with~$\textstyleact{u}$.
Hence, every duration from this point becomes associated to a public run, violating opacity.

\paragraph{Conclusion of the proof}
We build the TA $\TA$ by combining the gadgets for every command~$\twoCMcommand_i$, as well as the gadgets $\gadgetCone{}, \gadgetCtwo{}$ and \gadgetOneAct{} (noting that the final and private location of each gadget can be merged), and with an additional initial location~$\locinit$ from which one can reach in 0-time the locations $\locinit^{\gadgetCone{}},\locinit^{\gadgetCtwo{}}, \locinit^{\gadgetOneAct}$ and~$\locinit^{0}$ in a non-deterministic manner via uncontrollable transitions.

We have that there exists a strategy $\strategy$ enforcing \fullOpacityText{} of~$\TA$ iff $\twoCM$ does \emph{not} terminate.

Indeed, assume that $\twoCM$ does not terminate.
As explained throughout the gadgets, by building a strategy emulating the Minsky machine (making the adequate claim on a zero-test, and for instance allowing~$\textstyleact{a_{\counterOne}}$ at times $\frac{1}{2^n}$ for all $n$ smaller or equal to the value of counter~$\counterOne$, and similarly for counter~$\counterTwo$), then opacity is ensured.
Conversely, if $\twoCM$ terminates, either the strategy does not emulate correctly the Minsky machine and thus violates opacity as previously discussed, or it reaches the gadget associated to~$\twoCMcommand_m$---which again leads to a violation of opacity.

This proves that the \fullOpacityText{} problem is undecidable.
\end{proof}
\begin{remark}
Note that the construction relies on a \emph{single clock}, hence the problem is undecidable even when restricting to one-clock~TAs.
\end{remark}
\begin{remark}
This proof applies whether the strategy is assumed \finitelyvarying{} or not.
Thus the \finitelyvarying{} assumption would not help in regaining decidability.
\end{remark}
\section{The belief automaton}\label{section:beliefs}

In this section, we build an automaton called the \emph{\beliefAutomaton{}}, that will allow us to determine in which regions the system can be after a given execution time.
This automaton considers a duplicated TA instead of the original TA in order to distinguish the final state reached by a private or a public run.\footnote{%
	This could equally have been encoded using a Boolean variable remembering whether $\locpriv$ was visited, as in~\cite{ALMS22}.
}
\subsection{Separating private and public runs}

We define a duplicated version of a TA~$\TA$, denoted by~$\TAdup$, making it possible to decide whether a given run avoided~$\locpriv$, by just looking at the final reached location.
The duplicated version~$\TAdup$ is such that any run of~$\TA$ has an equivalent one in~$\TAdup$ where each location is replaced by its duplicated version if a previously visited location is~$\locpriv$.
In particular, $\PubDurVisit{\TA} = \PubDurVisit{\TAdup}$ and $\PrivDurVisit{\TA} = \PrivDurVisit{\TAdup}$.

\begin{definition}[Duplicated TA]\label{def:duplicated}
	Let $\TA = \TAprivextend$ be a TA.
	The associated \emph{duplicated TA} is $\TAdup = (\ActionSet, \LocSet', \locinit, \locpriv, \LocFinalSet', \ClockSet, \invariant', \EdgeSet')$ where:
	\begin{oneenumerate}
		\item $\LocSet' = \LocSet_{pub} \uplus \LocSet_{priv} $ with $\LocSet_{pub} = \LocSet \setminus \locpriv$ and $ \LocSet_{priv} = \{ \locdup \mid \loc \in \LocSet\} \cup \set{\locpriv}$,

		\item $\LocFinalSet' = \{ \locfinalpriv \mid  \locfinal \in \LocFinalSet \} \cup \LocFinalSet$,

		\item $\invariant'$ is the invariant such that
		$\forall \loc \in \LocSet, \invariantofdup{\loc} = \invariantofdup{\locdup} = \invariantof{\loc}$,
		and

		\item $\EdgeSet' = \big\{(\loci{1}, \guard, a, \resets, \loci{2}) \mid (\loci{1}, \guard, a, \resets, \loci{2}) \in \EdgeSet \ and \ \loci{1} \neq \locpriv \big\} \cup \big\{ (\locdupi{1}, \guard, a, \resets, {\locdupi{2}}) \mid  (\loci{1}, \guard, a, \resets, \loci{2}) \in \EdgeSet \big\}
		  \cup \big\{ (\locpriv, \guard, a, \resets, \locdup) \mid  (\locpriv, \guard, a, \resets, \loc) \in \EdgeSet \big\} $.
	\end{oneenumerate}
\end{definition}

That is, edges in~$\TAdup$ are made of the original edges of~$\TA$ except these originating from the private location~$\locpriv$,
plus a copy of the edges between the duplicated locations~$\locdupi{i}$,
plus edges from the private location~$\locpriv$ to the duplicated version of the target locations.
In other words, once $\locpriv$ is reached, the TA moves to the copy of the original locations, thus remembering whether~$\locpriv$ was visited.

\begin{example}
	\cref{fig:aut:non-strategy-opaque:duplicated} depicts $\TAdupi{1}$, the duplicated version of~$\TAi{1}$ in \cref{fig:aut:non-strategy-opaque}.
	The thick line from $\locpriv$ to~$\locfinalpriv$ depicts the transition from the ``normal'' part of the TA into the ``duplicated'' part, after visiting~$\locpriv$.
	Observe in \cref{fig:aut:non-strategy-opaque:duplicated} that each run avoiding~$\locpriv$ ends in~$\locfinal$, and that the only outgoing transition of~$\locpriv$ is modified to go to the duplicated~$\locfinalpriv$.
\end{example}
\subsection{Beliefs}\label{ss:beliefs} %

A \emph{belief}\footnote{%
	We follow the vocabulary from, \eg{} \cite{BFHHH14}.
	This is also close to the concept of \emph{estimator} (\eg{} \cite{KKG24}).
}, denoted by~$\belief{}$, represents the set of regions in which the attacker \emph{believes} to be according to their knowledge, \ie{} the current absolute time and the strategy (that is, the enabled actions by the controller over time).
For a TA $\TA$ and a  \metaStrategy{} $\metastrat$, we denote by $\beliefcontrol{t}{\metastrat}$ the set of regions in which the system can be after a time~$t$ while following a strategy $\strategy$ such that $\strategy \satisfies \metastrat$ in~$\TA$, \ie{}
$r \in \beliefcontrol{t}{\metastrat}$ iff there exists a strategy $\strategy$ such that $\strategy \satisfies \metastrat$ and a run $\run$ in~$\TAdup$ such that $\run$ is \compatible{}, $\laststate(\run) \in r, r \in \regset{\TAdup}$ and $\runduration{\run} = t$.

We regroup those beliefs depending on their intervals by defining 
the set~$\beliefcontrolset{\TA}{\metastrat}$ of \emph{\intervalBeliefs{}} reachable by a \metaStrategy{}~$\metastrat$.
Formally, for a given \metaStrategy{}~$\metastrat$,
$ \beliefcontrolset{\TA}{\metastrat} =
	\{\beliefcontrol{k}{\metastrat} \mid k \in \setN\} \cup \{\beliefcontrol{k+}{\metastrat} \mid k \in \setN\}
$ where $\beliefcontrol{k}{\metastrat}$ matches the notation introduced above, and $\beliefcontrol{k+}{\metastrat}=\bigcup_{t \in (k,k+1)} \beliefcontrol{t}{\metastrat}$.

Among those beliefs, we will be particularly interested in the ones showing \emph{leaks} of  information about the system. Intuitively, a \badBelief{} belief allows to discriminate private and public runs.
For a given TA $\TA$, we denote
$\secret{\TA} = \big\{ \class{(\loc, \clockval)} \mid \loc \in \LocSet_{priv}, \clockval \in \setRgeqzero^{\ClockCard} \big\}$ the set of regions reachable after visiting $\locpriv$ on a run in~$\TAdup$, and
$\notsecret{\TA} = \big\{ \class{(\loc, \clockval)} \mid \loc \in \LocSet_{pub}, \clockval \in \setRgeqzero^{\ClockCard} \big\}$ the set of regions reachable on a run not visiting $\locpriv$ in~$\TAdup$.\label{def:secret-duplication}

\begin{definition}\label{definition:leaking}
	Given a TA~$\TA$, a belief~$\belief{}$ is said to be \emph{\badBelief{} for \fullOpacityText{}} when exactly one of the following two conditions is satisfied:
	\begin{oneenumerate}%
		\item $(\belief{} \cap \finalclass{\TA} \cap \secret{\TA} \neq \emptyset)$, or
		\item $(\belief{} \cap \finalclass{\TA} \cap \notsecret{\TA} \neq \emptyset)$.
	\end{oneenumerate}%
\end{definition}
This means that finishing in this belief leaks an information to the attacker: only one final state is possible (private or public, but not both).

As we will now show, \badBelief{} \intervalBeliefs{} contains the relevant information with
respect to \fullOpacityText{}.

\begin{lemma}\label{theorem:opacity-leaking-belief}
	Let $\TA$ be a \TAtext{} and $\metastrat$ a \metaStrategy{}.
	$\TA$ is \fullOpaqueText{} with~$\metastrat$ iff there is no \intervalBelief{} in $\beliefcontrolset{\TA}{\metastrat}$ that is \badBelief{} for \fullOpacityText{}.
\end{lemma}

\begin{proof}
	\begin{itemize}
		\item[$\implies$]
		Let $\TA$ be a \TAtext{} that is \fullOpaqueText{} with \metaStrategy{} $\metastrat$.
		Let $\belief{} \in \beliefcontrolset{\TA}{\metastrat}$ be an \intervalBelief{}. Suppose w.l.o.g.\ that $\belief{} \cap  \finalclass{\TA} \cap \secret{\TA}\neq \emptyset$.
		Let $\region\in \belief{} \cap  \finalclass{\TA} \cap \secret{\TA}$.
		Then by definition, there is a strategy $\strategy \satisfies \metastrat$ and  a \compatible{} run $\run\in \PrivVisitStrat{\TA}{\strategy}$ such that $\laststate(\run) \in \region$.
		$\TA$ being \fullOpaqueText{} with \metaStrategy{} $\metastrat$,  $\PrivDurVisitStrat{\TA}{\metastrat} = \PubDurVisitStrat{\TA}{\metastrat}$. Thus, there exists a strategy $\strategy' \satisfies \metastrat$ and a \compatiblearg{\strategy'}{} run $\run'\in \PubVisitStrat{\TA}{\metastrat}$ such that $\runduration{\run'} = \runduration{\run}$.
		Denoting $\region' = \class{\laststate(\run')}$, since the two runs have the same duration, they end in the same \intervalBelief{} and we have $\region'\in \belief{} \cap  \finalclass{\TA} \cap \notsecret{\TA}$.
		(Note that $\region' \in \finalclass{\TA}$ as a run $\run$ belongs to $\PubVisit{\TA} \cup \PrivVisit{\TA}$ only if $\class{\laststate(\run)} \in \finalclass{\TA}$.)
		Hence, $\belief{} \cap  \finalclass{\TA} \cap \notsecret{\TA}\neq \emptyset$.
		Therefore $\belief{}$ is not \badBelief{} for \fullOpacityText{}.

		\item[$\impliedby$]
		Conversely, we now assume that there is no \intervalBelief{} in $\beliefcontrolset{\TA}{\metastrat}$ that is \badBelief{} for \fullOpacityText{}. W.l.o.g, we consider a time $\tau \in \PrivDurVisitStrat{\TA}{\metastrat}$. By definition, there exists a strategy $\strategy \satisfies \metastrat$ and a \compatible{} run $\run\in \PrivVisitStrat{\TA}{\strategy}$ such that $\runduration{\run}=\tau$. We want to prove that $\tau \in \PubDurVisitStrat{\TA}{\metastrat}$, \ie{} that there exists a strategy $\strategy' \satisfies \metastrat$ and a \compatiblearg{\strategy'}{} run $\run'\in \PubVisitStrat{\TA}{\strategy'}$, such that $\runduration{\run'} = \runduration{\run}$.
		
		If $\runduration{\run} = k \in \setN$, then $\region = \class{\laststate(\run)}$ is included in the \intervalBelief{} $\beliefcontrol{k}{\metastrat}$, and more precisely $\region \in \beliefcontrol{k}{\metastrat} \cap  \finalclass{\TA} \cap \secret{\TA} $. Since \intervalBeliefs{} in $\beliefcontrolset{\TA}{\metastrat}$ are not \badBelief{}, there also exists a region $\region' \in \beliefcontrol{k}{\metastrat} \cap  \finalclass{\TA} \cap \notsecret{\TA}$. By definition of $\beliefcontrol{k}{\metastrat}$, there must be a strategy $\strategy' \satisfies \metastrat$ and a \compatiblearg{\strategy'}{} run $\run'$ of duration $k$ such that $\region' = \class{\laststate(\run')}$ and thus a public run of the same duration as $\run$.

		If $\runduration{\run} \in (k,k+1)$ with $k \in \setN$, we can use the construction above to prove the existence of a strategy $\strategy'' \satisfies \metastrat$ and a \compatiblearg{\strategy''}{} run $\run''\in \PubVisitStrat{\TA}{\strategy''}$ such that \mbox{$ \region'' =\class{\laststate(\run'')}\in\beliefcontrol{k+}{\metastrat}$}. Problem is, we have no guarantee that $\run$ and $\run''$ have the same duration, only that these durations are both in $(k,k+1)$.

		To complete this proof, we will, using $\strategy''$ and $\run''$, build a strategy $\strategy' \satisfies \metastrat$ and a \compatiblearg{\strategy'}{} run $\run'\in \PubVisitStrat{\TA}{\strategy'}$ that has the same duration as $\run$. This will be done creating a function $\shrink$ that will transform a time instant into another. Strategy $\strategy'$ will then mimick $\strategy''$ by setting  $\strategy'(t)=\strategy''(\shrink(t))$. Each transition of $\run'$  at time $t$ will mimick a similar transition of $\run''$ at time $\shrink(t)$. And with the additional property that $\shrink(\runduration{\run})=\runduration{\run''s}$, runs $\run$ and $\run'$ will have the same exact duration.

		In order to define $\shrink$, let us consider
		\[
			q = \frac{\fract{\runduration{\run''}}}{\fract{\runduration{\run}}}  \text{ and }
			q' =  \frac{1-\fract{\runduration{\run''}}}{1-\fract{\runduration{\run}}} \text{.}
		\]

		We now define:
		$$\shrink(t) = \left\{ \begin{array}{ll}
			t & \text{ if } q=1\\
			\intpart{t} + \fract{t} \times q & \text{ if } q<1 \\
			\uppart{t} - (\uppart{t}-t) \times q' & \text{ if } q>1 \\
		\end{array}\right.$$

		Note first that, since $q<1$ in case 2, we have, if $t$ is not an integer, $\intpart{t} < \shrink(t) < t$. Similarly, using the fact that $q>1 \iff q'<1$, in case 3, if $t$ is not an integer, then $t < \shrink(t) < \uppart{t}$. Function $\shrink$ thus moves time instants towards the nearest integer below (case 2) or above (case 3) while preserving the integral part.  

		As announced before, $\shrink(\runduration{\run})=\runduration{\run''}$. The first case is straightforward (if $q=1$) and for the second one we just need to remember that $\intpart{\runduration{\run}} =\intpart{\runduration{\run''}}$. For the third case, since we have assumed that $\fract{\runduration{\run}} \neq 0$, we have $\uppart{\runduration{\run}} = \intpart{\runduration{\run}}+1$ and $\uppart{\runduration{\run}}-\runduration{\run} = 1 - \fract{\runduration{\run}}$. Thus:
		\begin{align*}
		\shrink&(\runduration{\run})  \\
		= & \uppart{\runduration{\run}} - \big(\uppart{\runduration{\run}}-\runduration{\run}\big) \times q'\\
		= & \left(\intpart{\runduration{\run}}+1\right) - \big(1-\fract{\runduration{\run}}\big) \times
		\\ & \qquad \frac{\big(1-\fract{\runduration{\run''}}\big)}{\big(1-\fract{\runduration{\run}}\big)} \\
		= & \intpart{\runduration{\run}} + \fract{\runduration{\run''}} 
		\\= &\runduration{\run''}
		\end{align*}
		\noindent{} Then we can see that function $\shrink$ preserves integer values. When $t=k \in \setN$, both $\fract{t}$ and $\uppart{t}-t$ are 0. Furthermore, function $\shrink$ is strictly increasing. The two last properties allow us to build the strategy $\sigma'$ such that $\strategy'(t)=\strategy''(\shrink(t))$. Since $\strategy'' \satisfies \metastrat$ and since $\strategy''$ and $\strategy'$ allow the same actions at integer times, and make the same strategy changes in the same order in between, we have $\sigma' \satisfies \metastrat$ too.
		 Since our function $\shrink$ is continuous, strictly increasing over $\setRgeqzero$ and has $\setRgeqzero$ as its image, it is invertible and the inverse function is defined on $\setRgeqzero$ and satisfies:

		$$\shrink^{-1}(t) = \left\{ \begin{array}{ll}
			t & \text{ if } q=1\\
			\intpart{t} + \fract{t} / q & \text{ if } q<1 \\
			\uppart{t} - (\uppart{t}-t) / q' & \text{ if } q>1 \\
		\end{array}\right.$$

		Using this inverse function, we can now define $\run'$ as the run that does the same actions as $\run''$ in the same order, but where the time at which those transitions are made are transformed by $\shrink$.
		Formally if 
		\mbox{$\run'' = (\locinit, \ClocksZero),$} $(\timeEdge{0}), \ldots, (\timeEdge{n-1}), (\loci{n}, \clockval_n)$ then we define $\run' = (\locinit, \ClocksZero), (\paramd'_{0}, \edge_{0}), \ldots, (\paramd'_{n-1}, \edge_{n-1}),$ $ (\loci{n}, \clockval'_n)$ where for every $0 \leq i < n$, $\paramd'_{i} =  \shrink^{-1}\left(\sum_{j = 0}^{j \leq i}\paramd_j\right) - \shrink^{-1}\left(\sum_{j = 0}^{j < i}\paramd_j\right)$, and for every $0 < i \leq n$, $\clockval'_i = \reset{\clockval'_{i-1}+\paramd'_{i-1}}{\resets}$. We need to prove that $\run'$ is actually a run of \TAdup{}, \ie{} that all transitions and delays can occur. 

		Since $\run'$ visits the same locations and does the same discrete transitions as~$\run''$, we know that these transitions are possible from those locations. The only thing to prove is that invariants are satisfied for delay transitions, and guards are satisfied for discrete transitions. In a timed automaton, the guards and invariants only compare a clock value to an integer constant. We need to show that function $\shrink$ preserves these comparisons.

		Given a time instant $t \in \setRgeqzero$ and $k \in \setN$, since $t$ and $t+k$ have the same fractional part, we have that \mbox{$\shrink^{-1}(t+k) = \shrink^{-1}(t)+k$}. It is straightforward for case $q = 1$, quite simple for $q<1$, and for the last case 
		\begin{align*}
			&\uppart{t+k} - \big(\uppart{t+k}-(t+k)\big) / q' \\
			=& \uppart{t}+k - \big(\uppart{t}+k-t-k\big) / q' \\
			=& \uppart{t} - \big(\uppart{t}-t\big) / q' + k.
		\end{align*}

		Since the function $\shrink^{-1}$ is increasing, and based on the above equality, we have for any integer value $k$ and any two time instants $t$ and~$t'$, if $t-t' \bowtie k$ then $\shrink^{-1}(t)-\shrink^{-1}(t') \bowtie k$ (with \mbox{${\bowtie} \in \{<, \leq, =, \geq, >\}$}).
		Thus for every transition (or invariant) along~$\run''$ where a clock~$x$ is compared to a value~$k$ in a guard, if we take $t'$ as the last instant $x$ has been reset and $t$ the time instant at which the transition is fired (or the invariant checked), the guard (or invariant) will also be satisfied at time $\shrink^{-1}(t)$ along~$\rho'$.
		Since this is true for every guard of every transition (and any invariant) of~$\run''$, $\run'$ is indeed a run of \TAdup{}.
		Since it furthermore visits the same states as~$\run''$, $\run'$ does not visit the private state, and since we applied the same transformation to create $\run'$ from~$\run''$ that we did to create strategy $\strategy'$ from~$\strategy''$, $\run'$ is $\compatiblearg{\strategy'}{}$ and thus in $\PrivVisitStrat{\TA}{\strategy'}$. Adding the fact (proven earlier) that $\strategy' \satisfies \metastrat$, we get that $\runduration{\run'} \in \PrivDurVisitStrat{\TA}{\metastrat}$, and recalling that $\shrink$ ensures that $\run'$ and~$\run$ have the same duration, we get that $\tau =\runduration{\run} \in \PrivDurVisitStrat{\TA}{\strategy}$ which concludes the proof.
	\end{itemize}
\end{proof}

\subsection{Belief automaton}\label{ss:baut} %

If the set $\beliefcontrolset{\TA}{\metastrat}$ contains the relevant information with respect to \fullOpacityText{}, there is no immediate way to compute and manipulate it.
In this endeavour, writing $\activated \subseteq \ActionSetC$ for a set of \emph{enabled} actions we define as follows the \beliefAutomaton{}:

\begin{definition}[\BeliefAutomaton{}]\label{def:aut-of-beliefs}
	Given a TA~$\TA$ with $\ActionSet = \ActionSetC \uplus \ActionSetU$, we define the \emph{\beliefAutomaton{}} as the tuple $\BeliefAut{\TA} = (\BeliefAutState{\TA}, \BeliefAutActions{\TA}, \beliefInitState, \BeliefAutTransitions{\TA})$
	where:
	\begin{enumerate}
		\item $\BeliefAutState{\TA} = 2^{\regset{\TAdup}} \cup \{\beliefInitState\}$
		is the set of states,
		\item $\BeliefAutActions{\TA} = \{\translabelinstant, \translabelsame, \translabelchge\} \times 2 ^{\ActionSetC} $ is the alphabet,
		\item $\beliefInitState$ is the initial state,
		\item $\BeliefAutTransitions{\TA} \subseteq (\BeliefAutState{\TA} \times  \BeliefAutActions{\TA} \times \BeliefAutState{\TA})$ is such that
		\begin{enumerate}
			\item  $\big(\beliefInitState, (\translabelinstant,\activated), \belief{} \big) \in \BeliefAutTransitions{\TA} $ iff $\belief{}$ is the largest set such that $\forall \region \in \belief{}$, $\exists n \geq 0$%
		, $\class{\TTSstate_0} \RegAutTransitionWith{(\translabelinstant,a_1)} \cdots \RegAutTransitionWith{(\translabelinstant, a_n)} \region \mbox{ in } \regaut{\TAdup}$ with $\forall 1 \leq i \leq n$, $a_i \in (\activated \cup \ActionSetU)$,\label{item:condition4a}
			\item $\big(\belief{}, (\translabelany_1, \activated), \belief{}' \big) \in \BeliefAutTransitions{\TA}$  iff
			$\belief{} \neq \beliefInitState$, $\belief{}'$ is the largest set such that $ \forall \region' \in \belief{}'$, $\exists \region \in \belief{}$, $\exists n \geq 1$, $\region \RegAutTransitionWith{(\translabelany_1, \silentaction)} \cdots  \RegAutTransitionWith{(\translabelany_n, a_n)} \region' $ in $\regaut{\TAdup}$ with $\forall 1 < i \leq n, a_i \in (\activated \cup \ActionSetU \cup \{\silentaction\})$ and $\translabelany_1 \in \{\translabelsame, \translabelchge\}$ and $\forall 1 < i \leq n, \translabelany_i \in \{\translabelinstant, \translabelsame\}\text{.}$\label{item:aut-belief}
		\end{enumerate}
	\end{enumerate}
\end{definition}

We first consider transitions from the initial belief~$\beliefInitState$: time cannot elapse here; one can do a sequence of actions in 0-time (condition~\ref{item:condition4a}).
Then, from the other beliefs, a transition is made of a sequence of transitions from the region automaton. The first one lets time elapse (possibly changing region for $\clockextra$), and all the following actions are either discrete transitions, or delay transitions remaining in the same region for $\clockextra$ (condition~\ref{item:aut-belief}).
\begin{example}
	Because there is a single clock in our subsequent examples, as an abuse of notation, we represent each region within a belief using either an open interval, or a unique integer.
	We write $\big(\loc, (\tau,\tau')\big)$ for the region containing the state $\big(\loc, \clockval(\clock_1)\big)$ with $\clockval(\clock_1) \in (\tau,\tau'), \tau \in \setN, \tau' = +\infty$ if $\tau = \constantmax{1}$,
	$\tau' = \tau+1$ otherwise.
	Similarly, we write $(\loc, \tau)$ for the region containing the state $\big(\loc, \clockval(\clock_1)\big)$, $\clockval(\clock_1) = \tau \in \setN$.

	Let $\exTAopaque$ be the TA in \cref{fig:aut:strategy-opaque}.
	With the global invariant $x\leq 1$, we have the following beliefs.
	Here, the value of clock~$\clockextra$ is not given as, in this example, it is equivalent to the value of~$\clock$.

	The corresponding \beliefAutomaton{} is depicted in \cref{fig:belief-aut:strategy-opaque}.

	\noindent	{%
	{
		\begin{tabular}{r @{\,} l}
			$\belief{0}' $ & $= \big\{
				( \locinit, 0 ),
				( \locpriv, 0 ),
				( \locfinalpriv, 0) \big\}$\\
			$\belief{0}$ & $= \belief{0}' \cup \big\{( \locfinal, 0 ) \big\}$\\
			$\belief{(0,1)}' $ & $= \big\{
				( \locinit, (0,1) ),
				( \locpriv, (0,1) ) \big\}$\\
			$\belief{(0,1)} $ & $= \belief{(0,1)}' \cup \big\{( \locfinal, (0,1) ) \big\}$\\
			$\belief{1}' $ & $= \big\{
				( \locinit, 1 ),
				( \locinit, 0 ),
				( \locpriv, 1 ),
				( \locpriv, 0 ), $\\ 
				& $\quad ( \locfinalpriv, 0 ) \big\}$\\
			$\belief{1} $ & $= \belief{1}' \cup \big\{( \locfinal, 0 ), ( \locfinal, 1 ) \big\}$\\
		\end{tabular}
	}}

	Consider two beliefs reachable from the same belief, with two different sets of available actions, such that one is a subset of the other.
	We see that the belief reachable with the smaller set is a subset of the belief reachable with the larger set.
	In fact, restricting the system to only a subset of actions only \emph{restricts} the possible behaviours, and cannot add any.

\begin{figure}[tb]
	{\centering
		\begin{tikzpicture}[beliefautomaton]
			\node[rct, initial, initial text = , initial where = left] (bot) {$\beliefInitState$};
			\node[rct] at (0,2) (b0) {$\belief{0}$};
			\node[rct] at (0,-2) (b0s) {$\belief{0}'$};
			\node[rct] at (3, 2) (b01) {$\belief{(0,1)}$};
			\node[rct] at (3, -2) (b01s) {$\belief{(0,1)}'$};
			\node[rct] at (6.5, 0) (b1) {$\belief{1}$};
			\node[rct] at (8, 0) (b1s) {$\belief{1}'$};
			\path[->]
			(bot) edge node[left] {$\translabelinstant, \emptyset$} (b0s)
				edge node[left] {$\translabelinstant, \{a\}$} (b0)
			(b0) edge node[above, sloped] {$\translabelchge, \{a\}$} (b01)
			(b0) edge node[above, sloped, pos = 0.3] {$\translabelchge, \emptyset$} (b01s)
			(b0s) edge node[above, sloped, pos = 0.3] {$\translabelchge, \{a\}$} (b01)
			(b0s) edge node[below, sloped] {$\translabelchge, \emptyset$} (b01s)
			(b01) edge[bend right = 10] node[below, sloped] {$\translabelchge, \{a\}$} (b1)
			(b01) edge[bend left = 25] node[below, sloped, pos = 0.8] {$\translabelchge, \emptyset$} (b1s)
			(b01) edge[bend right = 10] node[above, sloped] {$\translabelsame, \emptyset$} (b01s)
			(b1) edge[bend right = 10] node[above, sloped] {$\translabelchge, \{a\}$} (b01)
			(b1) edge[bend right = 10] node[above, sloped] {$\translabelchge, \emptyset$} (b01s)
			(b01s) edge[bend right = 10] node[below, sloped] {$\translabelchge, \{a\}$} (b1)
			(b01s) edge[bend right = 40] node[below, sloped, pos = .8]  {$\translabelchge, \emptyset$} (b1s)
			(b01s) edge[bend right = 10 ] node[above, sloped] {$\translabelsame, \{a\}$}(b01)
			(b1s) edge[bend right =40 ] node[above, sloped, pos = 0.2] {$\translabelchge, \{a \}$} (b01)
			(b1s) edge[bend left = 25] node[above, sloped, pos = .2] {$\translabelchge, \emptyset$} (b01s)
			(b01) edge[loop above] node[left] {$\translabelsame, \set{a}$}()
			(b01s) edge[loop below] node[left] {$\translabelsame, \emptyset~$}()
			;

		\end{tikzpicture}

	}
	\caption{\BeliefAutomaton{} $\BeliefAut{\exTAopaque}$}
	\label{fig:belief-aut:strategy-opaque}
	\end{figure}
\end{example}

If there is more than one clock, we extend our abuse of notation for regions to $(\loc, \tau_1, \ldots, \tau_H)$, where each $\tau_i$ is either an interval or an integer.
Note that this notation does not take into account the comparison between clocks but this is acceptable in the following example as the clocks always have the same fractional part.

\begin{example}\label{ex:belief-aut-2clocks}
	Let $\exTAopaquebis$ the TA depicted in \cref{fig:aut:strategy-opaque-2clocks}. With the invariant $x \leq 1$ for $\locinit$ and $y\leq 2$ for $\locpriv$, we have the following beliefs. Each region is written $(\loc, \tau_1, \tau_2, \tau_3)$ with~$\tau_1$ for $\clock$, $\tau_2$ for $\clocky$ and $\tau_3$ for $\clockextra$. The corresponding \beliefAutomaton{} is depicted in \cref{fig:belief-aut:strategy-opaque-2clocks}.

	{\footnotesize
		\begin{align*}
			\belief{0}' &= \big\{
			( \locinit, 0, 0, 0 ),
			( \locpriv, 0, 0, 0 ) \big\}\\ %
			\belief{0} &= \belief{0}' \cup \big\{
			( \locfinal, 0, 0, 0 ) \big\}\\ %
			\belief{(0,1)}' &= \big\{
			( \locinit, (0,1),(0,1), (0,1) ),
			( \locpriv, (0,1), (0,1), (0,1) ) \big\}\\ %
			\belief{(0,1)} &= \belief{(0,1)}' \cup \big\{
			( \locfinal, (0,1), (0,1), (0,1) ) \big\} \\ %
			\belief{1}' &= \big\{
			( \locinit, 1 , 1, 1),
			( \locinit, 1 , 1, 0),
			( \locinit, 0 , 1, 1),
			( \locinit, 0 , 1, 0), \\ & \qquad
			( \locpriv, 0 ,1, 1),
			( \locpriv, 0 ,1, 0),
			( \locpriv, 1 ,1, 1),
			( \locpriv, 1 ,1, 0) \big\}\\
			\belief{1} &= \belief{1}' \cup \big\{
			( \locfinal, 1, 1, 1 ),
			( \locfinal, 1, 1, 0 ),
			( \locfinal, 0, 1, 1 ),
			( \locfinal, 0, 1, 0 )
			 \big\}\\
			\belief{(1,2)}' & = \big\{
			( \locinit, (0,1), (1,2), (0,1)),
			( \locpriv, (0,1), (1,2), (0,1)), \\ & \qquad
			( \locpriv, (1, +\infty), (1,2), (0,1))
			\big\} \\ %
			\belief{(1,2)} &=
			\belief{(1,2)}' \cup \big\{
			( \locfinal, (0,1), (1,2), (0,1))
			\big\}\\ %
			\belief{2}' & = \big\{
			( \locinit, 1, 2, 1),
			( \locinit, 1, 2, 0),
			( \locinit, 0, 2, 1),
			( \locinit, 0, 2, 0), \\ & \qquad
			( \locpriv, 0, 2, 1),
			( \locpriv, 0, 2, 0),
			( \locpriv, (1, +\infty), 2, 1), \\ & \qquad
			( \locpriv, (1, +\infty), 2, 0),
			( \locpriv, 1, 2, 1),
			( \locpriv, 1, 2, 0), \\ & \qquad
			( \locfinalpriv, 0, 2, 1),
			( \locfinalpriv, 0, 2, 0)
			\big\} \\ %
			\belief{2} &= \belief{2}' \cup \big\{
				( \locfinal, 1 ,2, 1),
				( \locfinal, 1 ,2, 0),
				\big\}\\ %
			\belief{(2,3)}' & = \big\{
			( \locinit, (0,1), (2, +\infty), (0,1))
			\big\} \\
			\belief{(2,3)} &=
			\belief{(2,3)}' \cup\big\{
			( \locfinal, (0,1), (2, +\infty), (0,1))
			\big\}\\ %
			\belief{3} &=
			\belief{3}' \cup \big\{
			( \locfinal, 0, (2, +\infty), 1),
			( \locfinal, 0, (2, +\infty), 0),\\ & \qquad
			(\locfinal, 1, (2, +\infty), 1),
			(\locfinal, 1, (2, +\infty), 0)
			\big\}\\ %
			\belief{3}' & = \big\{
			( \locinit, 1, (2, +\infty), 1),
			( \locinit, 1, (2, +\infty), 0),
			( \locinit, 0, (2, +\infty), 1),\\ & \qquad
			( \locinit, 0, (2, +\infty), 0)
			\big\}%
		\end{align*}

	}

	\begin{figure*}[tb]
		\begin{center}
	   	\begin{subfigure}{0.25\textwidth}
	   		\begin{tikzpicture}[pta,  scale= .9, every node/.style={scale=1}]
	   			\node[location, initial, initial text =] (q0) {$\locinit$};
	   			\node[location, private] at (2.5,2) (qpriv) {$\locpriv$};f
	   			\node[location, final] at (2.5,-2) (qfin) {$\locfinal$};
			 	\node[invariant] at (0, 0.6) [] (invariant) {$\styleclock{x} \leq 1$};
	   			\node[invariant] at (2.5, 2.7) (invariant2) {$\styleclock{y} \leq 2$};

				\node at(0,-2.7) (phantom) {};

	   			\path[->] (q0) edge[loop below] node[below,align=center] {\shortstack{${\edgeicolor{1}}$ \\ $\styleclock{x} = 1$ \\ $\styleact{u}$ \\ $\styleclock{x} \assign 0 $}} (q0)
	   			edge node[above, sloped] {\shortstack{${\edgeicolor{2}}$ \\ $\styleclock{x} = 0$}} node[below, sloped] {$\styleact{u}$} (qpriv)
	   			edge node[below, sloped] {$\styleact{a}$} node[above, sloped]{${\edgeicolor{4}}$}(qfin)
	   			(qpriv) edge node[right, align=left] {\shortstack{${\edgeicolor{3}}$ \\ $\styleclock{x} = 0$}} node[left] {$\styleact{u}$} (qfin);

	   		\end{tikzpicture}
	   		\caption{TA $\exTAopaquebis$}
	   		\label{fig:aut:strategy-opaque-2clocks}
	   	\end{subfigure}
   \hfill
	   	\begin{subfigure}{.7\textwidth}
			   \begin{tikzpicture}[beliefautomaton]
				   \node[rct, initial, initial text = , initial where = left] (bot) {$\beliefInitState$};
				   \node[rct] at (0,2) (A) {$\belief{0}$};
				   \node[rct] at (0,-2) (B) {$\belief{0}'$};
				   \node[rct] at (2, 2) (C) {$\belief{(0,1)}$};
				   \node[rct] at (2, -2) (D) {$\belief{(0,1)}'$};
				   \node[rct] at (4, 2) (E) {$\belief{1}$};
				   \node[rct] at (4, -2) (F) {$\belief{1}'$};
				   \node[rct] at (6, 2) (G) {$\belief{(1,2)}$};
				   \node[rct] at (6, -2) (H) {$\belief{(1,2)}'$};
				   \node[rct] at (8, 2) (I) {$\belief{2}$};
				   \node[rct] at (8, -2) (J) {$\belief{2}'$};
				   \node[rct] at (10, 2) (K) {$\belief{(2,3)}$};
				   \node[rct] at (10, -2) (L) {$\belief{(2,3)}'$};
				   \node[rct] at (13.5, 0) (M) {$\belief{3}$};
				   \node[rct] at (15, 0) (N) {$\belief{3}'$};
				   \path[->]
				   (bot) edge node[left] {$\translabelinstant, \emptyset$} (B)
					   edge node[left] {$\translabelinstant, \{a\}$} (A)
				   (A) edge node[above, sloped] {$\translabelchge, \{a\}$} (C)
				   (A) edge node[above, sloped, pos = 0.3] {$\translabelchge, \emptyset$} (D)
				   (B) edge node[above, sloped, pos = 0.3] {$\translabelchge, \{a\}$} (C)
				   (B) edge node[below, sloped] {$\translabelchge, \emptyset$} (D)
				   (C) edge node[above, sloped] {$\translabelchge, \{a\}$} (E)
				   (C) edge node[above, sloped, pos = 0.8] {$\translabelchge, \emptyset$} (F)
				   (C) edge[bend right = 10] node[above, sloped] {$\translabelsame, \emptyset$} (D)
				   (D) edge node[above, sloped, pos=0.8] {$\translabelchge, \{a\}$} (E)
				   (D) edge node[below, sloped]  {$\translabelchge, \emptyset$} (F)
				   (D) edge[bend right = 10 ] node[above, sloped] {$\translabelsame, \{a\}$}(C)
				   (E) edge node[above, sloped] {$\translabelchge, \{a\}$} (G)
				   (E) edge node[above, sloped, pos=0.3] {$\translabelchge, \emptyset$} (H)
				   (F) edge node[above, sloped, pos = 0.3] {$\translabelchge, \{a \}$} (G)
				   (F) edge node[below, sloped] {$\translabelchge, \emptyset$} (H)
				   (G) edge node[above, sloped] {$\translabelchge, \{a\}$} (I)
				   (G) edge node[above, sloped, pos = .8] {$\translabelchge, \emptyset$} (J)
				   (G) edge[bend right = 10] node[above, sloped] {$\translabelsame, \emptyset$} (H)
				   (H) edge node[above, sloped, pos = .8] {$\translabelchge, \{a\}$} (I)
				   (H) edge node[below, sloped] {$\translabelchge, \emptyset$} (J)
				   (H) edge[bend right = 10] node[above, sloped] {$\translabelsame, \{a\}$} (G)
				   (I) edge node[above, sloped] {$\translabelchge, \{a\}$} (K)
				   (I) edge node[above, sloped, pos = .3] {$\translabelchge, \emptyset$} (L)
				   (J) edge node[above, sloped, pos =.3] {$\translabelchge, \{a\}$} (K)
				   (J) edge node[above, sloped] {$\translabelchge, \emptyset$} (L)
				   (K) edge[bend right = 10] node[below, sloped] {$\translabelchge, \{a\}$} (M)
				   (K) edge[bend left = 15] node[above, sloped] {$\translabelchge, \emptyset$} (N)
				   (K) edge[bend right = 10] node[above, sloped] {$\translabelsame, \emptyset$} (L)
				   (L) edge[bend right = 15] node[above, sloped] {$\translabelchge, \{a\}$} (M)
				   (L) edge[bend right = 15] node[below, sloped] {$\translabelchge, \emptyset$} (N)
				   (L) edge[bend right = 10] node[above, sloped] {$\translabelsame, \{a\}$} (K)
				   (M) edge[bend right = 15] node[below, sloped] {$\translabelchge, \{a\}$} (K)
				   (M) edge[bend right = 10] node[above, sloped] {$\translabelchge, \emptyset$} (L)
				   (N) edge[bend right = 35] node[above, sloped] {$\translabelchge, \{a\}$} (K)
				   (N) edge[bend left = 35] node[below, sloped] {$\translabelchge, \emptyset$} (L)
				   (C) edge[loop above] node[above] {$\translabelsame, \set{a}$} ()
				   (D) edge[loop below] node[below] {$\translabelsame, \emptyset$}()
				   (G) edge[loop above] node[above] {$\translabelsame, \set{a}$}()
				   (H) edge[loop below] node[below] {$\translabelsame, \emptyset$}()
				   (K) edge[loop above] node[above] {$\translabelsame, \set{a}$}()
				   (L) edge[loop below] node[below] {$\translabelsame, \emptyset$}()
				   ;
			   \end{tikzpicture}
			   \caption{\BeliefAutomaton{} $\BeliefAut{\exTAopaquebis}$}
			   \label{fig:belief-aut:strategy-opaque-2clocks}
\end{subfigure}
		   \caption{$\exTAopaquebis$ and the corresponding \beliefAutomaton{}}
		   \label{fig:strategy-opaque-2clocks}
	   \end{center}
   \end{figure*}

	\end{example}

\subsubsection{Controlled belief automaton and encountered beliefs}

We will introduce in \cref{definition:controlled-belief-automaton} a version of the \beliefAutomaton{} controlled by a \metaStrategy{}~$\metastrat$.
One transition of the controlled \beliefAutomaton{} will group all possible sequences of transitions made in the \beliefAutomaton{} between two strategy changes in~$\metastrat$.
We thus need to keep track, in the states, of the sequence of strategy choices made until that state is reached, together with the current belief.
To do so, we first introduce a notation:

\begin{itemize}
\item $\metastrat_{k,0} = \metastrat([k,k])$;
\item if $\metastrat((k,k+1))= (\substrat_1, \ldots, \substrat_{m_k})$, then for all $1 \leq i \leq m_k$, $\metastrat_{k,i} = \substrat_i$.
\end{itemize}

In an execution of the \beliefAutomaton{}, the time elapsed can be inferred from the number of actions of the form $(1,\cdot)$ that have been taken so far.
If this number is of the form~$2k$, then exactly $k$~time units have elapsed.
If it is of the form $2k+1$ then the time elapsed is within the interval \mbox{$(k,k+1)$}.
We thus need a function that, given a sequence of elements in $\{\translabelinstant, \translabelsame, \translabelchge\} \times 2^{\ActionSetC}$ stating when clock~$\clockextra$ has changed from one region to another, and the consecutive choices of enabled actions made by the \metaStrategy{} so far, gives the next choice the meta strategy will make and whether or not $\clockextra$~will change region next.

Given a sequence $\SeqTransitions \in (\{\translabelinstant, \translabelsame, \translabelchge\} \times 2 ^{\ActionSetC})^*$, denoting by $2k+k'$ (with $k \in \setN$ and $k' \in\{0,1\}$) the number of actions of the form $(1,\cdot)$ in~$\SeqTransitions$, by $i$ the length of the longest suffix of $\SeqTransitions$ without any action of the form~$(1,\cdot)$, and by~$m_k$ the length of the sequence $\metastrat((k,k+1))$, the function $\next{\metastrat}$ is defined by:

\begin{itemize}
\item $\next{\metastrat}(\silentaction)=(0,\metastrat_{0,0})$
\item if $\SeqTransitions \neq \silentaction$ and $k'=0$, $\next{\metastrat}(\SeqTransitions) = (1,\metastrat_{k,1})$
\item if $k'=1$ and $i<m_k$, $\next{\metastrat}(\SeqTransitions) = (\translabelsame,\metastrat_{k,i+1} )$
\item if $k'=1$ and $i=m_k$, $\next{\metastrat}(\SeqTransitions) = (\translabelchge,\metastrat_{k+1,0} )$
\end{itemize}

\begin{example}
	Let $\metastrat$ a \metaStrategy{} defined (on the interval $[0,1]$) by: $\metastrat_{0,0} = \activated_0$, $\metastrat_{0,1} = \activated_1$, $\metastrat_{0,2} = \activated_2$, $\metastrat_{1,0} = \activated_3$. 
	Then, 
	\begin{itemize}
		\item $\next{\metastrat}(\silentaction) = (\translabelinstant, \activated_0)$ (first case),
		\item $\next{\metastrat}((\translabelinstant, \activated_0)) = (\translabelchge, \activated_1)$ (second case),
		\item $\next{\metastrat}((\translabelinstant, \activated_0),(\translabelchge, \activated_1)) = (\translabelsame, \activated_2)$ (third case), 
		\item $\next{\metastrat}((\translabelinstant, \activated_0),(\translabelchge, \activated_1), (\translabelsame, \activated_2)) = (\translabelchge, \activated_3)$ (fourth case). 
	\end{itemize}
\end{example}
\begin{definition}[Controlled \beliefAutomaton{}]\label{definition:controlled-belief-automaton}
	Given a \beliefAutomaton{}~$\BeliefAut{\TA} = (\BeliefAutState{\TA}, \BeliefAutActions{\TA}, \beliefInitState, \BeliefAutTransitions{\TA})$ and a \metaStrategy{} $\metastrat$, we define
	$\BeliefControlAut{\TA}{\metastrat} = \big(\BeliefControlAutState{\TA}{\metastrat}, \BeliefControlAutActions{\TA}{\metastrat}, (\silentaction,\beliefInitState), \BeliefControlAutTransitions{\TA}{\metastrat} \big)$ the \emph{\beliefAutomaton{} controlled by~$\metastrat$} as follows:
	\begin{enumerate}
		\item $\BeliefControlAutState{\TA}{\metastrat} =
			(\BeliefAutActions{\TA})^* \times \BeliefAutState{\TA}$ is the set of states,

		\item $\BeliefControlAutActions{\TA}{\metastrat} = \BeliefAutActions{\TA} =  \{\translabelinstant, \translabelsame, \translabelchge\} \times 2 ^{\ActionSetC} $ is the alphabet,

		\item $(\silentaction,\beliefInitState)$ is the initial state,

		\item\label{definition:controlled-belief-automaton:item4} $\BeliefControlAutTransitions{\TA}{\metastrat} \subseteq (\BeliefControlAutState{\TA}{\metastrat} \times \BeliefControlAutActions{\TA}{\metastrat} \times \BeliefControlAutState{\TA}{\metastrat})$ and $ \big((\SeqTransitions, \belief{}), (\translabelany, \activated), (\SeqTransitions\cdot(\translabelany, \activated), \belief{}')\big) \in \BeliefControlAutTransitions{\TA}{\metastrat}$ if  $(\belief{}, (\translabelany, \activated), \belief{}') \in \BeliefAutTransitions{\TA}$, and $\next{\metastrat}(\SeqTransitions) = (\translabelany, \activated)$.
	\end{enumerate}
\end{definition}

In other words, the  \beliefAutomaton{} controlled by a \metaStrategy{}~$\metastrat$ retains  only the transitions from the \beliefAutomaton{} that correspond to the \metaStrategy{}.

We now define the beliefs encountered by the controlled \beliefAutomaton{} as sets obtained
as the union of every belief between a pair of choices of the form
$(\translabelchge, \activated)$, in other words, the beliefs on an integer timestamp, or 
by regrouping the beliefs visited during an open interval. We will see later that this
object is equal to the set of \intervalBeliefs{}~$\beliefcontrolset{\TA}{\metastrat}$.

\begin{definition}[Belief encountered by the controlled \beliefAutomaton{}]\label{definition:reach-belief-controlled-belief-automaton}
A belief $\belief{}$ is said to be encountered by a controlled \beliefAutomaton{}~$\BeliefControlAut{\TA}{\metastrat}$
if $((\translabelinstant, \activated) ,\belief{})$ is reachable in~$\BeliefControlAut{\TA}{\metastrat}$, or
there exists a sequence $(\SeqTransitions_0,\belief{0}),\dots (\SeqTransitions_m,\belief{m})$ of states of $\BeliefControlAutState{\TA}{\metastrat}$ such
that
\begin{itemize}
\item $\forall i\leq m, (\SeqTransitions_i,\belief{i})$ are reachable in $\BeliefControlAut{\TA}{\metastrat}$,
\item $\forall i\leq m, \next{\metastrat}(\SeqTransitions_i) = (\translabelany_i, \activated_i)$,  $\translabelany_i\in \{\translabelsame, \translabelchge\}$ with 
$\translabelany_i = \translabelchge$ iff $i=0$ or $i=m$,
\item $\forall i< m, \SeqTransitions_{i+1}= \SeqTransitions_i \cdot (\translabelany_i, \activated_i)$,
\item $\belief{}=\bigcup_{i=1}^m \belief{i}$.
\end{itemize}
The set of beliefs encountered by $\BeliefControlAut{\TA}{\metastrat}$ is denoted by~$\beliefencountset{\TA}{\metastrat}$.
\end{definition}

\subsubsection{Feasible runs}

Let us now relate a controlled \beliefAutomaton{} and runs of~$\TA$.
In the following definition, $\SeqTransitions$ is a sequence of subsets of controllable actions ( which will be associated later to a sequence of strategy choices).

\begin{definition}[Run admitting a sequence]\label{definition:admissible-run}
	Let $\TA$ be a TA, $\run$ be a run of $\TAdup$ and \mbox{$\SeqTransitions \in  (\BeliefAutActions{\TA})^* $}.
	We say that $\run$ \emph{admits} $\SeqTransitions$, denoted by $\run \vdash \SeqTransitions$, when either:
	\begin{description}
	\item[(run reduced to the initial state)] $\run = (\locinit, \ClocksZero)$ and $\SeqTransitions=(0,\activated_0)$ with $\activated_0 \subseteq \ActionSetC$,
	\\
	or

	\item[(normal run)] $\run = \run', (\timeEdge{n-1}), (\loci{n}, \clockval_n)$ with $\edge_{n-1} = (\loc_{n-1}, \guard, \action, \resets, \loc_n)$
	and one of the following holds:
	\begin{enumerate}
		\item $\paramd_{n-1} = 0$, $\run' \admits \SeqTransitions$,
	$\SeqTransitions=\SeqTransitions'\cdot(\translabelany,\activated)$ and \mbox{$\action \in \activated \cup \ActionSetU \cup \{\silentaction\}$,} \label{item:conditio1}%

		\item $0 < \paramd_{n-1} <1$, $\SeqTransitions = \SeqTransitions' \cdot (\translabelchge, \activated_0) \cdot (\translabelany_1, \activated_1) \cdots$ $ (\translabelany_{m-1}, \activated_{m-1})\cdot (\translabelany_m, \activated_m)$, $\action \in  \activated_m \cup \ActionSetU \cup \set{\silentaction}$, for all $1 \leq k < m$, $\translabelany_k=\translabelsame$, and 
		one of the following holds:\label{item:conditio2}
			\begin{enumerate}
				\item  $\fract{\mu_{n-1}(\clockextra)}\neq 0$, $\fract{\mu_{n}(\clockextra)}\neq 0$,  \mbox{$\translabelany_m = \translabelsame$}, and there exists $i, 0 \leq i \leq m$, $\run' \admits \SeqTransitions' \cdot (\translabelchge, \activated_0) \cdot (\translabelsame, \activated_1) \cdots (\translabelany', \activated_i)$, \label{item:conditio2a}%
				\item  $\fract{\mu_{n-1}(\clockextra)}=0$, $\run' \admits \SeqTransitions'$,
and either \mbox{$m=0$} or $\translabelany_m = \translabelsame$,\label{item:conditio2b}
				\item  $\fract{\mu_{n}(\clockextra)}=0$, $\translabelany_m = \translabelchge$, and there exists $i \in \{ 0 , \dots, m-1\}$ such that $\run' \admits \SeqTransitions' \cdot (\translabelchge, \activated_0) \cdot (\translabelsame, \activated_1) \cdots (\translabelany', \activated_i)$,\label{item:conditio2c}
			\end{enumerate}
		\item  $\paramd_{n-1} = 1$, 
		$\SeqTransitions = \SeqTransitions' \cdot (\translabelchge, \activated_0) \cdot (\translabelsame, \activated_1) \cdots (\translabelsame, \activated_m) \cdot (\translabelchge, \activated)$ with $ 0 \leq m$ and such that $\run' \admits \SeqTransitions'$, and $\action \in  \activated \cup \ActionSetU \cup \set{\silentaction}$\label{item:conditio3}.
	\end{enumerate}
\end{description}
\end{definition}

In the above definition, when the run has no transition, it is associated to a sequence of just one element, $(0,\activated_0)$ where 0 means that no time has passed, and $\activated_0$ is the first choice of actions. For a non-empty run, condition~\ref{item:conditio1} states that when two transitions of~$\run$ occur at the same time, the corresponding set of actions has to be the same (as there will be only one set of actions enabled at a given time instant).
Condition~\ref{item:conditio3} corresponds to two actions separated by exactly one time unit.
In this case (since clock~$\clockextra$ is reset at every integer time) the two actions occur at an integer time.
The sequence associated to~$\run'$ is completed with a sequence of pairs where the first and last one correspond to a change of region for~$z$ (showed by the ``$1$'' as first element).
Condition~\ref{item:conditio2} is more complex.
Let us denote $t_{n-1}$ the time instant of the end of~$\run'$, and $t_n$ the time instant of the end of~$\run$.
There is a sequence~$\SeqTransitions''$ such that $\run' \admits \SeqTransitions''$ and such that the following holds.
If both time instants have a non-0 fractional part (condition~\ref{item:conditio2a}), then we complete $\SeqTransitions''$ with a (possibly empty) sequence where every pair has a first component equal to~$\translabelsame$, as $\clockextra$ does not switch region outside integer time instants. If only $t_{n-1}$ is an integer (condition~\ref{item:conditio2b}), then we add to~$\SeqTransitions''$ a sequence starting with a region change (first component equal to~$\translabelchge$), and if only $t_n$ is an integer (condition~\ref{item:conditio2c}), then $\SeqTransitions''$ is completed with a non empty sequence ending by a region change (first component of the pair equals~$\translabelchge$).
\begin{definition}[Feasible run]\label{definition:feasible}
Let $\TA$ be a TA, $\run$ be a run of $\TAdup$ and $\metastrat$ a \metaStrategy{}.
We say that $\run$ is \emph{feasible in~$\BeliefControlAut{\TA}{\metastrat}$}
when there exist $\SeqTransitions\in (\BeliefAutActions{\TA})^*$ and a belief $\belief{} \in \BeliefAutState{\TA}$ such that $\run \admits \SeqTransitions$, $(\SeqTransitions, \belief{})$ is reachable in $\BeliefControlAut{\TA}{\metastrat}$ and $\region = \class{\laststate(\run)} \in \belief{}$.
\end{definition}

Note that, from item~\ref{definition:controlled-belief-automaton:item4} of \cref{definition:controlled-belief-automaton}, there is only one action possible in each belief of $\BeliefControlAut{\TA}{\metastrat}$, and thus only one execution.
For every feasible run~$\run$ the corresponding $\SeqTransitions$ is then of the form $(\translabelinstant,\metastrat_{0,0}) \cdot (\translabelchge,\metastrat_{0,1})\cdot (\translabelsame,\metastrat_{0,2}) \cdots (\translabelsame,\metastrat_{0,m_0}) \cdot (\translabelchge,\metastrat_{1,0}) \cdot (\translabelchge,\metastrat_{1,1})\cdot (\translabelsame,\metastrat_{1,2})\cdots (\translabelsame,\metastrat_{1,m_1})\cdots$, where for all~$k$, $\metastrat((k,k+1))=\metastrat_{k,1},\dots,\metastrat_{k,m_k}$ with $m_k \geq 1$.

\begin{example}\label{example:controlled-aut-beliefs}
	Consider again the TA~$\exTAopaque$ in \cref{fig:aut:strategy-opaque}.
	Let $\strategy$ be the strategy defined as follows:
	\[
		\strategy(\tau) =
		\left\{
			\begin{array}{ll}
				\{a\} & \mbox{ if } \tau \in \setN \\
				\emptyset & \mbox{ otherwise }
			\end{array}
		\right.
	\]
	$\exTAopaque$ is \fullOpaqueText{} with~$\strategy$.

	Let $\metastrat$ be a \metaStrategy{} defined as follows:
		\begin{itemize}
			\item $\metastrat_{i,0} = \{a\}$, for all $i \leq 0$, 
			\item $\metastrat_{i,1} = \emptyset$, for all $i \leq 0$. 
		\end{itemize}
	First note that $\strategy \satisfies \metastrat$.
	Then, the first states of the automaton $\BeliefControlAut{\exTAopaque}{\metastrat}$ are depicted in \cref{fig:ex:belief-strat}.

	\begin{figure*}[tb]
		\begin{center}

		\begin{tikzpicture}[auto]
			\begin{scope}[beliefautomaton]
				\node[rct, initial, initial text =] at (0,0) (bot) {$(\silentaction, \beliefInitState)$};
				\node[rct] at (3.5, 0) (b0) {$\big((\translabelinstant, \{a\}), \belief{0}\big)$};
				\node[rct] at (8, 0) (b01) {$\big((\translabelinstant, \{a\}) \cdot (\translabelchge, \emptyset), \belief{(0,1)}'\big)$};
				\node[rct] at (8, -2) (b1) {$\big((\translabelinstant, \{a\}) \cdot (\translabelchge, \emptyset) \cdot (\translabelchge, \{a\}), \belief{1}\big)$};
				\node[rct] at (1, -2) (b01-2) {$\big((\translabelinstant, \{a\}) \cdot (\translabelchge, \emptyset) \cdot (\translabelchge, \{a\})\cdot (\translabelchge, \emptyset), \belief{(0,1)}'\big)$};
				\node at (-4, -2) (fantome) {};
				\path[->]
				(bot) edge node[] {$\translabelinstant, \{a\}$} (b0)
				(b0) edge node[] {$\translabelchge, \emptyset$} (b01)
				(b01) edge node[] {$\translabelchge, \{a\}$} (b1)
				(b1) edge node[above]  {$\translabelchge, \emptyset$} (b01-2)
				(b01-2) edge[dotted] node[above]{$\translabelchge, \{a\}$} (fantome)
				;

			\end{scope}
		\end{tikzpicture}
		\caption{$\BeliefControlAut{\exTAopaque}{\metastrat}$: First states of the controlled \beliefAutomaton{} for TA $\exTAopaque$ and \metaStrategy{} $\metastrat$ }
		\label{fig:ex:belief-strat}
	\end{center}
	\end{figure*}
	\noindent
	Runs $\run_1$ and $\run_2$ are feasible in $\BeliefControlAut{\exTAopaque}{\metastrat}$:
	\[\begin{array}{lll}
		\run_1 & = & (\locinit, 0), (1, \edge_1), (\locinit, 0), (0, \mathsf{e_2}), (\locpriv, 0),\\ && \quad (0, \edge_3), (\locfinalpriv, 0) \\
		\run_2 & = & (\locinit, 0), (1, \edge_1), (\locinit, 0), (0, \mathsf{e_4}), (\locfinal, 0)
	\end{array}\]
	\noindent For $\SeqTransitions = (0, \set{a}) \cdot (1, \emptyset) \cdot (1, \set{a})$ a sequence of actions in the automaton $\BeliefControlAut{\exTAopaque}{\metastrat}$, we have $\run_1 \admits \SeqTransitions$ and $\run_2 \admits \SeqTransitions$.
\end{example}

We can now link a run and a \metaStrategy{}.
\begin{restatable}[Strategies and run feasibility]{lemma}{lemmastrategies}\label{lem:strat-runs-equivalence}
	Let $\TA$ be a TA, $\run$ a run of $\TAdup$ and $\metastrat$ a \metaStrategy{}.
	There exist $\strategy \satisfies \metastrat$ such that $\run$ is \compatible{} iff $\run$ is feasible in $\BeliefControlAut{\TA}{\metastrat}$.
\end{restatable}
\begin{proof}
	\begin{itemize}
		\item [$\implies$]
	Let $\strategy \satisfies \metastrat$ be a strategy and $\run$ be a \compatible{} run of $\TAdup$ s.t. \mbox{$\run = (\locinit,\ClocksZero), (\timeEdge{0}), \ldots, (\timeEdge{n-1}), (\loci{n}, \clockval_n)$}.
		We build inductively from $\run$ a sequence $\SeqTransitions \in (\BeliefAutActions{\TA})^*$ such that $\run \admits \SeqTransitions$ and the sequence $\SeqTransitions$ can be fired in $\BeliefControlAut{\TA}{\metastrat}$.\\
		As $\run$ is \compatible{}, for all $0 \leq i < n$, $(\loc_i, \clockval_i, \sum_{j < i} \paramd_j) \transitionControlLabel{i}{\strategy} (\loc_{i+1}, \clockval_{i+1}, \sum_{j \leq i} \paramd_j)$.

		First for $\run_0=(\locinit,\ClocksZero)$, we set ${\SeqTransitions}_0=(0,\metastrat([0,0]))=\next{\metastrat}(\silentaction)$. By construction of the controlled \beliefAutomaton, $\next{\metastrat}(\silentaction)$ is precisely the action allowed in the initial state of $\BeliefControlAut{\TA}{\metastrat}$ and there exists a belief $\belief{0}$ such that $(\SeqTransitions_0,\belief{0})$ is reachable in $\BeliefControlAut{\TA}{\metastrat}$.

		Now, let us assume that we have built a sequence ${\SeqTransitions}_i$ such that $\run_i=(\locinit,\ClocksZero), (\timeEdge{0}), \ldots, (\timeEdge{i-1}), (\loci{i}, \clockval_i) \admits {\SeqTransitions}_i$ with  $i<n$, and a belief $\belief{i}$ such that $(\SeqTransitions_i,\belief{i})$ is reachable in $\BeliefControlAut{\TA}{\metastrat}$. We define ${\SeqTransitions}_{i+1}$ by completing ${\SeqTransitions}_{i}$ with all the actions (in order) of the controlled \beliefAutomaton{} corresponding to the changes in the strategy~$\strategy$ between time instants $\sum_{j < i} \paramd_j$ excluded and $\sum_{j \leq i} \paramd_j$ included.

		More precisely, following the definition of admissible sequence, we have three possible cases:
		\begin{itemize}
			\item $\paramd_i=0$, in which case the strategy has not changed since the two consecutive transitions occur at the same time. We thus set ${\SeqTransitions}_{i+1}={\SeqTransitions}_i$. 
			
			\item $\paramd_i=1$, the time between two resets of clock~$z$. In this case $\sum_{j<i} \paramd_j = k \in \setN$\footnote{Since clock~$\clockextra$ is reset precisely every time unit, $\fract{\clockval_{i}(\clockextra)}=\fract{\sum_{j < i} \paramd_j}$ and so $\sum_{j < i} \paramd_j \in \setN$ iff $\fract{\clockval_{i}(\clockextra)}=0$.} and, if $\metastrat((k,k+1)) = \metastrat_{k,1}, \ldots, \metastrat_{k,m}$, we set ${\SeqTransitions}_{i+1}={\SeqTransitions}_i \cdot (\translabelchge,\metastrat_{k,1}), (\translabelsame,\metastrat_{k,2}),\ldots,$ $(\translabelsame,\metastrat_{k,m}),(\translabelchge,\metastrat_{k+1,0})$. This construction respects the properties of \cref{definition:admissible-run} (since $\strategy \satisfies \metastrat$, the last action of $\run_{i+1}$ fired at time $k+1$ has to belong to $\metastrat_{k+1,0}$), so we have $\run_{i+1} \admits {\SeqTransitions}_{i+1}$.

			\item $0<\paramd_i<1$. Let $k$ be the integral part of $\sum_{j < i} \paramd_j$, $\metastrat((k,k+1)) = \metastrat_{k,1}, \ldots, \metastrat_{k,m}$, $\subinterval_1,  \ldots, \subinterval_m$ the ordered partition of $(k,k+1)$ corresponding to the changes of strategy for $\strategy$ in this interval and $\subinterval_0=[k,k]$, $\subinterval_{m+1}=[k+1, k+1]$. Denoting $l$ the index such that $\sum_{j < i} \paramd_j \in \subinterval_l$ and $l'$ the index such that $\sum_{j \leq i} \paramd_j \in \subinterval_{l'}$, then if $l=l'$ we have $\SeqTransitions_{i+1} = \SeqTransitions_{i}$, otherwise  $\SeqTransitions_{i+1} = \SeqTransitions_{i}\cdot (\translabelany_{l+1},\metastrat_{k,l+1}) \cdots (\translabelany_{l'},\metastrat_{k,l'})$ where:
			\begin{itemize}
				\item if $l'=m+1$, we define $\metastrat_{k,m+1}$ as $\metastrat_{k+1,0}$;
				\item $\translabelany_1 = \translabelany_{m+1} = \translabelchge$, and for all values $1<l''<m+1$, $\translabelany_{l''}=\translabelsame$.
			\end{itemize}
			In every case above (whether $l'=m+1$ or not), since $\strategy \satisfies \metastrat$ and given our choice of $l'$ such that the next action takes place during $\subinterval_{l'}$, we ensure that the next action $a_i$ of $\edge_i$ belongs to $\metastrat_{k,l'}$.
		\end{itemize}

		This construction of $\SeqTransitions$ thus respects at each step the definition of admissibility and thus $\rho \admits \SeqTransitions$. Furthermore, our construction also ensures that, for every $\SeqTransitions'$, $\SeqTransitions''$ such that $\SeqTransitions= \SeqTransitions'\cdot(\translabelany,\substrat) \cdot\SeqTransitions''$, $\next{\metastrat}(\SeqTransitions')=(\translabelany,\substrat)$. Thus, by definition of the controlled \beliefAutomaton{}, we have that the actions of $\SeqTransitions$ can be fired in this order in $\BeliefControlAut{\TA}{\metastrat}$ and hence the existence of a belief $\beliefNoArg$ such that $(\SeqTransitions,\belief{})$ is reachable in $\BeliefControlAut{\TA}{\metastrat}$. Since the sequence $\SeqTransitions$ has been built precisely to follow all the strategy's changes that have occured during $\run$, we also get that the belief $\beliefNoArg$ contains the last state of $\run$ and thus $\region = \class{\laststate(\run)} \in \beliefNoArg$.
		\item [$\impliedby$]
		Let $\run$ be a feasible run in $\BeliefControlAut{\TA}{\metastrat}$. By definition there is  a sequence $\SeqTransitions$ such that $\run \admits \SeqTransitions$ and such that all actions of $\SeqTransitions$ can be executed in order in $\BeliefControlAut{\TA}{\metastrat}$. As claimed before, this sequence is of the form $(\translabelinstant,\metastrat_{0,0})\cdot(\translabelchge,\metastrat_{0,1})\cdot(\translabelsame,\metastrat_{0,2}) \cdots (\translabelsame,\metastrat_{0,m_0})\cdot$ $(\translabelchge,\metastrat_{1,0})\cdot(\translabelchge,\metastrat_{1,1})\cdot(\translabelsame,\metastrat_{1,2})\cdots(\translabelsame,\metastrat_{1,m_1})$ where for all $k$, $\metastrat((k,k+1))=\metastrat_{k,1},\ldots,\metastrat_{k,m_k}$ with $m_k \geq 1$. We will now construct a strategy $\strategy$ such that $\strategy \satisfies \metastrat$ and $\run$ is \compatible{}.

		Let $\strategy$ be a strategy such that $\forall k \in \setN$:
		\begin{itemize}
			\item $\strategy(k) = \metastrat([k,k])$;
			\item if no event of $\run$ occur in interval \mbox{$(k,k+1)$} then, with~$m_k$ the number of strategy changes for $\metastrat$ during this interval, we set $\subinterval_{1} = (k,k+1/m_k)$ and for all \mbox{$1<j\leq m_k$}, $\subinterval_j=$\mbox{$[k+(j-1)/m_k,$} \mbox{$k+j/m_k)$}. The strategy $\strategy$ is then defined on each $\subinterval_j$ by $\metastrat_{k,j}$;
			\item if one or more action of $\run$ occur in interval $(k,k+1)$ then we need to constrain the ordered partition so that the resulting $\strategy$ matches those actions. The construction of $\SeqTransitions$ precisely tells to which choice of the \metaStrategy{} correspond each action of $\run$. We thus get, for every $j \leq m_k$, a (possibly empty) set of time instants at which transitions occur in $\run$ for the strategy choice $\metastrat_{k,j}$. If the set is not empty, we note $t_j$ the first of those instants.
			If the set is empty, we define an ``artificial'' $t_j$ to help build the ordered partition afterwards. We denote by $t'_j$ the global time of the last event that occured ``before'', \ie{} associated to a $(\translabelany,\metastrat_{k,j'})$ in $\SeqTransitions$ with $j'<j$ (in case no such event exists within $(k,k+1)$, we set $t'_j=k$ and $j'=0$), by $t''_j$ the global time of the first event that occured ``after'', \ie{} associated to an element $(\translabelany,\metastrat_{k,j''})$ in $\SeqTransitions$ with $j''>j$, (in case no such event exists within $(k,k+1)$, we set $t''_j=k+1$ and $j''=m_k+1$), we can then set \mbox{$t_j=t'_j+(t''_j-t'_j)*(j-j')/(j''-j')$}. This formula gives uniformly distributed artificial $t_j$s. Based on these time instants, we can define an ordered partition $\subinterval_1, \cdots, \subinterval_{m_k}$ with $t_{m_k+1}=k+1$:
			\begin{itemize}
				\item $\subinterval_1 = (k,t_2)$;
				\item for all $1<j \leq m_k$, $\subinterval_j = [t_j,t_{j+1})$. 
			\end{itemize} 
			The strategy $\strategy$ is then defined on each $\subinterval_j$ by $\metastrat_{k,j}$.
		\end{itemize}
		We have thus defined a strategy $\strategy \satisfies \metastrat$ (both make the same choices at every integer time~$k$ and in the same order in every interval of the form $(k,k+1)$) and such that $\run$ is \compatible{}, which concludes the proof.\qedhere{}
	\end{itemize}
\end{proof}

\begin{corollary}\label{cor:encint}
Let $\TA$ be a TA, and $\metastrat$ a \metaStrategy{},
$\beliefencountset{\TA}{\metastrat}=\beliefcontrolset{\TA}{\metastrat}$.
\end{corollary}

\begin{proof}
	Let $\TA$ be a TA, and $\metastrat$ a \metaStrategy{}. 
	Let $\belief{\interval}\in \beliefcontrolset{\TA}{\metastrat}$.
We focus on the case where there exists $k\in \setN$, such that 
$\belief{\interval}=\belief{k+}^{\metastrat}$
(the interval belief associated to $\interval = (k,k+1)$) as the case where
$\belief{\interval}= \belief{k}^{\metastrat}$ (the interval belief associated to $\{k\}$) is simpler and can be obtained with a similar proof.
Let $(\SeqTransitions_0,\belief{0}),\dots (\SeqTransitions_m,\belief{m})$ be the sequence
used in~\cref{definition:reach-belief-controlled-belief-automaton} where
$\SeqTransitions_1,\dots \SeqTransitions_m$ all contain $2k+1$ elements of the form
$(\translabelchge, \activated)$.
Let $\belief{\mathcal{E}} = \bigcup_{i=1}^m\belief{i}\in \beliefencountset{\TA}{\metastrat}$.
Let us show that $\belief{\mathcal{E}}=\belief{\interval}$.

Given $\region\in \belief{\interval}$, then by definition of $\beliefcontrolset{\TA}{\metastrat}$, there exists a run $\run$ such that $\run$ is \compatible{} with $\strategy \satisfies \metastrat$, $\regionclass{\laststate(\run)} =\region$ 
and $\duration(\run)\in \interval$.
Thus, by \cref{lem:strat-runs-equivalence}, $\run$ is feasible in $\BeliefControlAut{\TA}{\metastrat}$. 
By \cref{definition:feasible},
there exists $\SeqTransitions$ and $\belief{}$ such that $\run \admits \SeqTransitions$, and 
$(\SeqTransitions, \belief{})$ is reachable in $\BeliefControlAut{\TA}{\metastrat}$.
Moreover, $\duration(\run)\in \interval$ is equivalent to $\SeqTransitions$ containing 
exactly $2k+1$ elements of the form $(\translabelchge, \activated)$.
Hence, there exists $1\leq i\leq m$ such that 
$(\SeqTransitions, \belief{})=(\SeqTransitions_i, \belief{i})$.
Hence, $\region\in \belief{i}\subseteq \belief{\mathcal{E}}$.

Conversely, given $\region\in \belief{\mathcal{E}}$, there exists $1\leq i\leq m$ 
such that $\region\in \belief{i}$. By following $\SeqTransitions_i$ in the 
\BeliefAutomaton{}, we can build a feasible run $\rho$ admitting $\SeqTransitions_i$ and
such that $\regionclass{\laststate(\run)} =\region$. 
Thus, by \cref{lem:strat-runs-equivalence}, there exists $\strategy \satisfies \metastrat$, 
such that $\run$ is \compatible{}. Again, as $\SeqTransitions_i$ contains
exactly $2k+1$ elements of the form $(\translabelchge, \activated)$, 
$\duration(\run)\in \interval$. Hence, by definition of $\beliefcontrolset{\TA}{\metastrat}$, $\region\in \belief{\interval}$.
\end{proof}

\subsection{Justifying \metaStrategies{}}

The following example shows why the results in~\cite{ADLL24} were distorted, and why adding the notion of \metaStrategy{} corrects this.

\begin{figure}
	\begin{center}
		\begin{tikzpicture}[ pta,  scale= .75, every node/.style={scale=1}]
			\node[location, initial, initial text =] at (-1, 0)(q0) {$\locinit$};
			\node[location] at (1.8, 2.5) (q1) {$\loci{1}$};
			\node[location] at (4.6, 2.5) (q2) {$\loci{2}$};
			\node[location] at (1.8, 0) (q3) {$\loci{3}$};
			\node[location, private] at (4.6,0) (qpriv) {${\locpriv}$};
			\node[location, final] at (7.4,0) (qfin) {$\locfinal$};

			\path[->] (q0) edge[bend left=30] node[above, sloped, align=right] {\shortstack{${\edgeicolor{1}}$ \\ $0 < \styleclock{x} < 1$}}  node[below,sloped] {$\styleact{a}$ } (q1)
				edge node[above, align=center] {\shortstack{${\edgeicolor{4}}$ \\ $0<\styleclock{x}<1$}}node[below, align=center] {\shortstack{$\styleact{a} $\\  $\styleclock{x} \assign 0$}} (q3)
				edge[bend right=65] node[below] {$\styleact{u_1}$} (qpriv);
			\path[->] (q1) edge node[above, align=center] {\shortstack{${\edgeicolor{2}}$ \\ $1< \styleclock{x} <2$ }} node[below] {$\styleact{b}$} (q2);
			\path[->] (q2) edge[bend left=30] node[above, sloped, align=left] {\shortstack{${\edgeicolor{3}}$ \\ $ \styleclock{x} = 2$}} node[below, sloped] {$\styleact{u_2}$} (qfin);
			\path[->] (q3) edge node [above] {\shortstack{${\edgeicolor{5}}$ \\ $ \styleclock{x} = 1$}} node[below, align=center] {$\styleact{b}$  } (qpriv);
			\path[->] (qpriv) edge node[below] {$\styleact{u_3}$} node[above]{\shortstack{${\edgeicolor{6}}$ \\ $\styleclock{x} = 2$ }} (qfin);

			\node[] at (1.8, -1.4) (e7) {${\edgeicolor{7}}$};

		\end{tikzpicture}
		\caption{TA $\exTAcounterex$}
		\label{fig:aut:counter-example}
	\end{center}
\end{figure}

\begin{example}
Let $\exTAcounterex$ be the TA depicted in \cref{fig:aut:counter-example}. The upper path, through locations $\loci{1}$ and~$\loci{2}$, allows public runs of duration of 2 units of time.
The lower path, through transition~$u_1$ and location~$\locpriv$, allows private runs of duration of 2 time units.
Finally, the middle path, through transitions $a$ and~$b$, and location ${\locpriv}$, allows private runs of durations in $(2,3)$.

Because of the reset of~$x$ on the transition between $\locinit$ and~$\loci{3}$, finding a strategy that makes $\exTAcounterex$ opaque means finding a strategy that prevent reaching ${\locpriv}$ by~$\loci{3}$.
An acceptable strategy allows~$a$ somewhen in $(0,1)$ and $b$ somewhen in $(1,2)$ but prohibits $b$ if $a$ was allowed 1 time unit before.

Then, we can define a strategy~$\strategy$ that makes $\exTAcounterex$ opaque, for example:
\[
    \strategy(\tau) =
    \left\{
        \begin{array}{ll}
            \{a\} & \mbox{ if } \tau \in (0,0.3) \\
            \{b\} & \mbox{ if } \tau \in (1.3,2) \\
            \emptyset & \mbox{ otherwise }
        \end{array}
    \right.
\]

Let $\metastrat$ the \metaStrategy{} such that $\strategy \satisfies \metastrat$. Then, the only path in $\BeliefControlAut{\TA_{counterex}}{\metastrat}$ is $\SeqTransitions = (\translabelinstant, \emptyset) \cdot (\translabelchge, \{a\}) \cdot(\translabelsame, \emptyset) \cdot $ $(\translabelchge, \emptyset) \cdot (\translabelchge, \emptyset) \cdot (\translabelsame, \{b\})$.

But this \metaStrategy{} allows to reach a belief containing a final private region but not a public one. Indeed, the run
$\run = (\locinit, 0), (0.2,\edge{_4}),(\loci{3}, 0),$ $(1, \edge{_5}), (\locpriv, 1), (1, \edge{_6}), (\locfinal, 2)$, with $\runduration{\run}= 2.2$, is feasible in $\BeliefControlAut{\exTAcounterex}{\metastrat}$ but is not \compatible{} (but there exists $\strategy' \satisfies \metastrat$ such that it is $\strategy'$-compatible.)

Finally, the conclusion is that the \metaStrategy{} (named \bStrategy{} in~\cite{ADLL24}) is less precise than the strategy can be, \ie{} the \metaStrategy{} can allow behaviours that the strategy prevent, which is why we are working here only with \metaStrategies{}. 

\end{example}

\section{Solving \opacityText{} problems through the \beliefAutomaton{}}\label{section:solution}

From \cref{theorem:opacity-leaking-belief} and \cref{cor:encint}, a \metaStrategy{} ensures \fullOpacityText{}
if it avoids \badBelief{} encountered beliefs in the controlled \beliefAutomaton{}.
Intuitively, finding such a \metaStrategy{} amounts
to solving a one-player game on the \beliefAutomaton{}, where one needs to infinitely often select an action of the form $(\translabelchge, \activated)$ (in order for time to progress) and avoid the \badBelief{} beliefs. 
We will thus translate this into solving a one-player Büchi game.

More precisely, a one-player Büchi game can be defined by a tuple
$\game = (\gameStateSet,\gameInitState, \gameEvent,\gameTransitions,\gameGood)$ where $\gameStateSet$ is a set of states, $\gameInitState\in \gameStateSet$ is the initial state,
$\gameEvent$ is a set of actions, $\gameTransitions \subseteq \gameStateSet\times \gameEvent\times \gameStateSet$ describes the transitions, and $\gameGood\subseteq \gameEvent$ is a set of ``good'' actions.
Starting from~$\gameInitState$, at each step, the player selects a transition from~$\gameTransitions$ to reach a new state.
The player wins if transitions labelled by actions from $\gameGood$ are taken infinitely often (note that Büchi games usually require that a set of ``good'' \emph{states} is visited infinitely often instead of actions, but both frameworks are trivially equivalent).

\begin{lemma}[\cite{VARDI19941}]\label{lem:solving-belief}
Deciding the existence of a winning strategy in a one-player Büchi game can be done in
\NLOGSPACE. Moreover this strategy, if it exists, can be constructed in polynomial time.
\end{lemma}

As shown in~\cite{ALLMS23}, the \fullOpacityText{} problem for timed automata (without control) is decidable in \NEXPTIME. The ability to control the system slightly increases the complexity:

\begin{restatable}{theorem}{theoremfullproblems}\label{theorem:full-problems}
	The \fullOpacityText{} \metaStrategy{} emptiness problem is decidable
	in \EXPSPACE;
	and the \fullOpacityText{} \metaStrategy{} synthesis problem is solvable in \twoEXPTIME.
\end{restatable}
\begin{proof}
Given a TA~$\TA$, using the \beliefAutomaton{} of~$\TA$ as a basis, we define the
one-player Büchi game $\game = ((\BeliefAutState{\TA})^2, (\beliefInitState,\emptyset), \BeliefAutActions{\TA}, \delta_{\BeliefAutTransitions{\TA}}, \gameGood) $ where
$\gameGood = \{(\translabelany, \activated) \in \BeliefAutActions{\TA} \mid \translabelany=1 \}$ and $((\belief{1},\belief{2}), (\translabelany, \activated), (\belief{3}, \belief{4}))\in \delta_{\BeliefAutTransitions{\TA}}$ iff
\begin{itemize}
\item $(\belief{1}, (\translabelany, \activated), \belief{3})\in \BeliefAutTransitions{\TA}$,
\item $\belief{4} = \belief{3}$ if $\translabelany=\translabelchge$, 
$\belief{4} = \belief{2}\cup \belief{3}$ otherwise,
\item if $\translabelany=\translabelchge$, then $\belief{2}$ is not a \badBelief{} belief for \fullOpacityText{}.
\end{itemize}
In other words, we manipulate pairs of belief, the first corresponding to the current state
of the \beliefAutomaton{}, while the second accumulates the belief.
It is important to note that when an action labelled by a $\translabelchge$ is taken,
this second component contains the encountered belief corresponding to the current interval.
Hence, in~$\game$, we cannot leave the current interval if the
encountered belief is \badBelief{}.

First, assume that a given \TAtext{} \TA{} is \fullOpaqueText{} for a 
\metaStrategy{}~$\metastrat$. By \cref{theorem:opacity-leaking-belief} and
\cref{cor:encint}, it means $\metastrat$ avoids \badBelief{} encountered beliefs in the 
controlled \beliefAutomaton{}. Hence, it can be applied within $\game$, as in
the controlled \beliefAutomaton{}, and the removal of the transitions from \badBelief{} \intervalBeliefs{} does not affect it.
By definition of a \metaStrategy{}, $\metastrat$ selects infinitely often an action from~$\gameGood$, ensuring the strategy is winning.
Conversely, since a winning strategy of~$\game$ has to take ``good'' actions infinitely often, the number of consecutive ``non good'' actions (matching the changes in the meta strategy within an interval $(k,k+1)$) is finite, and it directly entails a \metaStrategy{}~$\metastrat$ which, being winning in~$\game$, does not encounter a \badBelief{} belief. Hence,
again by \cref{theorem:opacity-leaking-belief} and
\cref{cor:encint}, \TA{} is \fullOpaqueText{} with~$\metastrat$.

	By \cref{lem:solving-belief}, the existence of a winning strategy in this game is decidable in \NLOGSPACE in the size of the game. As $\game$ is based on the \beliefAutomaton{}, it consists in a form of determinisation of the labelled Region Automaton. The latter is exponential in the size of the \TAtext{}, and the determinisation
can produce a second exponential, hence $\game$ is at most doubly exponential.
Hence, solving $\game$ is in \EXPSPACE (the non-determinism can freely be removed thanks to Savitch theorem, which implies that \EXPSPACE = \NEXPSPACE). Moreover, through the computation of a solution in~$\game$, one directly obtains that the \fullOpacityText{} \metaStrategy{} synthesis problem is solvable in \twoEXPTIME.
	\end{proof}
\section{Weak and existential \opacityText{}}\label{section:weak-existential-opacity}

In this Section as well as in \cref{section:robust-extensions}, we will consider 
a few variants opacity notions. In most cases, the structure of the proofs are similar to 
the one for \fullOpacityText{} and as such is not repeated here.

More precisely, in this Section we study two other versions of opacity~\cite{ALLMS23}, namely \weakOpacityText{} (in which it is harmless that the attacker deduces that the private location was \emph{not} visited) and \existentialOpacityText{} (in which we are simply interested in the \emph{existence} of one execution time for which opacity is ensured). 
In \cref{section:robust-extensions} we will consider variants relaxing the accuraty of the measure of the attacker.

\subsection{Definitions}\label{appendix:def-weak-existential}

We recall definitions of weak and existential \opacityText{} from \cite{ALLMS23}.

\begin{definition}[Weak \opacityText{}]\label{def:weak-opacity}
	A \TAtext{}~$\TA$ is \emph{\weakOpaqueText{}} when $\PrivDurVisit{\TA} \subseteq \PubDurVisit{\TA}$.
\end{definition}
\begin{definition}[Existential \opacityText{}]\label{def:exists-opacity}
	A \TAtext{}~$\TA$ is \emph{existentially \opaqueText{}} (or \emph{\existentialOpaqueText{}}) when $\PrivDurVisit{\TA} \cap \PubDurVisit{\TA} \neq \emptyset$.
\end{definition}

That is, the TA is \weakOpaqueText{} whenever, for any run of duration~$\paramd$ reaching a final location after visiting~$\locpriv$, there exists another run of the same duration reaching a final location but not visiting the private location.
In addition, whenever there is at least one private run such that there exists a public run of the same duration, the TA is \existentialOpaqueText{}.

\begin{example}\label{example:weak-opacity}
	We have seen in \cref{example:full-weak-opacity} that $\TAi{1}$ (given in \cref{fig:aut:non-strategy-opaque}) is not \fullOpaqueText{}.
	However, $\TAi{1}$ is \existentialOpaqueText{} since, for example, $1 \in \PrivDurVisit{\TA} \cap \PubDurVisit{\TA} \neq \emptyset$.
	Furthermore, since we can reach $\locfinal$ at any time without visiting~$\locpriv$ (and therefore $\PubDurVisit{\TAi{1}} = [0, \infty)$), it holds that $\PrivDurVisit{\TAi{1}} \subseteq \PubDurVisit{\TAi{1}}$ and $\TAi{1}$ is therefore \weakOpaqueText{}.
\end{example}
\begin{example}[\weakOpaqueText{} TA]
	Consider again the TA~$\exTAopaque$ in \cref{fig:aut:strategy-opaque}.
	Recall from \cref{example:opaque-TA} that strategy~$\strategy_1$ is such that
	$\forall \tau \in \setRgeqzero , \strategy(\tau) = \{ a \}$, \ie{} $a$ is allowed anytime.
	Recall that $\PrivDurVisitStrat{\TA}{\strategy_1} = \setN$ while $\PubDurVisitStrat{\TA}{\strategy_1} = \setRgeqzero$.
	Therefore $\PrivDurVisitStrat{\TA}{\strategy_1} \subseteq \PubDurVisitStrat{\TA}{\strategy_1}$, and hence $\exTAopaque$ is \weakOpaqueText{} with~$\strategy_1$.
\end{example}
\paragraph*{Strategies and weak and existential \opacityText{}}

We lift \cref{def:opacity-strategy} to weak and existential \opacityText{}.

\begin{definition}[Weak and existential \mbox{\opacityText{}} with a strategy]\label{def:opacity-strategy-weak}
	Given a strategy~$\strategy$, a~TA~$\TA$ is \emph{\weakOpaqueText{} with~$\strategy$} whenever $\PrivDurVisitStrat{\TA}{\strategy} \subseteq \PubDurVisitStrat{\TA}{\strategy}$.

	In addition, $\TA$ is \emph{\existentialOpaqueText{} with~$\strategy$} whenever
	$\PrivDurVisitStrat{\TA}{\strategy} \cap \PubDurVisitStrat{\TA}{\strategy} \neq \emptyset$.
\end{definition}

Similarly, we now lift \cref{def:opacity-metastrategy} to weak and existential \opacityText{}.

\begin{definition}[Weak and existential \opacityText{} with a \metaStrategy{}]\label{def:opacity-metastrategy-weak}
	Given a \metaStrategy{}~$\metastrat$, a TA~$\TA$ is \emph{\weakOpaqueText{} with~$\metastrat$} whenever
	$\PrivDurVisitStrat{\TA}{\metastrat} \subseteq \PubDurVisitStrat{\TA}{\metastrat}$.

	In addition, $\TA$ is \emph{\existentialOpaqueText{} with~$\metastrat$} whenever
	$\PrivDurVisitStrat{\TA}{\metastrat} \cap \PubDurVisitStrat{\TA}{\metastrat} \neq \emptyset$.
\end{definition}
\paragraph*{Problems}
We consider the following variations of the emptiness problems defined in \cref{paragraph:problems}.

\defProblem
	{\WeakOpacityText{} (resp.\ \ExistentialOpacityText{}) \metaStrategy{} emptiness problem}
	{A TA $\TA$}
	{Decide the emptiness of the set of \metaStrategies{} $\metastrat$ such that $\TA$ is \weakOpaqueText{} (resp.\ \existentialOpaqueText{}) with~$\metastrat$.}

Similarly, we define as follows their synthesis counterpart:

\defProblem
	{\WeakOpacityText{} (resp.\ \ExistentialOpacityText{}) \metaStrategy{} synthesis problem}
	{A TA $\TA$}
	{Synthesize a \metaStrategy{} $\metastrat$ such that $\TA$ is \weakOpaqueText{} (resp.\ \existentialOpaqueText{}) with~$\metastrat$.}
\subsection{Results for \existentialOpacityText{}}\label{appendix:results-existential}

First, we focus on \existentialOpacityText{} for which control, as understood here, is useless.
Indeed, the strategy of the controller can only \emph{prevent} some behaviour, \ie{} remove possible executions.
However, \existentialOpacityText{} wonders whether there \emph{exists} an opaque time in the~TA, so adding a controller (that only removes execution times) cannot change the result.
The following theorem proves this claim.

\begin{theorem}
	Let $\TA$ be a \TAtext{}.
	$\TA$ is \existentialOpaqueText{} iff
	there exists a strategy $\strategy$ such that $\TA$ controlled by strategy~$\strategy$ is \existentialOpaqueText{}.
\end{theorem}
\begin{proof}
	\begin{itemize}
		\item[$\implies$]  Assume that $\TA$ is \existentialOpaqueText{}.
		Therefore, a strategy enabling all controllable actions at all times does not restrict the behaviour, and therefore the controlled TA remains \existentialOpaqueText{}.
		Concretely, let $\strategy$ be such that $\forall \tau \in \setRgeqzero : \strategy(\tau) = \ActionSetC$.
		Then, $\PrivDurVisitStrat{\TA}{\strategy} = \PrivDurVisit{\TA}$ and $\PubDurVisitStrat{\TA}{\strategy} = \PubDurVisit{\TA}$, and hence because $\TA$ is \existentialOpaqueText{}, $\PrivDurVisit{\TA} \cap \PubDurVisit{\TA} \neq \emptyset \implies \PrivDurVisitStrat{\TA}{\strategy} \cap \PubDurVisitStrat{\TA}{\strategy} \neq \emptyset$, hence the $\TA$ controlled by strategy~$\strategy$ is \existentialOpaqueText{}.

		\item[$\impliedby$]
		Let $\TA$ be a \TAtext{} and $\strategy$ a strategy such that $\TA$ controlled by strategy~$\strategy$ is \existentialOpaqueText{}.
		From \cref{def:opacity-strategy-weak}, we have $\PrivDurVisitStrat{\TA}{\strategy} \cap \PubDurVisitStrat{\TA}{\strategy} \neq \emptyset$.
		That is, there exists a duration $\paramd$ such that there exist a private run~$\run$ and a public run~$\run'$, both of duration~$\paramd$.
		Let $\run = (\locinit, \ClocksZero), (\timeEdge{0}), \cdots, (\timeEdge{{n-1}}), (\loci{n}, \clockval)$.
		For all $0 \leq i < n$, $e_i = (\loc_i, \guard_i, \action_i, \resets_i, \loc_i')$,
		with $\action_i \in \strategy(\sum_{j = 0}^{j \leq i}d_j) \cup \ActionSetU$.
		As $ \strategy(\sum_{j = 0}^{j \leq i}d_j) \cup \ActionSetU \subseteq \ActionSetC \cup \ActionSetU = \ActionSet $, the action $\action_i$ is available when $\TA$ is not controlled and $\run \in \PrivVisit{\TA}$.
		With the same reasoning for~$\run'$, we have that $\run' \in \PubVisit{\TA}$---and therefore $\TA$ is \existentialOpaqueText{}.
	\end{itemize}
\end{proof}

The same reasoning can apply to \nonFinitelyvarying{} strategies, as well as to \metaStrategies{}:

\begin{corollary}
	Let $\TA$ be a \TAtext{}.
	\begin{itemize}
		\item $\TA$ is \existentialOpaqueText{} iff
	there exists a \nonFinitelyvarying{} strategy $\strategy$ such that $\TA$ controlled by strategy~$\strategy$ is \existentialOpaqueText{}.
		\item $\TA$ is \existentialOpaqueText{} iff
	there exists a \metaStrategy{}~$\metastrat$ such that $\TA$ is \existentialOpaqueText{} with \metaStrategy~$\metastrat$.
	\end{itemize}
\end{corollary}
\subsection{Results for \weakOpacityText{}}\label{appendix:result-weak}

We now address \weakOpacityText{} which can be derived from our analysis of  \fullOpacityText{}.
We adapt to \weakOpacityText{} the concept of \badBelief{} belief from \cref{definition:leaking}.

\begin{definition}[\BadBelief{} belief for \weakOpacityText{}]
	A belief $\belief{}$ is \emph{\badBelief{} for \weakOpacityText} when
	\begin{itemize}
		\item $(\belief{} \cap \finalclass{\TA} \cap \secret{\TA} \neq \emptyset)$, and
		\item $(\belief{} \cap \finalclass{\TA} \cap \notsecret{\TA} = \emptyset)$.
	\end{itemize}
\end{definition}

We now adapt \cref{theorem:opacity-leaking-belief} to \weakOpacityText{}.

\begin{lemma}[Beliefs characterization for \weakOpacityText{}]\label{lem:beliefs-weak-opacity}
	A TA $\TA$ is \weakOpaqueText{} with a \metaStrategy{} $\metastrat$ whenever, for all $\belief{} \in \beliefcontrolset{\TA}{\metastrat}$, $\belief{}$ is not \badBelief{} for \weakOpacityText{}.
\end{lemma}
\begin{proof}
	This proof can be achieved similarly to the proof of \cref{theorem:opacity-leaking-belief}. %
 \end{proof}
\begin{theorem}\label{theorem:weak-problems}
	The \weakOpacityText{} \metaStrategy{} emptiness problem is decidable in \EXPSPACE; and the \weakOpacityText{} \metaStrategy{} synthesis problem is solvable in \twoEXPTIME.
\end{theorem}
\begin{proof}
	This proof can be achieved similarly to the proof of \cref{theorem:full-problems}, using \cref{lem:beliefs-weak-opacity} in place of \cref{theorem:opacity-leaking-belief}. 
\end{proof}
\section{Extension: robust definitions of ET-opacity}\label{section:robust-extensions}

So far, the attacker needed to measure the execution time with an infinite precision---this is often unrealistic in practice~\cite{DDMR04,BMS13}.
We therefore consider \mbox{variants} of opacity where intervals of non-opaque execution times can be considered acceptable as long as they are hard to detect, for instance by being of size~0, \ie{} reduced to a point (note that there can be an infinite number of them).

In order to formally define these new notions, we introduce new notations:
given a set~$S$, then let $\closed{S}$ denote the \emph{closure} of~$S$ (\ie{} the smallest closed set containing~$S$)
and let $\open{S}$ denote the \emph{interior} of~$S$ (\ie{} the largest open set contained in~$S$).
Let $\oplus$ denotes the exclusive \texttt{OR} operator such that, for two sets $A$ and $B$, $A \oplus B = \set{v \mid v \in (A \cup B) \setminus (A \cap B)}$.

We introduce two new notions of opacity:
\begin{oneenumerate}%
	\item \emph{\almostFullOpacityText{}}, where every punctual opacity violation is ignored, and
	\item \emph{\closedFullOpacityText{}}, where a punctual violation is ignored only if it is followed or preceded by an opaque interval.
\end{oneenumerate}%
\subsection{Definitions and problems}
\subsubsection{\almostFullOpacityText{}}

Let us first define \emph{\almostFullOpacityText{}}, where every punctual opacity violation is ignored.
That is, a \TAtext{} is \almostFullOpaqueText{} whenever
all non-opaque durations are isolated from each other;
that is, since these non-opaque durations must be punctual, then the interior of intervals of non-opaque durations must be empty.
\begin{definition}[\AlmostFullOpacityText{}]\label{def:almost-full-opacity}
	A \TAtext{} \TA{} is \almostFullOpaqueText{} when
	$\open{\PrivDurVisit{\TA} \oplus \PubDurVisit{\TA}} = \emptyset$.
\end{definition}
\subsubsection{\closedFullOpacityText{}}

Let us now define \emph{\closedFullOpacityText{}}, where a punctual violation is ignored only if it is followed or preceded by an opaque interval.
That is, we say a \TAtext{} is \closedFullOpaqueText{} when the closure of the private durations equals the closure of the public durations.
\begin{definition}[\ClosedFullOpacityText{}]\label{def:closed-full-opacity}
	A \TAtext{} \TA{} is \closedFullOpaqueText{} when
	$\closed{\PrivDurVisit{\TA}} = \closed{\PubDurVisit{\TA}}$.
\end{definition}

A difference between both definitions is, for example, whenever a non-opaque duration is such that the immediately neighbouring durations do not correspond to any accepting run (neither public nor private).
In that case, this non-opaque duration will be left out by \almostFullOpacityText{}, but will still be considered non-opaque by \closedFullOpacityText{}.

We can consider \AlmostFullOpacityText{} and \ClosedFullOpacityText{} with a 
\metaStrategy{} $\metastrat$ by replacing 
$\PrivDurVisit{\TA}$ by $\PrivDurVisitStrat{\TA}{\metastrat}$
and $\PubDurVisit{\TA}$ by $\PubDurVisitStrat{\TA}{\metastrat}$.

\newcommand{\Aquatre}{\ensuremath{\TAi{3}}}
\subsubsection{Example}
\begin{example}\label{example:almost-opacity}
	Let $\TAi{2}$ be a \TAtext{} such that $\PrivDurVisit{\TAi{2}} = [0,2]$ and $\PubDurVisit{\TAi{2}} = (0,1) \cup (1,2)$. $\TAi{2}$ is not \fullOpaqueText{} (but it is \weakOpaqueText{}).
	But	$\open{\PrivDurVisit{\TAi{2}} \oplus \PubDurVisit{\TAi{2}}} = \open{\set{0} \cup \set{1} \cup \set{2}} = \emptyset$, so $\TAi{2}$ is \almostFullOpaqueText{}.
	Note that look at the interior of private and public interval would not be equivalent: $\open{\PrivDurVisit{\TAi{2}}} \neq \open{\PubDurVisit{\TAi{2}}}$ as $(0,2) \neq ((0,1) \cup (1,2))$.
	Moreover, $\closed{\PrivDurVisit{\TAi{2}}} = \closed{\PubDurVisit{\TAi{2}}} = [0,2]$ so $\TAi{2}$ is \closedFullOpaqueText{}.

	Now, let $\Aquatre$ be a \TAtext{} such that $\PrivDurVisit{\Aquatre} = (0,1) \cup \set{2}$ and $\PubDurVisit{\Aquatre} = (0,1)$. $\Aquatre$ is not \fullOpaqueText{}.
	$\open{\PrivDurVisit{\Aquatre} \oplus \PubDurVisit{\Aquatre}} = \emptyset$ so $\Aquatre$ is \almostFullOpaqueText{}.
	$\closed{\PrivDurVisit{\Aquatre}} = [0,1] \cup \set{2} \neq \closed{\PubDurVisit{\Aquatre}} = [0,1]$ so $\Aquatre$ is not \closedFullOpaqueText{}.
\end{example}

\begin{figure}[tb]
	\begin{center}
	   \begin{subfigure}{0.22\textwidth}
		\begin{tikzpicture}[scale=.8, every node/.style={scale=0.8}]
			\node[below] (a) at (0,0) {0};
			\node[below] (b) at (1,0) {1};
			\node[below] (c) at (2,0) {2};
			\node[below] (d) at (3,0) {3};

			\draw[->] (0,0) -- (4,0);
			\draw (0,.05) -- (0,-.05);
			\draw (1,.05) -- (1,-.05);
			\draw (2,.05) -- (2,-.05);
			\draw (3,.05) -- (3,-.05);

			\draw[very thick,red] (0,.4) -- (2,.4);
			\draw[red] (.05,.5) -- (0,.5) -- (0,.3) -- (.05, .3);
			\draw[red] (1.95,.5) -- (2,.5) -- (2,.3) -- (1.95, .3);

			\draw[very thick,blue] (0.05,-.5) -- (.95,-.5);
			\draw[blue] (0,-.6) -- (0.05,-.6) -- (0.05,-.4) -- (0, -.4);
			\draw[blue] (.99,-.6) -- (.95,-.6) -- (.95,-.4) -- (.99, -.4);

			\draw[very thick,blue] (1.05,-.5) -- (1.95,-.5);
			\draw[blue] (1.01,-.6) -- (1.05,-.6) -- (1.05,-.4) -- (1.01, -.4);
			\draw[blue] (2,-.6) -- (1.95,-.6) -- (1.95,-.4) -- (2, -.4);
			\end{tikzpicture}
		   \caption{}
		   \label{fig:time-bar-closed}
	   \end{subfigure}
	   \begin{subfigure}{.22\textwidth}
		\begin{tikzpicture}[scale=.8, every node/.style={scale=0.8}]
		\node[below] (a) at (0,0) {0};
		\node[below] (b) at (1,0) {1};
		\node[below] (c) at (2,0) {2};
		\node[below] (d) at (3,0) {3};

		\draw[->] (0,0) -- (4,0);
		\draw (0,.05) -- (0,-.05);
		\draw (1,.05) -- (1,-.05);
		\draw (2,.05) -- (2,-.05);
		\draw (3,.05) -- (3,-.05);

		\draw[very thick,red] (0,.4) -- (1,.4);
		\draw[red] (-.05,.5) -- (0,.5) -- (0,.3) -- (-.05, .3);
		\draw[red] (1.05,.5) -- (1,.5) -- (1,.3) -- (1.05, .3);
		\draw[thick,red] (2, .5) -- (2, .3);
		\draw[red] (1.9,.5) -- (2.1, .5);
		\draw[red] (1.9,.3) -- (2.1, .3);

		\draw[very thick,blue] (0,-.5) -- (1,-.5);
		\draw[blue] (-.05,-.6) -- (0,-.6) -- (0,-.4) -- (-.05, -.4);
		\draw[blue] (1.05,-.6) -- (1,-.6) -- (1,-.4) -- (1.05, -.4);
		\end{tikzpicture}
	\label{fig:time-bar-almost}
	\caption{}
	\end{subfigure}
	\caption{Private (above in red) and public (below in blue) durations in $\TAi{2}$ (a) and \Aquatre{} (b)}
	\label{fig:time-bar}
\end{center}
\end{figure}

\subsubsection{Problems}

We are interested in the same problems as before, this time in the context of \closedOpacityText{} or \almostOpacityText{}.
Formally:

\defProblem
	{\closedFullOpacityText{} (resp.\ \almostFullOpacityText{}) \metaStrategy{} emptiness problem}
	{A TA $\TA$}
	{Decide the emptiness of the set of \metaStrategies{} $\metastrat$ such that $\TA$ is \closedFullOpaqueText{} (resp.\ \almostFullOpaqueText{}) with~$\metastrat$.}
\defProblem
	{\closedFullOpacityText{} (resp.\ \almostFullOpacityText{}) \metaStrategy{} synthesis problem}
	{A TA $\TA$}
	{Synthesize a \metaStrategy{} $\metastrat$ such that $\TA$ is \closedFullOpaqueText{} (resp.\ \almostFullOpaqueText{}) with~$\metastrat$.}
\subsection{Characterization}

A single belief is not sufficient to characterize a \TAtext{} that is not \almostOpaqueText{} (resp.\ closed).
Indeed, suppose a time~$t$ such that $\belief{t}$ is \badBelief{} for \fullOpacityText{}.
This means that a punctual violation of opacity exists.
This kind of violation can be allowed in the context of almost and \closedFullOpacityText{}.
It is problematic if the times around it are also a violation of opacity.

More specifically, a violation to \almostFullOpacityText{} corresponds to a succession of \badBelief{} beliefs, \ie{} every punctual violation is ignored.
On the other hand, a violation to \closedFullOpacityText{} corresponds either to a succession of \badBelief{} beliefs, or to a unique \badBelief{} belief surrounded by beliefs that do not contain any final region.
Intuitively, a punctual violation is ignored if it belongs to an interval where private and public final states can be reached.

When considering \metaStrategies{}, this issue is partially lifted as the behaviour of the system within an open interval is the same:
given $k\in\setN$, and $t,t'\in (k,k+1)$, $\beliefcontrol{t}{\metastrat}=\beliefcontrol{t'}{\metastrat}$.
This a consequence of the shrinking argument within
the proof of \cref{theorem:opacity-leaking-belief}. As a consequence, only a belief representing an interval can be \badBelief{} for \almostFullOpacityText{}.

We define formally \badBelief{} belief for \almostFullOpacityText{} and \closedFullOpacityText{} for a \metaStrategy{} as follows.
\begin{definition}[\BadBelief{} belief for \almostFullOpacityText{}]\label{def:leaking-belief-almost-full}
Let $\metastrat$ a \metaStrategy{}, a belief $\belief{} \in \beliefcontrolset{\TA}{\metastrat}$ is \badBelief{} for \almostFullOpacityText{} when it is \badBelief{} for \fullOpacityText{} and there is $k \in \setN$ such that $\belief{} = \belief{k+}$.  
\end{definition}
\begin{definition}[\BadBelief{} belief for \closedFullOpacityText{}]\label{def:leaking-belief-closed-full}
Let $\metastrat$ a \metaStrategy{}, a belief $\belief{} \in \beliefcontrolset{\TA}{\metastrat}$ is \badBelief{} for \closedFullOpacityText{} when either: 
\begin{itemize}
	\item it is \badBelief{} for \almostOpacityText{}, or
	\item there is $k \in \setN$ such that $\belief{} = \belief{k}$, and 
	\begin{itemize}
		\item $\belief{(k-1)+} \cup \finalclass{\TA} = \emptyset$,
		\item $\belief{(k+1)+} \cup \finalclass{\TA} = \emptyset$ and
		\item \belief{} is \badBelief{} for \fullOpacityText{}
	\end{itemize} 
\end{itemize}
\end{definition}

As previously, the \almostFullOpacityText{} and \closedFullOpacityText{} of a TA can be characterized thanks to the interval beliefs.
Both lemma can be achieved with a proof similar to the proof of \cref{theorem:opacity-leaking-belief}.

\begin{lemma}[Beliefs characterization for \almostFullOpacityText{}]\label{lem:beliefs-almost-full-opacity}
	A TA $\TA$ is \almostFullOpaqueText{} with a \metaStrategy{} $\metastrat$ iff, for all $\belief{} \in \beliefcontrolset{\TA}{\metastrat}$, $\belief{}$ is not \badBelief{} for \almostFullOpacityText{}.
\end{lemma}
\begin{lemma}[Beliefs characterization for \closedFullOpacityText{}]\label{lem:beliefs-closed-full-opacity}
	A TA $\TA$ is \closedFullOpaqueText{} with a \metaStrategy{} $\metastrat$ iff, for all $\belief{} \in \beliefcontrolset{\TA}{\metastrat}$, $\belief{}$ is not \badBelief{} for \closedFullOpacityText{}.
\end{lemma}
\subsection{Results}

We can conclude positively for the problems on \closedFullOpacityText{}
and \almostFullOpacityText{}.

As for \fullOpacityText{}, this result is due to the equivalence between finding a \metaStrategy{} for a \TAtext{} and finding a strategy in a variation of the corresponding \beliefAutomaton{}, and to the fact that such a strategy corresponds to a winning strategy in a one-player Büchi game.
Hence, both following result can be achieved similarly to the proof of \cref{theorem:full-problems}, using \cref{lem:beliefs-closed-full-opacity} in place of \cref{theorem:opacity-leaking-belief}. 
Note though that the game that one needs to build when considering 
\closedFullOpacityText{} is slightly expanded compared to \cref{theorem:full-problems}:
a Boolean must be included to note whether a singleton is violating opacity and the previous
interval belief did not contain a region associated to a final location. 
This information strengthens the constraint on the next interval belief, requiring that it
must contain a region associated to a final location.

\begin{theorem}\label{theorem:almost-full}
	The \almostFullOpacityText{} \finitelyvarying{} controller emptiness problem is decidable; the \almostFullOpacityText{} \finitelyvarying{} controller synthesis problem is solvable.
\end{theorem}
\begin{theorem}\label{theorem:closed-problem}
	The \closedFullOpacityText{} \metaStrategy{} controller emptiness problem is decidable;  the \closedFullOpacityText{} \metaStrategy{} controller synthesis problem is solvable.
\end{theorem}

\begin{remark}
	We can extend these notions of almost and \closedFullOpacityText{} to their weak counterparts, with similar results.
\end{remark}

\section{Conclusion}\label{section:conclusion}
We addressed here the control of a system modelled by a TA to make it \fullOpaqueText{} (execution-time opaque).
On the one hand, we showed that the strategy emptiness problem is undecidable.
On the other hand, we showed that not only the strategy emptiness problem becomes decidable when considering \metaStrategies{} (\ie{} in which we specify the order of---a finite number of---strategy changes within interval time units, without fixing their actual changing time),
but also we can effectively solve the controller synthesis problem, by building such a controller.

In addition, we studied two other versions of opacity from~\cite{ALLMS23}, namely
	\existentialOpacityText{} (in which we are simply interested in the \emph{existence} of one execution time for which opacity is ensured),
	and
	\weakOpacityText{} (in which it is harmless that the attacker deduces that the private location was \emph{not} visited).

\sectionEight{%
We also addressed two extensions (\closedFullOpacityText{} and \almostFullOpacityText{}) which can relate to a \emph{robust} setting where the attacker cannot have an infinite precision.
}

\paragraph{Future works}

A natural next step will be to introduce timing parameters \emph{à la}~\cite{ALLMS23}, and address control in that setting.

Addressing the control for the definition of opacity (based on languages) as in~\cite{Cassez09} would be interesting in two settings:
\begin{oneenumerate}%
	\item the general setting, where the controller synthesis will be undecidable but may terminate for some semi-algorithms, and
	\item decidable subclasses that remain to be exhibited, presumably one-clock TAs, as in~\cite{ADL24}.
\end{oneenumerate}%
Moreover, an analysis of the strategy obtained to ensure opacity might lead to only a static modification of the structure (\eg{} deletion of a transition)---which will be interesting to study.

Finally, the implementation of this work is on our agenda.
While implementing the beliefs directly would be straightforward, it would probably result in an unnecessary blowup, and therefore an adaptation with structures such as zones~\cite{BBLP06} (which does not seem immediate) should be designed.

\subsection*{Acknowledgments}

	This work is partially supported by ANR BisoUS (ANR-22-CE48-0012) and by ANR TAPAS (ANR-24-CE25-5742).
	\newcommand{\CCIS}{Communications in Computer and Information Science}
	\newcommand{\ENTCS}{Electronic Notes in Theoretical Computer Science}
	\newcommand{\FAC}{Formal Aspects of Computing}
	\newcommand{\FundInf}{Fundamenta Informaticae}
	\newcommand{\FMSD}{Formal Methods in System Design}
	\newcommand{\IJFCS}{International Journal of Foundations of Computer Science}
	\newcommand{\IJSSE}{International Journal of Secure Software Engineering}
	\newcommand{\IPL}{Information Processing Letters}
	\newcommand{\JAIR}{Journal of Artificial Intelligence Research}
	\newcommand{\JLAP}{Journal of Logic and Algebraic Programming}
	\newcommand{\JLAMP}{Journal of Logical and Algebraic Methods in Programming} %
	\newcommand{\JLC}{Journal of Logic and Computation}
	\newcommand{\LMCS}{Logical Methods in Computer Science}
	\newcommand{\LNCS}{Lecture Notes in Computer Science}
	\newcommand{\RESS}{Reliability Engineering \& System Safety}
	\newcommand{\RTS}{Real-Time Systems}
	\newcommand{\SCP}{Science of Computer Programming}
	\newcommand{\SOSYM}{Software and Systems Modeling ({SoSyM})}
	\newcommand{\STTT}{International Journal on Software Tools for Technology Transfer}
	\newcommand{\TCS}{Theoretical Computer Science}
	\newcommand{\TOPLAS}{{ACM} Transactions on Programming Languages and Systems ({ToPLAS})}
	\newcommand{\ToPNoC}{Transactions on {P}etri Nets and Other Models of Concurrency}
	\newcommand{\TOSEM}{{ACM} Transactions on Software Engineering and Methodology ({ToSEM})}
	\newcommand{\TSE}{{IEEE} Transactions on Software Engineering}

\newpage
	\bibliographystyle{alphaurl} %
	\bibliography{biblio}

\newpage
\appendix

\section{Table of notations}\label{appendix:notations}
	\begin{tabular}{p{.075\textwidth} p{.35\textwidth}}
	\rowcolor{USPNsable}
	\multicolumn{2}{l}{
		\bfseries{}Timed automaton
	}\\
	\hline
	\TA & a timed automaton\\
	$\ActionSet$ & a finite set of actions of a TA\\
	$\LocSet$ & a finite set of locations of a TA\\
	$\locinit$ & the initial location of a TA\\
	$\locpriv$ & the private location of a TA\\
	$\LocFinalSet$ & the set of final locations of a TA \\
	$\ClockSet$ & a finite set of clocks \\
	$\clocki{i}$ & the $i^{th}$ clock \\
	$\ClockCard$ & the number of clocks \\
	$\invariantof{\loc}$ & the invariant of location $\loc$ \\
	$\EdgeSet$& the finite set of edges of a TA \\
	$\edge$ & an edge \\
	$\resets$ & a set of clocks to be reset\\
	$\guard$& a guard\\
	$\clockval$ & a clock valuation\\
	$\paramd$ & a delay\\
	$\clockextra$ & extra clock\\
	$\constantmax{i}$ & largest constant for clock~$\clock_i$\\
	\hline
	\end{tabular}

\bigskip

\noindent
\begin{tabular}{p{.075\textwidth} p{.35\textwidth}}
\rowcolor{USPNsable}
\multicolumn{2}{l}{
		\bfseries{}Semantics
}\\
\hline
$\semantics{\TA}$ & the semantics of TA \TA \\
$\TTSStates$ & the set of states in $\semantics{\TA}$\\
$\TTSinit$ & the initial state in $\semantics{\TA}$\\
$\TTSstate$ & a state in $\semantics{\TA}$ \\
$\TTStransition$ & transition function of $\semantics{\TA}$ \\
$\transitionWith{\edge}$ & a discrete transition with edge $\edge$\\
$\transitionWith{\paramd}$ & a delay transition with delay $\paramd$\\
$\run$ & a run\\
$\laststate(\run)$ & last state of run $\run$\\
\hline
\end{tabular}

\bigskip

\noindent
\begin{tabular}{p{.075\textwidth} p{.35\textwidth}}
	\rowcolor{USPNsable}
	\multicolumn{2}{l}{
		\bfseries{}Regions
}\\
\hline
$\region$ & a region\\
$\class{\TTSstate}$ & equivalence class of $\TTSstate$\\
$\regset{\TA}$ & regions set of~$\TA$\\
$\finalclass{\TA}$ & set of final regions of \TA\\
$\regaut{\TA}$ & region automaton of~$\TA$\\
$\RegAutActions$ & set of actions of $\regaut{\TA}$\\
$\RegAutTransitions$ & transition function of $\regaut{\TA}$\\
$\RegAutTransitionWith{}$ & a transition in $\regaut{\TA}$ \\
\hline
\end{tabular}

\bigskip

\noindent
\begin{tabular}{p{.12\textwidth} p{.305\textwidth}}
	\rowcolor{USPNsable}
\multicolumn{2}{l}{
	\bfseries{}Opacity
}\\
\hline
$\PrivVisit{\TA}$ & set of private runs of \TA\\
$\PubVisit{\TA}$ & set of public runs of \TA\\
$\PrivDurVisit{\TA}$ & set of durations of private runs of \TA\\
$\PubDurVisit{\TA}$ & set of durations of public runs of \TA\\
$\PrivDurVisitStrat{\TA}{\strategy}$ & set of durations of private and  \compatible{} runs \\ 
$\PubDurVisitStrat{\TA}{\strategy}$ & set of durations of public and \compatible{}  runs \\ 
$\PrivDurVisitStrat{\TA}{\metastrat}$ & set of durations of private and  \compatible{} runs, $\strategy \satisfies \metastrat$ \\ 
$\PubDurVisitStrat{\TA}{\metastrat}$ & set of durations of public and \compatible{}  runs, such that $\strategy \satisfies \metastrat$ \\ 
\hline
\end{tabular}

\bigskip

\noindent
\begin{tabular}{p{.075\textwidth} p{.35\textwidth}}
	\rowcolor{USPNsable}
	\multicolumn{2}{l}{
		\bfseries{}Strategies
}\\
\hline
$\ActionSetC$ & controllable actions \\
$\ActionSetU$ & uncontrollable actions \\
$\strategy$ & a strategy \\
$\metastrat$ & a \metaStrategy{}\\
$\substrat$ & a strategy within an interval \\
$\subinterval$ & a subinterval \\ 
$\TTStransitionControl$ & transition function of the semantics of~$\TA$ controlled by strategy~$\strategy$\\
$\transitionWithControl{}$ & a discrete or delay transition in the semantics of~$\TA$ controlled by strategy~$\strategy$\\
$\xrightarrow{\paramd,e}_{\strategy}$ & a transition with delay $\paramd$ and edge $\edge$ in the semantics of~$\TA$ controlled by strategy~$\strategy$\\
$\activated$ & set of activated controllable actions \\
$\strategy \satisfies \metastrat$ & $\strategy$ satisfies $\metastrat$\\ 
\hline
\end{tabular}

\bigskip

\noindent
\begin{tabular}{p{.075\textwidth} p{.35\textwidth}}
	\rowcolor{USPNsable}
	\multicolumn{2}{l}{
		\bfseries{}Duplicated TA
}\\
\hline
$\TAdup$ & a duplicated TA\\
$\LocSet_{priv}$ & set of private states \\
$\LocSet_{pub}$ & set of public states \\
$\locpriv$& a private state \\
$\secret{\TA}$ & set of regions reachable on a run visiting $\locpriv$\\
$\notsecret{\TA}$ & set of regions not reachable on a run visiting $\locpriv$\\
\hline
\end{tabular}

\bigskip

\noindent
\begin{tabular}{p{.075\textwidth} p{.35\textwidth}}
	\rowcolor{USPNsable}
	\multicolumn{2}{l}{
		\bfseries{}Beliefs
}\\
\hline
$\belief{t}$ & belief for time $t$\\
$\belief{k+}$ & belief for the interval $(k,k+1)$\\  
$\translabelinstant, \translabelsame, \translabelchge, \translabelany$ & evolution of $\fract{\clockextra}$\\ 
$\BeliefAut{\TA}$ & \beliefAutomaton{} of~$\TA$\\
$\BeliefAutState{\TA}$ & states of $\BeliefAut{\TA}$\\
$\BeliefAutActions{\TA}$ & actions of $\BeliefAut{\TA}$\\
$\SeqTransitions$ & a sequence of actions of $\BeliefAut{\TA}$\\
$\run \admits \SeqTransitions$ & $\run$ admits $\SeqTransitions$\\ 
$\BeliefAutTransitions{\TA}$ & transition function of $\BeliefAut{\TA}$\\
$\beliefcontrolset{\TA}{\metastrat}$ & set of interval beliefs reachable by a \metaStrategy{} $\metastrat$ \\
$\BeliefControlAut{\TA}{\metastrat}$ & controlled \beliefAutomaton{} of~$\TA$ and \metaStrategy{} $\metastrat$\\
$\BeliefControlAutState{\TA}{\metastrat}$ & states of controlled \beliefAutomaton{} \\
$\BeliefControlAutActions{\TA}{\metastrat}$ & actions of controlled \beliefAutomaton{} \\
$\BeliefControlAutTransitions{\TA}{\metastrat}$ & transition function of controlled \beliefAutomaton{} \\
$\beliefencountset{\TA}{\metastrat}$ & set of beliefs encountered by $\BeliefControlAut{\TA}{\metastrat}$\\ 
\hline
\end{tabular}

\end{document}